\documentclass[a4paper]{article}
\usepackage[utf8]{inputenc}
\usepackage{amsmath}
\usepackage{amsfonts}
\usepackage{amssymb}
\usepackage{graphicx}
\usepackage{slashed}
\usepackage[english]{babel}
\usepackage{amsthm}
\usepackage{mathtools}
\usepackage[a4paper,top=3cm,bottom=3.5cm,left=3cm,right=3cm,marginparwidth=1.75cm]{geometry}
\usepackage{hyperref}

\title{Global stability of Minkowski spacetime for a spin-1/2 field}
\author{Xuantao Chen}
\date

\begin{document}

\maketitle
\begin{abstract}
We study the initial value problem of the Einstein-Dirac system, and show the stability of the Minkowski solution in the massless case with the use of generalized wave coordinates. This requires the understanding of the Dirac equation in curved spacetime, for which we establish various estimates. The proof is based on the vector-field method which is widely used in works on the stability of Minkowski problems for other Einstein-coupled systems. Under a specific choice of the tetrad, we show that components of the Dirac field satisfy a quasilinear wave equation, by resolving a potential loss of derivative problem. We also show that the semilinear nonlinearity of this equation behaves like a null form. In this way, we obtain the sharp decay of the field along the light cone. The structure of the energy-momentum tensor causes worse behavior of some components of the metric than the vaccum case, but an additional structure shows that there is no harm to the global existence result. In addition, we develop an estimate of the Dirac equation itself adpated to the decay of the metric, as this provides better estimates in the interior compared with the estimates from the second order equation. The combination of these estimates leads to a good control of the Dirac field that helps close the proof. We shall also see how our argument here gives a proof of the wellposedness of the system in the massive case.
\end{abstract}

\newtheorem{thm}{Theorem}[section]
\newtheorem{lem}[thm]{Lemma}
\newtheorem{Cor}[thm]{Corollary}
\newtheorem{Prop}[thm]{Proposition}
\newtheorem{Def}[thm]{Definition}
\theoremstyle{remark}
\newtheorem{remark}[thm]{Remark}

\numberwithin{equation}{section}

\def\y{\gamma}
\def\d{\partial}
\def\dd{\nabla}
\def\lg{{L_{Geo}}}
\def\lt{{\widetilde{L}}}
\def\ltb{{\underline{\lt}}}
\def\lb{{\underline{L}}}
\def\e{\widetilde{e}}
\def\Tt{\widetilde{T}}
\def\Ut{\widetilde{U}}
\def\Rt{\widetilde{R}}
\def\pt{\tilde{\psi}}
\def\ra{\rightarrow}
\def\dts{\d_{t^*}}
\def\drs{\d_{r^*}}
\def\gt{\widetilde{g}}
\def\ds{\d^*}
\def\zsi{Z^{*I}}
\def\zhsi{\hat{Z}^{*I}}
\def\zsj{Z^{*J}}
\def\zsk{Z^{*K}}
\def\dt{\widetilde{\d}}
\def\db{\overline{\d}}
\def\dtb{\overline{\widetilde{\d}}}
\def\gb{\bar{\gamma}}
\def\l{\mathcal{L}}
\def\lh{\widehat{\mathcal{L}}}
\def\go{\mathring{\gamma}}
\def\gbar{\overline{g}}
\def\z{\widetilde{Z}}
\def\zh{\hat\widetilde{Z}}
\def\x{\widetilde{x}}
\def\wb{\overline{w}}
\def\lz{\l_{\z}}
\def\lzr{\l_Z}
\def\fp{F^\psi}
\def\mh{\widehat{m}}
\def\Hi{\widetilde{H}_1}
\def\hh{\widehat{h}}
\def\Hh{\widehat{H}}
\def\hi{\widetilde{h}^1}
\def\Hh{\widehat{H}}
\def\wp{w^\psi}
\def\mt{\widetilde{m}}
\def\cht{\widetilde{\chi}}
\def\a{\alpha}
\def\b{\beta}
\def\k{\kappa}
\def\pb{\overline{\psi}}
\def\lzh{{\widehat{\mathcal{L}}_{\z}}}
\def\lzp{(\lz^\psi)}
\def\lhzp{(\widehat{\mathcal{L}}_{\z}^\psi)}
\def\et{\widetilde{e}}
\def\gat{\widetilde{\Gamma}}
\def\gh{\widehat{\Gamma}}
\def\dh{\widehat{\dd}}
\def\hti{\widetilde{h}}
\def\Hti{\widetilde{H}}
\def\Tc{\widetilde{\mathcal{T}}}
\def\lhhz{\widehat{\widehat{\l}}_{\z}}
\def\lc{\widetilde{\mathcal{L}}}
\def\Uc{\widetilde{\mathcal{U}}}
\def\omegat{\widetilde{\omega}}
\def\mn{{\mu\nu}}

\section{Introduction}
We study the solutions of the Einstein-Dirac system around the Minkowski solution, mainly the global aspect of the massless case\footnote{We use the Einstein summation convention. The Greek indices $\a,\b,\mu,\nu,\cdots=0,1,2,3$, and Latin indices $i,j,k,\ell,\cdots=1,2,3$.}
$$R_{\mu\nu}-\frac 12 Rg_{\mu\nu}=T_{\mu\nu},$$
$$\y^\mu D_\mu \psi 
=0.$$ 
The system involves a dynamic Lorentzian metric $g$ and its interaction with a Dirac spinor field $\psi$ taking value in $\mathbb{C}^4$. The system here is coordinate covariant, and the Dirac equation is also locally Lorentz covariant. The $4\times 4$ Gamma matrices $\gamma^\mu$ and the covariant derivative $D_\mu$ are defined to be compatible with these covariances.



The system consists of equations of motion of the metric (or equivalently, the tetrad $\{e^\mu{}_{a}\}$, which is a field of orthonormal frame with respect to $g$) field and the spinor field, of the action
$$S[\psi,e_{\ a}^\mu]=S_\text{grav}[e_{\ a}^\mu]+S_{M}[\psi,e_{\ a}^\mu]=\int R\, |\det(e_{\ \mu}^a)|\, dx+\int  \,   \frac i2 (\overline{\psi} \gamma^\mu D_\mu \psi-D_\mu \overline{\psi} \gamma^\mu \psi)\, |\det(e_{\ \mu}^a)|\, dx.$$


The energy-momentum tensor is determined by $T^a{}_{\mu}=\displaystyle\frac 1{|\det(e_{\ \mu}^a)|} \frac {\delta S_M}{\delta e^\mu_{\ a}}$ and we have the formula for it \textit{on-shell}:
\begin{equation}\label{EMT}
T_{\mu\nu}=\frac i4 (\overline{\psi}\gamma_\mu D_\nu\psi-D_\nu\overline{ \psi} \gamma_\mu \psi)+\frac i4 (\overline{\psi}\gamma_\nu D_\mu\psi-D_\mu\overline{ \psi} \gamma_\nu \psi),
\end{equation}
which can be shown to be covariantly constant, i.e.\ $\dd^\mu T_{\mu\nu}=0$, where $\dd$ is the standard Levi-Civita connection.

The Einstein equation is coordinate covariant and one can impose coordinate conditions. In wave (harmonic) coordinates ($\Box_g x^\mu=0$), the equation can be written as a system of quasilinear wave equations:
\begin{equation*}
    \widetilde{\Box}_g g_{\mu\nu}=F_{\mu\nu}(g)(\d g,\d g)+T_{\mu\nu}-\frac 12 g_\mn \mathrm{tr}_g T,
\end{equation*}
where $\widetilde{\Box}_g=g^{\a\b}\d_\a\d_\b$ is the reduced wave operator. In the vaccum case, i.e.\ $T_\mn=0$, the wellposedness of the Einstein's equation then becomes a local wellposedness problem of quasilinear wave equations, as is proved in \cite{Choques-Bruhat1952}.

The local wellposedness result, however, does not provide insight on long-time behavior of the solution. The work by Christodoulou-Klainerman \cite{christodoulou1993global} proved the global stability of Minkowski spacetime for the Einstein vacuum equation without using coordinate gauge, as it was believed that the wave coordinates would lead to blow up. Nevertheless, Lindblad-Rodnianski \cite{03,04} managed to show the stability in wave coordinates, with a much simplified proof based on the vector-field method introduced by Klainerman \cite{klainerman1985uniform}, and gave the notion of weak null condition for nonlinear wave equations. There are many generalizations of the method to Einstein's equation coupled with different fields, such as Einstein-Maxwell \cite{loizelet2009solutions} (and more generalized model \cite{speck2014global}), Einstein-Vlasov \cite{Stability2017Fajman,EV} (see also \cite{TaylormasslessEV,masslessEVnoncompact} for massless case), Einstein-Klein-Gordon equation \cite{LMCMP} (see also a more frequency space approach for Einstein-Klein-Gordon \cite{ionescu2019einstein}). 

\vspace{0.5ex}

The approach in \cite{04} is based on the foliation in Minkowski spacetime. This is a very convenient choice of background geometry, and in fact, an important reason that the proof in \cite{04} is significantly simpler than in \cite{christodoulou1993global}. However, the estimates we can get from it are sometimes not optimal. 
This is in part because the Schwarzschild part with tail $Mr^{-1}$ causes a change in the asymptotics of the light cone from the Minkowski one. In \cite{1606Einsteinasymptotics}, Lindblad derived the asymptotics of the solution of Einstein's equations in wave coordinates by considering the change of coordinates $r^*=r+M\ln r$, where $M$ is the ADM mass of the initial data, and proved that $u^*=t-r^*$ converges to the true eikonal function of the metric. In these new coordinates, Kauffman \cite{kauffman2018} studied the fractional Morawetz estimates for massless Maxwell-Klein-Gordon equations in a fixed curved background, generalizing the estimate in \cite{lindblad04MKG}, which studied the problem in Minkowski spacetime. Then Kauffman-Lindblad \cite{KLMKG21} proved the global stability of the Minkowski solution for the system coupled with Einstein's equation.

\vspace{1ex}

The Dirac equation in Minkowski spacetime, $\gb^\mu \d_\mu \psi+im\psi=0$, is originally derived to be a relativistic generalization of Schr\"{o}dinger equation. The equation, as a first-order PDE, is remarkably Lorentz covariant, and now understood as a relativistic description of a spin-1/2 field. One can study Dirac equations coupled with other relativistic equations in Minkowski spacetime, and \cite{bachelot1988DiracVectorField} gave a generalization of the vector-field method to the Dirac equation, which helps the study of long-time problems. 

\vspace{0.5ex}
The Einstein-Dirac coupled system, which is of great interest in physics, requires the understanding of the Dirac spinor in curved spacetime. There have been works on special solutions of the system (see for instance \cite{finster1999particlelike}, \cite{pechenick1979new}), but it is less understood in terms of the initial value problem. In fact, even the wellposedness of the initial value problem has not been widely studied. In this work, we study the solution of the initial value problem, mainly the global behavior of solutions for massless Einstein-Dirac system. We show for initial data close to Minkowski spacetime in our sense, the solution is global and asymptotically converging to the Minkowski solution\footnote{From our result here, one can show that the metric from the solution is causally geodesically complete, by a modification of the proof in \cite{03} adapted to the decay of the metric in our case.}, by developing various estimates of the Dirac equation in curved spacetime in harmony with the decay of the metric in our situation. To the author's knowledge, these estimates in curved setting have not been derived before. Some observations here also give a proof of the (local) wellposedness of the Einstein-Dirac system, possibly with a mass $m\neq 0$. We shall state our main result and discuss the idea of proof in this introduction.

\subsection{Main result}
\subsubsection{Initial value problem}
For the initial value problem of the Einstein equation, one needs to start with a 3-dim Riemannian manifold $(\Sigma,\gbar)$ and a second fundamental form $k_{ij}$, and the solution is a 4-dim Lorentzian manifold $(\mathcal{M},g)$ with an embedding $\Sigma\subset \mathcal{M}$, with $\gbar$ being the restriction of $g$ on $\Sigma$ and $k$ the second fundamental form of $\Sigma$ in $\mathcal{M}$.

For the Einstein-Dirac system the initial data is a set $(\overline{g}_{ij},k_{ij},\psi_0)$. We define the initial value on $\{t=0\}$ compatible with the wave coordinate condition:
$$\begin{aligned}
&\left.g_{i j}\right|_{t=0}=\gbar_{ i j},\left.\quad g_{00}\right|_{t=0}=-a^{2},\left.\quad g_{0 i}\right|_{t=0}=0,\left.\quad \psi\right|_{t=0}=\psi_{0} \\
&\left.\partial_{t} g_{i j}\right|_{t=0}=-2 a k_{ i j},\left.\quad \partial_{t} g_{00}\right|_{t=0}=2 a^{3} \gbar^{i j} k_{ i j}, \\
&\partial_{t} g_{0 \ell}|_{t=0}=a^{2} \gbar^{i j} \partial_{j} \gbar_{ i \ell}-\frac{1}{2} a^{2} \gbar^{i j} \partial_{\ell} \gbar_{ij}-a \partial_{\ell} a.
\end{aligned}
$$
Here $a^2=1-\frac Mr \chi(r)$, where $\chi$ is a smooth function taking value between 0 and 1, with $\chi(s)=1$ when $s>3/4$, and $\chi(s)=0$ when $s<1/2$. By the construction the wave coordinate condition\footnote{It is straightforward to verify that this is equivalent with the condition $\Box_g x^\mu=0$.}
\begin{equation}\label{wcc}
    g^{\a\b}\d_\b g_{\a\mu}=\frac 12 g^{\a\b}\d_\mu g_{\a\b}
\end{equation}holds on initial hypersurface. Under this condition, the Einstein equation can be written as a system of quasilinear wave equations:
\begin{equation}
    \widetilde{\Box}_g g_{\mu\nu}=F_{\mn}(g)(\d g, \d g)+T_\mn-\frac 12 \mathrm{tr}_g T g_\mn.
\end{equation}
We call this the reduced Einstein equation. Furthermore, the wave coordinate condition can be shown to hold for the solution of the new equations (see e.g.\  \cite{choquet2008general} for proof).

The initial data $(\gbar_{ij},k_{ij},\psi_0)$ needs to satisfy constraint equations:
\begin{equation}\label{constraintTnn}
    \overline{R}+\gbar^{ij}\gbar^{\ell m}(k_{ij}k_{\ell m}-k_{i\ell}k_{jm})=T_\mn n^\mu n^\nu,
\end{equation}
\begin{equation}\label{constraintTni}
    \overline{\dd}_i (\gbar^{j\ell}k_{j \ell})-\gbar^{j\ell}\overline{\dd}_\ell k_{ij}=T_{\mu i}n^\mu,
\end{equation}   
where $\overline\dd$ is the Levi-Civita connection on $(\Sigma,\gbar_{ij})$, $T$ is the energy momentum tensor, and $n$ is the future-oriented unit normal vector to $\Sigma$. Note that the right hand side can be determined only using $\gbar_{ij}$, $k_{ij}$ and $\psi_0$, as we will see after formulating the Dirac equation in curved spacetime.

The proof of the local wellposedness of the Einstein vacuum equation in \cite{Choques-Bruhat1952} becomes a problem of quasilinear wave equation when the wave coordinate condition is given. We shall use a similar way to get the local wellposedness of the Einstein-Dirac equation, but in order to write the Dirac equation explicitly, we need to pick a choice of the tetrad, in addition to the coordinate condition. We will discuss this in Section \ref{chaptersecondorderDiracderivation} and \ref{constructionoftetradsection}.


\subsubsection{Main result}
Suppose that $\Sigma$ is diffeomorphic to $\mathbb{R}^3$, and there is a global coordinate chart on $\Sigma$ such that $\gbar$ is close to the Euclidean metric. Consider the initial data $\overline h_{ij}=\gbar_{ij}-\delta_{ij}$ satisfies the asymptotic expansion:
$$\overline h_{ij}=\frac Mr \delta_{ij}+o(r^{-1-\gamma}),\ \ k_{ij}=o(r^{-2-\gamma}),\ \ \text{as }r\ra\infty$$
for some $M>0$ and $0<\gamma<1$. 
We decompose the metric to separate the Schwarzschild part\footnote{We use $m$ to denote both the Minkowski metric and the mass in Dirac equation, as it is not hard to distinguish them.}
\begin{equation}\label{h0}
    g_\mn=m_\mn+h_\mn=m_\mn+h_\mn^0+h_\mn^1, \ \ \textrm{where } h^0_{\mu\nu}=\chi(\textstyle\frac r{t+1})\frac Mr \delta_{\mu\nu},
\end{equation}
where $\chi$ is the same as in last subsection, satisfying $\chi(s)=1$ when $s>3/4$, and $\chi(s)=0$ when $s<1/2$.
We also use the decomposition for the inverse metric \begin{equation}\label{H0}
    g^\mn=m^\mn+H^\mn=m^\mn+H_0^\mn+H_1^\mn, \ \ \textrm{where } H_0^{\mu\nu}=\chi(\textstyle\frac r{t+1})\frac Mr \delta^{\mu\nu}.
\end{equation}
For initial data, we also do the decomposition $\gbar=\delta+\overline h^0+\overline h^1$, where $\overline h_{ij}^0=\chi(r)\frac Mr \delta_{ij}$. We define the weighted norm for the initial data $(\gbar_{ij},k_{ij},\psi_0)$:
\begin{equation}
    \mathcal{E}_N(0)=\sum_{|I|\leq N}\left(||(1+r)^{\frac 12+\gamma+|I|} \overline{\dd}\, \overline{\dd}^I \overline h^1||_{L^2}^2+||(1+r)^{\frac 12+\gamma+|I|} \overline{\dd}^I k||_{L^2}^2\right),
\end{equation}
\begin{equation}
    \mathcal{E}_N^1(0)=\sum_{|I|\leq N}||(1+r)^{1+s+|I|} \overline{\dd}\, \overline{\dd}^I \psi_0||_{L^2}^2+||\psi_0||_{L^2}^2.
\end{equation}
Here $0<\gamma<1$, $s>0$, and $\overline{\dd}$ is the Levi-Civita connection associated with $\gbar_{ij}$.

\begin{thm}
The initial value problem of the Einstein-Dirac system is wellposed, i.e.\ for a set of initial data satisfying the constraint equations and with enough regularity, there exists a unique maximal development\footnote{This corresponds to the notion of local wellposedness in PDE terminology.}. The Einstein-massless Dirac system 
admits a global solution, provided that the initial data also satisfies the smallness condition $\mathcal E_N(0)+\mathcal E^1_N(0)+M\leq \varepsilon$ for some $N\geq 9$, with $\varepsilon$ smaller than a small number $\varepsilon_0$. Moreover, there is a coordinate system $(t,\x_1,\x_2,\x_3)$ such that the solution satisfies the following decay estimates, for some small $\delta>0$: \footnote{The notation $A\lesssim B$ means that there is a constant $C>0$ such that $A\leq CB$. We will also often use $C$ to represent the existence of such constant, whose value may vary in different cases.}
\begin{equation}
    |\dt\z^I \hi|\lesssim \varepsilon (1+t+r^*)^{-1}\ln \left(\frac{1+t+r^*}{1+|q^*|}\right)\,  (1+|q^*|)^{-1+2\delta}(1+q_+^*)^{-\min(\y,2s)},\ \ |I|\leq N-5,
\end{equation}
\begin{equation}
    |\z^I \psi|\lesssim \varepsilon (1+t+r^*)^{-1}(1+|q^*|)^{-\frac 12+\delta}(1+q_+^*)^{-s},\ \ |I|\leq N-5,
\end{equation}
where $r^*=\sqrt{\x_1^2+\x_2^2+\x_3^2}$, $q^*=r^*-t$, $q^*_+=\max\{q^*,0\}$, and $\z\in \{\dt_\mu,\widetilde \Omega_{ij}=\x_i \dt_j-\x_j \dt_i,\widetilde\Omega_{0i}=t\dt_i+\x_i \dt_t,\widetilde S=t\dt_t+\x_i \dt_i\}$, the set of Minkowski commuting vector fields in this coordinate system.
\end{thm}

\subsection{Idea of proof}
We discuss here the main idea for our proof of the global result of the massless Einstein-Dirac system. We will also see how we can also show the local wellposedness for the case with a nonzero mass.
\paragraph{Weak null structure of the Einstein vacuum equation in wave coordinates}
We first review the weak null structure for nonlinear wave equations, introduced by Lindblad-Rodnianski in \cite{lindblad2003weaknull}, where they used the asymptotic systems to successfully predict the global existence of the Einstein equation in wave coordinates. A wave equation is said to satisfy the weak null condition if the asymptotic system admits a global solution. The asymptotic system is derived by ignoring cubic terms and terms involving derivatives of components tangential to the light cone, i.e.\  derivatives in $L=\d_t+\d_r$ and spherically tangential direction $S_1,S_2$, as they are expected to decay much better along the light cone, and the derivative in $\lb=-2\d_q=\d_t-\d_r$ decays less. Then the main contribution from the inhomogeneous term of the reduced Einstein equation is $L_\mu L_\nu P(\d_q h,\d_q h)$, where $P(\pi,\theta)=\frac 14m^{\a\b}\pi_{\a\b} m^{\rho\sigma}\theta_{\rho\sigma}-\frac 12 m^{\a\b}m^{\rho\sigma}\pi_{\a\rho}\theta_{\b\sigma}$.
The asymptotic system then reads
\begin{equation}\label{asymptoticsystemvacuum}
    (\d_t+\d_r)\d_q D_{\mu\nu} =H_{LL} \d_q^2 D_{\mu\nu}-\frac 1{4r} L_\mu L_\nu P(\d_q D,\d_q D),
\end{equation}
where $D_{\mu\nu}=rh_{\mu\nu}=r(g_\mn-m_\mn)$. We now first consider $D_{TU}$, where $T\in \{L,S_1,S_2\}$ is tangential and $U$ is not necessarily tangential. Using $L_\mu L_\nu T^\mu U^\nu=0$, we can integrate along the $\d_t+\d_r-H_{LL}\d_q$ direction (note that $H_{LL}$ is small as we will see that its leading behavior is $-2M/r$) to get $\d_q D_{TU}\sim \textrm{const}$, so $h_{TU}\sim t^{-1}$ along the light cone.

For the remaining component, we need the identifications $\d_q D_{LT}\sim 0$ and $\d_q \slashed{\mathrm{tr}} D\sim 0$, derived from the wave coordinate condition, by contracting \eqref{wcc} with $T$ and $\lb$ respectively, where $T$ is a tangential vector, and $\slashed{\mathrm{tr}}$ is the trace on the sphere. 
Then expanding $P(\d_q D,\d_q D)$ in the null frame $\{L,\lb,S_1,S_2\}$ we have
$P(\d_q D,\d_q D)\sim \frac 14 \d_q D_{LL}\d_q D_{\lb\lb}+\d_q D_{TT}\, \d_q D_{TT}$. Here the only components we do not know is $\d_q D_{\lb\lb}$, and we have $\d_q D_{LL}\sim 0$ so this is not a problem.
Then integrating the equation we get a logarithmic growth of $\d_q D_{\lb\lb}$, and get $h_{\lb\lb}\sim t^{-1}\ln t$. From this we know that the asymptotic system admits a global solution, and it turns out that this predicts the correct behavior of the solution of the original system of quasilinear wave equations.


\paragraph{Einstein-Dirac coupled system}To study the Einstein-Dirac system, we need to formulate the Dirac equation in curved spacetime. The Dirac spinor field in Minkowski spacetime is a field satisfying the transformation law under Lorentz transformation $\Lambda$
$$\psi\ra S(\Lambda)\psi,$$
where $S$ is the spin representation of Lorentz group, and we see that generalization is necessary for curved spacetime. As such, one needs to pick a field of orthonormal frame (tetrad) $\{e_a^{\ \mu}\}$ to make tangent space look like Minkowski spacetime, and then define the spinor to satisfy the same transformation property under \textit{local} Lorentz transformation. In the language of fiber bundle, the spinor field here is a section of the spinor bundle, associated with the spin representation and the principal orthonormal frame bundle.

This also gives rise to the covariant derivative of the spinor field, where a new connection term appears compared with the Levi-Civita connection: 
$$D_\mu \psi=\d_\mu \psi-\frac 14 \omega_{\mu ab}\Sigma^{ab}\psi.$$
Here $\omega_{\mu ab}=g(e_a,\nabla_\mu e_b)$ is the spin connection $1$-form with $\dd$ being the Levi-Civita connection associated with the metric $g$, and $\Sigma^{ab}$ is the generator of the spin representation. We will analyze the Dirac equation using this formulation.

\paragraph{A second-order equation for Dirac field}
The presence of the Dirac field $\psi$ produces new terms on the right hand side of the Einstein equation, where first-order derivatives of $\psi$ appear. This is undesirable, as the Dirac equation itself does not give good control of the derivative of $\psi$. We deal with this by considering a second-order equation of the spinor field, derived from the Dirac equation.

For the Dirac equation in Minkowski spacetime $\gb^\mu \d_\mu \psi+im\psi=0$, we can apply the Dirac operator $\gb^\mu \d_\mu$ again to get $$\d^\mu \d_\mu \psi-m^2 \psi=0.$$
This means that the components of $\psi$, assuming enough regularity, satisfiy the Klein-Gordon equation, and if $m=0$, it becomes wave equation. 

In our case, where the spacetime is curved, the covariant derivative $D$ is not commutative anymore:
$$[D_\mu,D_\nu]\psi=-\frac 14 R_{\mu\nu\sigma\delta} \gamma^\sigma \gamma^\delta \psi.$$
But still, by the anti-commutation relation of our adpated Gamma matrices $\gamma^\mu\gamma^\nu+\gamma^\nu\gamma^\mu=-2g^{\mu\nu}I$ one can show that
$$\gamma^\mu D_\mu(\gamma^\nu D_\nu\psi)=-D^\mu D_\mu \psi+\frac 14R\psi.$$
As we shall see, this gives a quasilinear hyperbolic equation for $\psi$, where we have the estimate for the derivative of $\psi$. This equation is also crucial for us to perform $L^\infty$-$L^\infty$ estimate using techniques of wave equations, in order to get the sharp decay of the field, as it looks quite difficult to integrate the Dirac equation directly to derive decay estimates.

\paragraph{Choice of gauge}
In order to derive the estimate explicitly, just as the case for the Einstein equation, where we want to fix a coordinate gauge, here we want to also fix the frame gauge, i.e.\  pick a choice of the tetrad. Since we are studying a spacetime that is very close to the Minkowski spacetime, we are able to make a global choice, and it might be the easiest to orthogonalize the coordinate basis directly. In view of the formula for Gram-Schmidt process, we see that the components of tetrad behave like the perturbation of metric from the Minkowski metric $h=g-m$:
$$|(e_a)^\mu-(\d_a)^\mu| \lesssim |h|,\ \ |\d (e_a)|\lesssim |\d h|,\ \ |Z (e_a)|\lesssim |Zh|,$$
where $Z$ is one of the Minkowski commuting vector fields, which we will explain later.

\paragraph{Avoiding loss of derivative}
Now we have established a new equation of $\psi$
$$D^\mu D_\mu \psi-m^2 \psi=\frac 14 R\psi.$$
When the metric $g$ and the field $\psi$ solve our system, we have $R=-m\pb\psi$, so the term $\frac 14 R\psi$ is a cubic term. 
While the equation looks nice, we need to take into consideration the behavior of the connection term. Recall we have $D_\mu=\dd_\mu+\Omega_\mu=\dd_\mu-\frac 14\omega_{\mu ab}\Sigma^{ab}$ for Dirac spinor field. Moving all terms involving the connection to the right hand side, we get $$g^{\mu\nu}\dd_\mu \dd_\nu \psi-m^2 \psi=-g^{\mu\nu}(\dd_\mu \Omega_\nu) \psi-2g^{\mu\nu}\Omega_\mu \d_\nu \psi-g^{\mu\nu}\Omega_\mu \Omega_\nu \psi+\frac 14 R\psi.$$
We note that $\Sigma^{ab}$ are constant matrices, and the spin connection term $\omega_{\mu ab}=g(e_a,\dd_\mu e_b)$ behaves like the first order derivative of the metric, according to our choice of the tetrad. Hence, there is a \textbf{potential loss of derivative} for the first term on the right hand side, which we need to deal with, in order to get even the local wellposedness result. We tackle this problem by the following observation: we first have
$$g^\mn \dd_\mu \omega_{\nu ab}=g^\mn g(e_a,\dd_\mu\dd_\nu e_b)+\text{lower order terms (l.o.t.)}.$$
There is a covariant wave operator on $e_b$, and we have $$g^{\mu\nu}\dd_\mu \dd_\nu (e_b)^\a=\widetilde\Box_g (e_b)^\a+g^\mn (\d_\mu \Gamma_{\nu\rho}^{\ \ \, \a}) (e_b)^\rho+\text{l.o.t.}$$
The first term is essentially like $\widetilde\Box_g h$, which can be controlled using the reduced Einstein equation; for the second term, we make use of the relation on the curvature tensor
$$\d_\mu \Gamma_{\nu\rho}^{\ \ \,\a}=\d_\rho \Gamma_{\mu\nu}^{\ \ \,\a}+R_{\rho\mu\nu}{}^{\a}+\text{l.o.t}.$$
Contracting this with $g^\mn$ we see the last term is like the Ricci curvature, which can be controlled by the Einstein equation; the first term is now 
$$g^\mn \d_\rho \Gamma_\mn^{\ \ \,\a}=\d_\rho(g^\mn \Gamma_\mn^{\ \ \,\a})+\text{l.o.t}=0+\text{l.o.t.}$$
where we used the wave coordinate condition $g^\mn \Gamma_\mn^{\ \ \,\a}=0$ for the last equality. We note that this works similarly for generalized wave coordinates, which we will use. Therefore, we have dealt with the potential loss of derivative and now this term just behaves like a cubic term\footnote{We remark that a similar idea was used in \cite{Baoavoidloss} decades ago in different notations.}. As a consequence, one can show the local wellposedness result of the Einstein-Dirac system, possibly with a nonzero mass.

\paragraph{Coupled asymptotic system}Now we have derived a system of quasilinear wave equations for the metric and the spinor field. We first look at the equation of $\psi$:
$$\widetilde\Box_g \psi=g^\mn \Omega_\mu \d_\nu \psi+\text{cubic terms}.$$
We again ignore tangential derivatives, so by expanding in the null frame $\{L,\lb,S_1,S_2\}$ we can focus on the term $\Omega_L \lb \psi$. For this we find that
$\omega_{L ab}=g(e_a,L^\mu \dd_\mu e_b)=g_{\alpha\beta} (e_a)^\alpha L (e_b)^\beta+g_{\alpha\beta}(e_a)^\alpha L^\mu \Gamma_{\mu\nu}^{\ \ \,\beta}(e_b)^\nu$. The first term behaves like a good derivative, according to the choice of tetrad. For the second term, it is a Christoffel symbol contracted with $L$, and we have
$$L^\mu \Gamma_{\mn}^{\ \ \,\b}g_{\a\b}=\frac 12 (Lg_{\nu\a}+\d_\nu g_{L\a}-\d_\a g_{L\nu}).$$
The worst term possible from this expression is $\underline{L} g_{L\lb}$. However, we see that the cancellation of last two terms means that there will be no such term. As a result, the worst term here is like $\lb g_{LT}$, which actually behaves like a good derivative using the estimates from the wave coordinate condition. Therefore, we see that there is a \textbf{null structure} in the semilinear term, and hence the asymptotic equation of $\psi$ reads $(\d_t+\d_r-H_{LL}\d_q)\d_q(r\psi)=0$, and we get $\d_q(r\psi)\sim \mathrm{const}$ and $\psi\sim t^{-1}$ along the light cone.

For the Einstein equation, we now have a new term from the energy-momentum tensor, and the asymptotic system reads, again using the identification from the wave coordinate condition: 
\begin{multline*}
    (\d_t+\d_r-H_{LL}\d_q)\d_q D_{\mu\nu} =-\textstyle\frac 1{4r} L_\mu L_\nu (\frac 14 \d_q D_{LL}\cdot\d_q D_{\lb\lb}+\d_q D_{TT}\cdot\d_q D_{TT})\\
    \textstyle-\frac 1{2r}\, \mathrm{Im}\, (L_\nu\overline{\Psi}\gb_\mu \d_q \Psi+L_\mu\overline{\Psi}\gb_\nu \d_q \Psi),
\end{multline*}
where $D_\mn=rh_\mn$, $\Psi=r\psi$. In this case, the right hand side is no longer zero when we contract the equation with $T^\mu U^\nu$, where $T$ is tangential and $U$ is arbitrary. As a result, compared with the vacuum case, we no longer have the sharp decay of all $TU$ components $h_{TU}$ along the light cone. Nevertheless, we can make use of \textbf{an extra room} in this weak null structure, which is not used in the original work $\cite{03,04}$\footnote{In \cite{03,04}, the term $\d_q D_{TT}\cdot \d_q D_{TT}$ is written as $\d_q D_{TU}\cdot \d_q D_{TU}$, and the identification $\d_q \slashed{\mathrm{tr}}D \sim 0$ is not used.}: we see that the right hand side is still zero when we contract it with $T^\mu T^\nu$, which gives $h_{TT}\sim t^{-1}$ and saves the equation in view of the structure. Moreover, expanding the Dirac equation (modulo cubic terms, it does not matter if we look at the curved or the Minkowski one) in the null frame, we notice that the term $\gb_L \d_q \psi$ is equal to a sum of tangential derivatives, which means that the asymptotic inhomogeneous term is also zero if we contract it with $L^\mu \lb^\nu$, so we also have $h_{L\lb}\sim t^{-1}$. This observation is crucial for us to derive the estimate of Dirac equation itself, as we will see.

\paragraph{Energy estimates for quasilinear wave equations; vector field method}
In \cite{04}, the following estimate is derived and played a fundamental role there:
\begin{equation}
    \int_{\Sigma_t}|\d \phi|^2w+\int_0^t\int_{\Sigma_\tau}|\db \phi|^2 w'\leq 8\int_{\Sigma_0}|\d \phi|^2 w+\int_0^t\int_{\Sigma_\tau}\frac {C\varepsilon}{1+\tau}|\d \phi|^2 w+16\int_0^t\int_{\Sigma_\tau} |\widetilde\Box_g \phi||\d_t\phi|w,
\end{equation}
where $\Sigma_t=\{t=\mathrm{const}\}$, $\db$ denotes tangential derivatives, and the weight function \begin{equation*}
    w(q)=\begin{cases}
1+(1+|q|)^{-2\mu},\, q<0,\\
1+(1+|q|)^{1+2\gamma},\, q>0,
\end{cases}
\end{equation*}
with $0<\mu\leq 1/2$, $0<\gamma<1$, and $q=r-t$. This is a generalization of the energy estimate associated with the multiplier $\d_t$ for linear wave equation. The weight function helps capture the exterior decay, and gives the spacetime integral on the left hand side, which is very useful for controlling terms with tangential derivative.

It is standard to combine the energy estimate of wave equations with commuting vector fields. There is a set of vector fields that commutes well with the flat wave operator:
$$\mathcal{Z}=\{\d_\mu,\ \Omega_{ij}=x_i\d_j-x_j\d_i,\ \Omega_{0i}=t\d_i+x_i\d_t,\ S=t\d_t+x^i\d_i\}$$
By commuting these vector fields with the wave operator, one can establish control of higher order energies $\int |\d Z^I \phi|^2\, dx$. We will apply this to the reduced Einstein equation, and the quasilinear wave equation of $\psi$ (of course, we will need to deal with the commutator since the wave operator here is not the flat one), using different weight functions in $q$, capturing the correct behavior of them in the exterior. As in many previous works, we use the bootstrap argument on the energy with higher order vector fields. With the bootstrap assumption on the energy, we immediately get preliminary decay estimates for various quantities, by using the Klainerman-Sobolev inequality. These decays are not sufficient, but they help us identify good terms, and derive improved estimates.

\paragraph{Estimate for the Dirac equation itself}
Note that the asymptotic system only suggests the behavior of solutions near the light cone, and it does not provide information about \textbf{the interior}. From the expression of the energy-momentum tensor, we see that there are terms of the form $\psi \cdot \d\psi$ on the right hand side of the Einstein equation. As $\psi$ exhibits worse behavior than $\d\psi$ in the interior, one needs to use estimates stronger than the standard energy estimate: even if the higher order energy $\int |\d Z^I \psi|^2\, dx$ is bounded, the Klainerman-Sobolev inequality only gives $|\d\psi|\leq C\varepsilon (1+t+r)^{-1}(1+|q|)^{-\frac 12}$, so integrating along $q$ direction we get $|\psi|\leq C\varepsilon (1+t+r)^{-1}(1+|q|)^{\frac 12}$ which is not good. The H\"{o}rmander estimate (Proposition \ref{Hormanderprop}) can help improve the interior behavior, but still in our context we can only get $|\psi|\leq C\varepsilon (1+t)^{\delta} (1+t+r)^{-1}$. 
The way we deal with this issue is to develop the estimate from the Dirac equation itself, which is also of its own interest.

For the Dirac equation in Minkowski spacetime, we have the conservation of $L^2$ norm of $\psi$. Moreover, from \cite{bachelot1988DiracVectorField} we know that the modified vector fields $\widehat\Omega_{\mu\nu}=\Omega_{\mu\nu}-\frac 12 \gb_\mu\gb_\nu$ commute well with the Dirac operator. In the massless case, the equation also commutes well with the scaling vector field $S=x^\a \d_\a$, so we can also use the vector-field method in this case. 

We try to establish here an estimate for $L^2$ norm in curved case. We make use of the identity
$$\nabla_\mu (\pb \gamma^\mu \psi)=\overline{\gamma^\mu D_\mu \psi}\psi+\pb \gamma^\mu D_\mu \psi.$$
The connection term on the left produces a term like $\Gamma_{\mu\nu}^{\ \ \,\mu}\pb\gamma^\nu \psi$. The Christoffel symbol can be controlled by the derivative of the metric, but this is not enough as we see from the asymptotic system that some components do not have $t^{-1}$ decay. To solve this issue, we make use of the special structure here:
$$|\Gamma_{\mu\nu}^{\ \ \,\mu}|=|\frac 12 g^{\mu\rho}\d_\nu h_{\mu\rho}|\lesssim |\d\,  \slashed{\mathrm{tr}} h|+|\d h|_{L\lb}+|h||\d h|.$$
These terms can be controlled, among which $|\d h|_{L\lb}$ is saved by our observation above, where we pointed out that $h_{L\lb}\sim t^{-1}$ near the light cone. In this way, we establish the following $L^2$ estimate:
\begin{equation*}
    \int_{\Sigma_{t}} |\psi|^2 \, dx \leq 2\int_{\Sigma_0} |\psi|^2 \, dx
    +4\int_{0}^{t} \int_{\Sigma_\tau} \left(\frac{C\varepsilon}{1+\tau} |\psi|^2  +|\psi||\gamma^\mu D_\mu \psi|\right) \, dx d\tau.
\end{equation*}
This $L^2$ estimate, combined with modified commuting vector fields, will provide much better preliminary estimate of the field in the interior, which is important for the improved decay estimates below, where we obtain a decay that is sharp along the light cone, and sufficient in the interior.

We note that it would be possible to derive a weighted version of the estimate as for the wave equations, and introduce a spacetime integral (which we will discuss in Appendix \ref{bulktermappendix}), which may help capture some null structures in contractions involving Gamma matrices. However, in our case, the decay of the metric is insufficient for us to do this.

\paragraph{Improved decay estimates}To derive the estimate predicted by the asymptotic systems (which one cannot get simply from the energy bound), one can integrate along the characteristics along the light cone, as in \cite{03,04}. We can also do this for the quasilinear wave equation of $\psi$, to get improved behavior of $\d\psi$. However, for the reason above, we need more decay in $q=r-t$ direction, in addition to the improved decay in the direction of light cone. 
Therefore, we need to add $q$-weight in this process\footnote{The size of the weight we can insert is heavily dependent on the preliminary decay we have, and one can understand as we get preliminary decay in the interior $q<0$ by deriving the $L^2$ estimate of $\psi$, and the one in the exterior $q>0$ by including the weight in the $L^2$ estimate of $\d\psi$.}. If the wave operator is the Minkowskian one, then this follows as $\Box \phi=\frac 1r (\d_t+\d_r)(\d_t-\d_r)(r\phi)$ plus angular derivatives, so any function of $q$ remains unchanged along the integral curve of $\d_t+\d_r$. But in our case, as we can see from the asymptotic system \eqref{asymptoticsystemvacuum}, we will get a new term from 
$$(\d_t+\d_r-H_{LL} \d_q)\wb(q)\d_q(r\phi)=\wb(q)(\d_t+\d_r-H_{LL} \d_q)\d_q(r\phi)-H_{LL} \wb'(q) \d_q(r\phi),$$
if we want to insert a weight function $\wb(q)$.
Here the leading behavior of $H_{LL}$ is determined by the Schwarzschild part, which only decays like $r^{-1}$, so we get $\d_q (r\phi)$ growing because of the second term on the right, which makes it \textit{insufficient} for us to close the argument. Therefore, we need to find a better approximation of the characteristics along which the $q$-weight changes less, so that there is no such effect.


\paragraph{Change of coordinates}We have seen from the asymptotic system above that a better approximation of the characteristic than the Minkowski light cone is given by the integral curve of $\d_t+\d_r-H_{LL}\d_q$. The leading behavior of $H_{LL}$ is contributed by the Schwarzschild part:
$$H_{LL}\sim -\frac Mr \delta_{LL}+o(r^{-1-\y})\sim -\frac{2M}r.$$
Motivated by this, Lindblad \cite{1606Einsteinasymptotics} made a change of coordinates $r^*=r+M\ln r$ close to the light cone, and gave asymptotics of solutions of the Einstein equation in wave coordinates. In the new coordinates, the reduced wave operator $\widetilde \Box_{\gt}=\gt^{\mu\nu}\dt_\mu\dt_\nu$ is much more close to the linear wave operator in our new coordinates $\Box^*=\mh^{\mu\nu}\dt_\mu\dt_\nu$. In this way, we are able to resolve the problem of $q$-weight above. Additionally, as in \cite{KLMKG21}, many estimates perform better in this coordinate, such as the estimate of the commutator from our wave equations.





\paragraph{Commutator for Dirac equation}As we remarked above, we also need to derive the estimate of the $L^2$ norm of $\psi$ itself. The modified vector fields commute well with the Minkowski part, i.e.\ the flat Dirac operator $\gb^\mu \d_\mu$. However, here we are dealing with quasilinear case, and this is even more difficult than the quasilinear part of the wave equations here, as for the wave equation the quasilinear part reads $g^{\mu\nu}\d_\mu\d_\nu \psi$, and we can use the contraction structure along with wave coordinate condition to estimate, while here the counterpart is from the difference between $\gamma^\mu$ and $\gb^\mu$, and in view of the construction, components of the metric are mixed, so the situation is more like the case in \cite{05}, where the scalar equation $g^{\a\b}(\phi)\d_\a\d_\b\phi=0$ is studied.

Indeed, the commutator for the curved Dirac equation looks like $$[\hat Z^I,\gamma^\mu D_\mu]\psi\approx Z^J(\gamma^\mu-\gb^\mu)\d_\mu Z^K\psi+\text{other terms},$$
where $|K|<|I|$, and the case where $|J|\leq |K|$ looks \textbf{problematic}, as the term $Z^J(\gamma^\mu-\gb^\mu)$ behaves like $Z^J h$, which does not have sharp decay $t^{-1}$ (recall we get $h\sim t^{-1}\ln t$\, from the asymptotic system), so the standard way will produce a term like 
$$\sum_{|I'|\leq |I|}\iint \frac{\ln t}{1+t}\, \frac{1}{(1+|q|)^{1-\delta}} |Z^{I'}\psi|^2\, dx dt$$ 
on the right hand side of the estimate, for some small $\delta>0$, which is \textit{not} good in view of Gronwall's inequality. 
An interesting fact is that since we have derived the estimate for the wave equation of $\psi$, we also have the estimate of $\int |\d Z^I \psi|^2\, dx$, and the control is better in $t$. However, this does not give satisfactory control either: the problem arises from the interior, as we have a $q$-decay loss in the decay $|Z^I h|\lesssim \varepsilon (1+t)^{-1} (1+|q|)^\delta$, so what we can get is $$\iint \varepsilon (1+t)^{-1}\ln t\, (1+|q|)^{\delta}\,  |\d Z^K\psi|^2\, dx dt$$ which is again not enough. The way we deal with this problem is to make use of both estimates, and absorb the decay loss of from one by the other. Using $|\d Z^K \psi|\lesssim (1+|q|)^{-1}\sum_{|K'|\leq |K|+1} |Z^{K'}\psi|$ we have 
\begin{multline*} 
|Z^I h||\d Z^K \psi|^2\lesssim \varepsilon (1+t)^{-1} \ln t (1+|q|)^\delta (\frac{1}{1+|q|}|Z^{K'} \psi|) |\d Z^K \psi|\\ \lesssim \varepsilon (1+t)^{-1} \frac{1}{(1+|q|)^{2-2\delta}} |Z^{K'}\psi|^2+\varepsilon (1+t)^{-1} (\ln t)^2 |\d Z^K \psi|^2,
\end{multline*}
and we are able to control both terms. In this way, we resolve the potential problem of the commutator.

\vspace{2ex}
\begin{remark}
We briefly discuss the global problem in the massive case. Instead of a wave equation of $\psi$, in this case we get the equation with an additional term $m^2\psi$, which makes it a Klein-Gordon equation. Note that unlike the energy estimate of wave equation, where we can only control $\d\psi$, the estimate of the Klein-Gordon equation directly gives the control of $\psi$ (more precisely, $m^{-1}\psi$). As such, in view of the analysis in this manuscript, it may not be necessary anymore to derive the estimate of the Dirac equation itself, and the strategy will be similar with the one of the Einstein-Klein-Gordon system (with mass $m^2$). We note that, however, for Einstein-Klein-Gordon one cannot use the full set of vector fields, as the scaling vector field does not commute well with the Klein-Gordon equation, and hence this gives rise to a highly nontrivial problem. See \cite{LMCMP} for the stablity of Minkowski for restricted data, using the hyperboloidal foliation method, and \cite{ionescu2019einstein} for more general data with the use of Fourier analysis. Also, one needs to deal with the different structure of the energy-momentum tensor for the Dirac field as in this work.
\end{remark}

\paragraph{Acknowledgements.}The author would like to thank his advisor, Hans Lindblad, for suggesting this problem and many helpful discussions.

\section{Preliminaries}\label{Chnotation}
Throughout this paper, we use the sign convention $(-+++)$, i.e.\ $m_\mn=\mathrm{diag}\{-1,1,1,1\}$. \textbf{In this section}, unless otherwise specified, we raise and lower tensorial indices with respect to the metric $g$. We will sometimes use the notation $h_{\mu\nu}=g_{\mu\nu}-m_{\mu\nu}$ and $H^{\a\b}=g^{\a\b}-m^{\a\b}$, where $g^{\a\b}$ is the inverse matrix of $g_{\mu\nu}$.



We use the notation $\psi^\dagger$ to represent the standard conjugate transpose. We will define the spinor conjugate $\pb$ later.

\subsection{Levi-Civita connection, Riemann curvature tensor}
We define $\dd$ as the standard Levi-Civita connection associated with the metric $g$. In any coordinate system $\{x^\a\}$, the covariant derivative can be written as 
\begin{multline}\label{covariant derivative}
    \dd_\mu T^{\a_1\a_2\cdots \a_s}_{\b_1\b_2\cdots \b_r}=\d_\mu T^{\a_1\a_2\cdots \a_s}_{\b_1\b_2\cdots \b_r}+\Gamma^{\ \ \,\a_1}_{\mu \nu}\,T^{\nu\, \a_2\cdots \a_s}_{\b_1\b_2\cdots \b_r}+\cdots+\Gamma^{\ \ \,\a_s}_{\mu \nu}\,T^{\a_1\a_2\cdots \nu}_{\b_1\b_2\cdots \b_r}-\Gamma_{\mu\b_1}^{\ \ \ \,\nu}\,T^{\a_1\a_2\cdots \a_s}_{\nu\, \b_2\cdots \b_r}\\
    -\cdots-\Gamma_{\mu\b_r}^{\ \ \ \,\nu}\,T^{\a_1\a_2\cdots \a_s}_{\b_1\b_2\cdots\, \nu}.
\end{multline}
Here the Christoffel symbol is given by the formula
\begin{equation}\label{Christoffelsymbol}
    \Gamma_{\mu\nu}^{\ \ \,\rho}=\frac 12 g^{\lambda\rho}(\d_\mu g_{\nu\lambda}+\d_\nu g_{\lambda \mu}-\d_\lambda g_{\mu\nu}).
\end{equation}

The Riemann curvature tensor is defined as the (1,3)-tensor $R$ such that
\begin{equation}
    \dd_\mu\dd_\nu v_\rho-\dd_\nu\dd_\mu v_\rho=R_{\mu\nu\rho}^{\ \ \ \ \lambda}v_{\lambda},
\end{equation}
or equivalently using vector fields 
\begin{equation}
    \dd_\mu\dd_\nu X^\rho-\dd_\nu\dd_\mu X^\rho=-R_{\mu\nu\lambda}^{\ \ \ \ \rho}X^{\lambda}.
\end{equation}
The Ricci curvature is defined by contraction: $R_{\mu\nu}=R_{\mu\rho\nu}^{\ \ \ \ \rho}$. Now contracting with the metric $g$ again we get the scalar curvature $R=g^{\mu\nu}R_{\mu\nu}$.

We list some properties of the curvature tensor, which are standard and can be found in any textbooks in (pseudo-)Riemannian geometry:
\begin{Prop}We have $R_{\mu\nu\rho\sigma}=-R_{\nu\mu\rho\sigma}$, $R_{\mu\nu\rho\sigma}=R_{\rho\sigma\mn}$, and the first Bianchi identity: $R_{\mu\nu\rho\sigma}+R_{\mu\rho\sigma\nu}+R_{\mu\sigma\nu\rho}=0$.
\end{Prop}

Using (\ref{Christoffelsymbol}), we can express the curvature tensor in given coordinates in the Christoffel symbol: 
\begin{equation}\label{curvatureinChristoffelsymbol}
    R_{\mu\sigma\nu}{}^{\rho}=-\d_\mu \Gamma_{\sigma\nu}^{\ \ \,\rho} +\d_\sigma \Gamma_{\mu\nu}^{\ \ \,\rho} - \Gamma_{\mu\lambda}^{\ \ \,\rho}\,\Gamma_{\sigma\nu}^{\ \ \,\lambda}+\Gamma_{\sigma\lambda}^{\ \ \,\rho}\,\Gamma_{\mu\nu}^{\ \ \,\lambda}.
\end{equation}

\subsection{Gamma matrices, tetrad formalism}
We start from this subsection the formulation of the Dirac equation, especially in curved spacetime. The case for Minkowski spacetime is standard, and for curved one, we refer to previous literatures \cite{brill1957interaction}, \cite{brill1966cartan}, \cite{pollock2010dirac}.

The Gamma matrices in Minkowski spacetime $m_{\mu\nu}=\mathrm{diag}(-1,1,1,1)$ reads
$$\begin{aligned}
&\gb^{0}=\left(\begin{array}{cccc}
1 & 0 & 0 & 0 \\
0 & 1 & 0 & 0 \\
0 & 0 & -1 & 0 \\
0 & 0 & 0 & -1
\end{array}\right), \quad \gb^{1}=\left(\begin{array}{cccc}
0 & 0 & 0 & 1 \\
0 & 0 & 1 & 0 \\
0 & -1 & 0 & 0 \\
-1 & 0 & 0 & 0
\end{array}\right), \\
&\gb^{2}=\left(\begin{array}{cccc}
0 & 0 & 0 & -i \\
0 & 0 & i & 0 \\
0 & i & 0 & 0 \\
-i & 0 & 0 & 0
\end{array}\right), \quad \gb^{3}=\left(\begin{array}{cccc}
0 & 0 & 1 & 0 \\
0 & 0 & 0 & -1 \\
-1 & 0 & 0 & 0 \\
0 & 1 & 0 & 0
\end{array}\right) .
\end{aligned}$$
They are chosen to satisfy the anticommutation relation
\begin{equation}\label{anticommutation}
    \gb^\mu\gb^\nu+\gb^\nu\gb^\mu=-2m^{\mu\nu} I,
\end{equation}
where $I$ is the $4\times 4$ identity matrix.

By a Dirac spinor field in Minkowski spacetime we mean a field $\psi$ taking value in $\mathbb{C}^4$, which transforms in a special way under Lorentz transformation $\Lambda$ by $\psi\ra S(\Lambda)\psi$, where $S(\varepsilon_{\mu\nu})=-\frac 14 \Sigma^{\mu\nu}\varepsilon_{\mu\nu}$ for infinitesimal Lorentz transformation $\varepsilon_\mn$ with $\Sigma^\mn=\frac 12[\gb^\mu,\gb^\nu]$, hence giving a representation of the Lorentz group $\mathrm{SO}(1,3)$ \footnote{This representation is, strictly speaking, not a representation of the Lorentz group, but its double covering group $\mathrm{Spin}(1,3)$. Nevertheless, since this only causes a difference of sign that has no effect on our analysis, we shall just view the structure group as $\mathrm{SO}(1,3)$.}, called spin representation.

The Gamma matrices transform under the Lorentz transformation $\Lambda_\mu^{\ \nu}$ as $$\gb^\mu\ra S(\Lambda)\gb^\mu S(\Lambda)^{-1} \Lambda_\mu^{\ \nu}=\gb^\nu,$$ where the equality can be verified by recalling the definition of $S$, and commuting Gamma matrices with $\Sigma^\mn$, using \eqref{anticommutation}. In particular, $\gb^\mu$ does not transform like a $4$-vector.

The adjoint spinor of $\psi$ is defined by $\pb=\psi^\dagger \gb^0$. We note that this applies to all situations, including the case for a curved background metric, which we will discuss below. Then in view of the transformation law of $\psi$, the adjoint spinor $\pb$ transforms under the Lorentz transformation as $\pb\ra\pb\,  S(\Lambda)^{-1}$, as one can verify for the generator that $\gb^0 \Sigma^\mn \gb^0=-\Sigma^\mn$ simply using \eqref{anticommutation}.
Combining this together, it is easy to see that $\pb\gb^\mu\psi$ is a $4$-vector, and $\pb\psi$ is a scalar.

\vspace{1ex}
In curved spacetime, we cannot define a global Lorentz transformation. In order to define the spinor, we need the notion of tetrad. By a tetrad $\{e_a\}$ we mean a field of orthonormal frame in our Lorentzian spacetime, i.e.\ a set of vector fields satisfying $g_{\mu\nu}(e_a)^\mu (e_b)^\nu=m_{ab}$. In the language of fiber bundle, this means a section of the orthonormal frame bundle over the spacetime manifold. In this way, the notion of local Lorentz transformation is clear, and we can define the Dirac spinor field as the section of the spinor bundle associated with the above spin representation. This is equivalent with the requirement that the field transforms under local Lorentz transformation $\Lambda(x)$ by
$$\psi(x) \ra S(\Lambda(x)) \psi(x).$$

We often write the tetrad as $e_{\ a}^\mu$, and $e_a=e_{\ a}^\mu\d_\mu$. The Greek letter here represents the standard tensor indices, and the Latin letter is called flat indices or tetrad indices, and we raise the flat indices only using the Minkowski metric. We also denote the dual frame by $e^a=e^a_{\ \mu} dx^\mu$ and often write it simply as $e^a_{\ \mu}$.

For a vector field $X^\mu$ we can define its flat components $X^a=X^\mu e^a_{\ \mu}$. Similarly, we can also get curved components from the flat ones. 

For the Levi-Civita connection $\dd$, we require it to only interact with curved indices. We now define a covariant derivative which also acts with flat indices, hence covariant under local Lorentz transformations. The spin connection is defined by \begin{equation}
    \omega_{\mu ab}=g(e_a,\dd_\mu e_b)=e_{\nu a}(\d_\mu e_{\ b}^\nu+\Gamma_{\mu\rho}^{\ \ \,\nu}e_{\ b}^\rho),
\end{equation} and the derivative covariant both under coordinate transformation and local Lorentz transformation reads for $v^a$
$$D_\mu v^a=\d_\mu v^a+\omega_{\mu\ b}^{\ a} v^b.$$
We see that the spin connection plays a similar role as the Christoffel symbol does for curved indices. Similarly, we can derive the expression of $D$ on other tensors, for instance on mixed tensors:
$$D_\mu T_{\ \nu}^a=\d_\mu T_{\ \nu}^a-\Gamma_{\mu\nu}{}^{\rho}\, T_{\ \rho}^a+\omega_{\mu\ b}^{\ a}\, T_{\ \mu}^b.$$

\begin{lem}[tetrad postulate]
The tetrad vector fields are covariantly constant, i.e.
\begin{equation}\label{tetradpostulate}
    D_\nu e^\mu_{\ b}=\dd_\nu e_{\ b}^\mu-\omega_{\nu\ b}^{\ a}\, e_{\ a}^\mu=0.
\end{equation}
\end{lem}

\begin{proof}
By the definition of spin connection and $(e^a)_\mu (e_a)^\nu=\delta_\mu^\nu$ we have 
\begin{equation*}
    D_\nu e^\mu_{\ b}=\dd_\nu(e_b)^\mu-(e^a)_\rho \dd_\nu (e_b)^\rho (e_a)^\nu=\dd_\mu (e_b)^\mu-\delta_\rho^\nu \dd_\nu (e_b)^\rho=0.\eqno\qed
\end{equation*}
\renewcommand{\qedsymbol}{}
\end{proof}
The curved Gamma matrices are defined using the tetrad: $\gamma^\mu=\gb^a e^\mu_{\ a}$, and we clearly have a similar relation with the metric:
\begin{equation}
    \gamma^\mu\y^\nu+\y^\nu\y^\mu=-2g^\mn I.
\end{equation}

\subsection{Covariant derivative of the spinor field}
The covariant derivative $D$ keeps the covariance of quantities under local Lorentz transformations. For the Dirac spinor field, the covariant derivative is derived by contracting the spin connection with the generator of the spinor representation
\footnote{The form depends on the choice of the sign convention and Gamma matrices. See \cite{diracguide} for the effects of the choices.}
\begin{equation}\label{eqdefinitionD}
    D_\mu \psi=\d_\mu\psi-\frac 14 \omega_{\mu ab}\Sigma^{ab} \psi.
\end{equation}
Here $\Sigma^{ab}=\frac 12(\gb^a\gb^b-\gb^b\gb^a)$ just as the Minkowski case, i.e.\ under an infinitesimal Lorentz transformation $\varepsilon_{ab}$, the spinor $\psi$ transforms as $\psi\ra (1-\frac 14\varepsilon_{ab}\Sigma^{ab})\psi$. One can verify $D$ is covariant under both coordinate transformation and local Lorentz transformations (see for instance \cite{anomalies}, \cite{yepez2011einsteinvierbein}).

For simplicity we denote $\Omega_\mu=-\frac 14 \omega_{\mu ab}\Sigma^{ab}$, so $D_\mu \psi=\d_\mu\psi+\Omega_\mu\psi$. For adjoint spinor $\pb=\psi^\dagger\gb^0$, the covariant derivative reads $$D_\mu \pb=\d_\mu \pb-\pb\, \Omega_\mu.$$

We will need some identities involing Gamma matrices, which we in fact have used some in the discussion for the case in Minkowski spacetime.
\begin{lem}
We have for the commutator of the generator of spinor representation:
\begin{equation}\label{commutatorsigmaabsigmacdeq}
    [\Sigma^{ab},\Sigma^{cd}]=2(m^{ac}\Sigma^{bd}-m^{ad}\Sigma^{bc}+m^{bd}\Sigma^{ac}-m^{bc}\Sigma^{ad}).
\end{equation}
\end{lem}

\begin{proof}
Using the identity (\ref{anticommutation}) twice we have $\gb^a\gb^b\gb^c=\gb^b\gb^c\gb^a-2m^{ab}\gb^c+2m^{ac}\gb^b$. Use this identity twice we get
\begin{multline*}
    \gb^a\gb^b\gb^c\gb^d=\gb^a(\gb^c\gb^d\gb^b-2m^{bc}\gb^d+2m^{bd}\gb^c)=\gb^c\gb^d\gb^a\gb^b+(-2m^{ac}\gb^d+2m^{ad}\gb^c)\gb^b\\
    +\gb^a(-2m^{bc}\gb^d+2m^{bd}\gb^c).
\end{multline*}
Interchanging $(a,b)$ and $(c,d)$ we get
\begin{equation*}
    \gb^c\gb^d\gb^a\gb^b=\gb^a\gb^b\gb^c\gb^d+(-2m^{ca}\gb^b+2m^{cb}\gb^a)\gb^d +\gb^c(-2m^{da}\gb^b+2m^{db}\gb^a).
\end{equation*}
Taking the difference of two equations and dividing it by 2 we get
\begin{equation}
    \gb^a\gb^b\gb^c\gb^d-\gb^c\gb^d\gb^a\gb^b=2(m^{ac}\Sigma^{bd}-m^{ad}\Sigma^{bc}+m^{bd}\Sigma^{ac}-m^{bc}\Sigma^{ad}).
\end{equation}
The right hand side here is antisymmetry in $a$ and $b$, and in $c$ and $d$, so antisymmetrizing them in the left hand side we get the result.
\end{proof}

\begin{lem}
We have $\gb^0 (\gb^a)^\dagger \gb^0=\gb^a$, so $\gb^0 (\gamma^\mu)^\dagger \gb^\mu=\gamma^\mu$, and $\gb^0 \Omega_\mu^\dagger \gb^0=-\Omega_\mu$.
\end{lem}
\begin{proof}
Using (\ref{anticommutation}), $(\gb^0)^\dagger=\gb^0$ and $(\gb^i)^\dagger=-\gb^i$ we have $\gb^0 (\gb^a)^\dagger \gb^0=\gb^a$. The second relation follows as we can similarly show that $\gb^0(\Sigma^{ab})^\dagger \gb^0=-\Sigma^{ab}$.
\end{proof}

As a corollary, one can prove from this that $D_\mu \pb=\overline{D_\mu\psi}$, and the spinor conjugate of $\gamma^\mu D_\mu\psi$ reads
\begin{equation}
    \overline{\gamma^\mu D_\mu \psi}=(D_\mu \psi)^\dagger (\gamma^\mu)^\dagger \gb^0=(D_\mu \psi)^\dagger \gb^0\gb^0(\gamma^\mu)^\dagger \gb^0=\overline{D_\mu \psi}\gamma^\mu=D_\mu\pb\gamma^\mu,
\end{equation}
so we get the adjoint equation of Dirac equation: $D_\mu \pb \gamma^\mu-im\pb=0$.

\begin{remark}
We can now show that the right hand side $T_{\mu\nu}n^\mu n^\nu$, $T_{\mu i}n^\mu$ on the constraint equations \eqref{constraintTnn}, \eqref{constraintTni} can be determined by the initial data set $(\gbar_{ij},k_{ij},\psi_0)$. The original data does not contain time components of $g$, so we are not able to determine $n$, but we only need to use the fact that it is the unit normal vector. By the transform law under the local Lorentz transformation, one can see that $T_\mn$ is a $(0,2)$-tensor, meaning that it is invariant under the choice of orthonormal frame on $(\Sigma,\gbar_{ij})$, which forms a tetrad together with $n$. Now given 
an orthonormal basis $\{e_1,e_2,e_3\}$ with respect to $\overline g$ on the initial hypersurface, 
the Dirac equation can be written as $\gb^0 D_n \psi+\gb^i D_{e_i} \psi+im\psi=0$, so we have
\begin{equation}
    D_n \psi=-(\gb^0)^{-1} \gb^i (\d_{e_i}\psi-\frac 14 g(e_a,\dd_{e_i}e_b)\Sigma^{ab}\psi)-i(\gb^0)^{-1}m\psi,
\end{equation}
the term $g(e_a,\dd_{e_i} e_b)$ can be determined by $\gbar_{ij}$ and $k_{ij}$, which is the situation we want. The conjugate term is similar. We also have $\gamma_\mu n^\mu=-\gb^0$, again indepedent of $n$. 
\end{remark}

We can also define the covariant derivative of Gamma matrices. Recall the relation between Gamma matrices and the tetrad postulate, it is expected that Gamma matrices are covariantly constant in some sense. We define as in \cite{anomalies}
$$D_\nu \gamma^\mu=\dd_\nu\gamma^\mu +\Omega_\nu \gamma^\mu-\gamma^\mu\Omega_\nu.$$
Here $\dd$ only interacts with the index on $\gamma^\mu$ as a tensor (curved) index. 

Similar with the Minkowski case, while $\gamma^\mu$ is not a $4$-vector, we have $\overline\psi_1 \gamma^\mu \psi_2$ a $4$-vector, and $\overline\psi_1\psi_2$ is a scalar. Also, $\gamma^\mu D_\mu \psi$ is again a spinor. It is easy to verify the following Leibniz rules:
\begin{equation}
    \begin{split}
        \dd_\nu(\pb_1 \gamma^\mu \psi_2)&=(\d_\nu \pb_1)\gamma^\mu \psi_2+\pb_1 (\dd_\nu \gamma^\mu) \psi_2+ \pb_1 \gamma^\mu \d_\nu \psi_2  \\
        &=(\d_\nu \pb_1-\pb_1\Omega_\nu)\gamma^\mu \psi_2+\pb_1 (\dd_\nu \gamma^\mu+\Omega_\nu\gamma^\mu-\gamma^\mu\Omega_\nu) \psi_2+ \pb_1 \gamma^\mu (\d_\nu \psi_2+\Omega_\nu \psi_2)\\
        &=(D_\nu \pb_1)\gamma^\mu \psi_2+\pb_1 (D_\nu \gamma^\mu) \psi_2+ \pb_1 \gamma^\mu D_\nu \psi_2,
    \end{split}
\end{equation}
\begin{equation}\label{onesidedleibnizeq}
\begin{split}
    D_\mu(\gamma^\nu D_\rho \psi)&=\dd_\mu(\gamma^\nu D_\rho \psi)+\Omega_\mu \gamma^\nu D_\rho \psi=(\dd_\mu\gamma^\nu)D_\rho\psi+\gamma^\nu \dd_\mu D_\rho\psi+\Omega_\mu \gamma^\nu D_\rho \psi\\
    &=(\dd_\mu \gamma^\nu+\Omega_\mu \gamma^\nu-\gamma^\nu\Omega_\mu)D_\rho\psi+\gamma^\nu\Omega_\mu D_\rho\psi+\gamma^\nu \dd_\mu D_\rho\psi\\
    &=(D_\mu\gamma^\nu)D_\rho \psi+\gamma^\nu D_\mu D_\rho\psi.
\end{split}
\end{equation}

\begin{Prop}
The Gamma matrices are covariantly constant, i.e.\ $D_\nu \gamma^\mu=0$.
\end{Prop}

\begin{proof}
We calculate by the antisymmetry $\omega_{\mu ab}=\omega_{\mu ba}$:
\begin{equation}
    \begin{split}
        \Omega_\nu \gamma^\mu-\gamma^\mu \Omega_\nu&=-\frac 18(\omega_{\nu ab} [\gb^a,\gb^b]\gamma^\mu-\gamma^\mu \omega_{\nu ab}[\gb^a,\gb^b])=-\frac 18(\omega_{\nu ab} \gb^a\gb^b\gamma^\mu-\gamma^\mu \omega_{\nu ab}\gb^a\gb^b)\\
        &=\frac 18 \omega_{\nu ab} e_{\ c}^\mu(-\gb^a\gb^b\gb^c+\gb^c\gb^a\gb^b).
    \end{split}
\end{equation}
Using (\ref{anticommutation}) twice we have $-\gb^a\gb^b\gb^c+\gb^c\gb^a\gb^b=-2m^{ac}\gb^b+2m^{bc}\gb^a$. Therefore
\begin{equation}
    \Omega_\nu \gamma^\mu-\gamma^\mu \Omega_\nu=\frac 18 \omega_{\nu ab} e_{\ c}^\mu(-2m^{ac}\gb^b+2m^{bc}\gb^a)  =-\omega_{\nu\ b}^{\ a} e_{\ a}^\mu \gb^b.
\end{equation}
The result is now a consequence of the tetrad postulate (\ref{tetradpostulate}):
$$D_\nu \gamma^\mu=\dd_\nu\gamma^\mu +\Omega_\nu \gamma^\mu-\gamma^\mu\Omega_\nu=\dd_\nu \gamma^\mu-\omega_{\nu\ b}^{\ a} e_{\ a}^\mu \gb^b=(\dd_\nu e_{\ b}^\mu-\omega_{\nu\ b}^{\ a} e_{\ a}^\mu)\gb^b=0.\eqno\qed$$
\renewcommand{\qedsymbol}{}
\end{proof}

\begin{Cor}\label{spinorleibnizcor}
We have for Dirac spinor fields $\psi_1,\psi_2$ that
$$\dd_\nu(\pb_1 \gamma^\mu \psi_2)=(D_\nu \pb_1) \gamma^\mu\psi_2+\pb_1\gamma^\mu D_\nu\psi_2.$$
\end{Cor}

\subsection{Null frame, vector fields and Lie derivatives}
We define the null frame in the Minkowski spacetime:
$$L=\d_t+\d_r,\ \lb=\d_t-\d_r,\ S_1,S_2\text{ tangent to }\mathbb{S}^2,\ \langle S_i,S_j\rangle =\delta_{ij}.$$
Here $\d_r=\omega^i\d_i$, where $\omega_i=x_i/|x|$ (one should distinguish this with the spin connection, which is also denoted by $\omega$).
We use $\mathcal{T}=\{L,S_1,S_2\}$ to denote vectors tangential to the outgoing light cones, and $\db\in \mathcal{T}$ to represent the tangential derivatives. When orthogonality is not strictly required, we can also use the vector field $\db_i=\d_i-\omega_i\d_r$ instead of $S_1,S_2$, as the latter does not admit a global choice on the sphere.
We also define the Minkowski optical function $q=r-t$ with $r=|x|$.

The Minkowski wave operator has good commutation properties with the following vector fields:
$$\d_\mu,\ \Omega_{ij}=x_i\d_j-x_j\d_i,\ \Omega_{0i}=t\d_i+x_i\d_t,\ S=t\d_t+x^i\d_i.$$
We denote each of them by $Z^\iota$ with a 11-component index $\iota=(0,\cdots,1,\cdots)$. Then we can define the multi-index $I=(\iota_1,\cdots,\iota_k)$ to be of length $|I|=k$, and $Z^I=Z^{\iota_1}\cdots Z^{\iota_k}$. By $J+K=I$ we mean a sum of $J$ and $K$ over all order preserving partitions of $I$ into $J$ and $K$.

For the Minkowskian wave operator $\Box=m^{\a\b}\d_\a\d_\b$, we have $[\Box,Z]=0$ for $Z\in\{\d_\mu,\Omega_{ij},\Omega_{0i}\}$, and $[\Box,S]=2\, \Box$. One can also verify for the commutator of these vector fields $[Z_1,Z_2]=\sum_{|I|=1}c_{12I}Z^I$, and $[Z,\d_\mu]=c_\mu^\a \d_\a$, where $c_{12I}$ and $c_\mu^\a$ are constants. 

The following estimates of derivatives using vector fields are important:
\begin{equation}\label{dtoZ}
    |\db \phi|\lesssim (1+t+r)^{-1}\sum_{|I|\leq 1}|Z^I \phi|,\ |\d \phi|\lesssim (1+|t-r|)^{-1}\sum_{|I|\leq 1} |Z^I \phi|.
\end{equation}

We will use the following notation. Let $\mathcal{T}=\{L,S_1,S_2\}$, $\mathcal{U}=\{L,\lb,S_1,S_2\}$, $\l=\{L\}$ and $\mathcal S=\{S_1,S_2\}$. For any of these two families of vector fields $\mathcal{V}$ and $\mathcal{W}$ and a $(0,2)$-tensor $p_{\a\b}$, we define
\begin{equation*}
    |p|_{\mathcal V\mathcal W}=\sum_{V\in \mathcal V,\, W\in\mathcal W}|p_{\a\b} V^\a W^\b|,
\end{equation*}
\begin{equation*}
    |\d p|_{\mathcal V\mathcal W}=\sum_{U\in\, \mathcal U,\,  V\in \mathcal V,\, W\in\mathcal W}|(\d_\gamma p_{\a\b})U^\gamma V^\a W^\b|.
\end{equation*}
We will also use capital letters $A,B,\cdots$ to denote spherical tangent frame, e.g.\ $\slashed{\mathrm{tr}}\, p=\delta^{AB}p_{AB}=p_{S_1S_1}+p_{S_2S_2}$.

One can commute the quasilinear wave equations with the vector fields directly; however, as we will see later, if we use Lie derivatives instead then some fine structures can be preserved. 

\begin{Def}
The Lie derivative applied to a (s,r)-tensor $K$ is defined by 
\begin{multline}\lzr K^{\a_1\a_2\cdots \a_s}_{\b_1\b_2\cdots \b_r}=Z\, K^{\a_1\a_2\cdots \a_s}_{\b_1\b_2\cdots \b_r}-\d_\nu Z^{\a_1}\, K^{\nu\a_2\cdots \a_s}_{\b_1\b_2\cdots \b_r}-\cdots-\d_{\nu} Z^{\a_s}\, K^{\a_1\a_2\cdots \nu}_{\b_1\b_2\cdots \b_r}+\d_{\b_1}Z^\nu\,  K^{\a_1\a_2\cdots \a_s}_{\nu\b_2\cdots \b_r}\\
    +\cdots+\d_{\b_r} Z^\nu\, K^{\a_1\a_2\cdots \a_s}_{\b_1\b_2\cdots \nu}.
\end{multline}
\end{Def}

If $Z$ is one of the commuting vector fields, then $\d_\mu Z^\nu$ are constants and the following properties hold (see \cite{EV} for proof):

\begin{Prop}
For $Z=\d_\mu,\Omega_{ij},\Omega_{0i},S$, we have
\begin{equation}
    \lzr \d_{\mu_1}\cdots\d_{\mu_k} K^{\a_1\a_2\cdots \a_s}_{\b_1\b_2\cdots \b_r}=\d_{\mu_1}\cdots \d_{\mu_k} \lzr K^{\a_1\a_2\cdots \a_s}_{\b_1\b_2\cdots \b_r},
\end{equation}
and
\begin{equation}
    \lzr \d_\mu K^{\mu\cdots \a_s}_{\b_1\cdots \b_r}=\d_\mu \lzr K^{\mu\cdots \a_r}_{\b_1\cdots\b_s}.
\end{equation}
\end{Prop}

With the appearance of partial derivatives here, we are taking the Lie derivative of a quantity that is not a geometric quantities at least at first glance, so we may understand this formally. In fact, it can be viewed as taking the covariant derivative with respect to the Minkowski metric in the corresponding coordinate system.

Since $\d_\mu Z^\nu$ are constants, we also have the following equivalence, for any tensor field $K$:
\begin{equation}\label{equivalenceLiederivative}
    \sum_{|I|\leq k}|Z^I K|\lesssim \sum_{|I|\leq k}|\lzr^I K|\lesssim \sum_{|I|\leq k}|Z^I K|.
\end{equation}

\section{A second-order equation for the spinor field}\label{chaptersecondorderDiracderivation}
In this section, we raise and lower tensor indices with respect to the curved metric $g$.
\subsection{The square of Dirac operator}

We can apply one more Dirac operator to the Dirac equation to get a second order equation for components of $\psi$.

\begin{lem}\label{commutatorofcovariantderivative}
Let $D$ be the covariant derivative operator as in \eqref{eqdefinitionD}. Then $$[D_\mu,D_\nu]\psi=-\frac 14 R_{\mu\nu\sigma\delta} \gamma^\sigma \gamma^\delta \psi.$$
\end{lem}

\begin{proof}
Recall that $\omega_{\mu ab}=g(e_a,\dd_\mu e_b)=(e_a)_\rho (\dd_\mu e_b)^\rho$, which gives the expression  of $\dd_\mu (e_b)$ in terms of the spin connection, by applying the dual frame $(e^a)^\lambda$: $\dd_\mu (e_b)^\lambda=\omega_{\mu ab}(e^a)^\lambda$. We calculate as follows, using the torsion-free property of the connection $\dd$. Note that $D$ interacts with both spinor fields and tensorial indices, and we calculate
\begin{equation}\label{commutatorofDfirststep}
    \begin{split}
        [D_\mu,D_\nu]\psi&=(\nabla_\mu-\frac 14 \omega_{\mu ab}\Sigma^{ab})(\nabla_\nu\psi-\frac 14 \omega_{\nu cd}\Sigma^{cd}\psi)-(\nabla_\nu-\frac 14 \omega_{\nu cd}\Sigma^{cd})(\nabla_\mu\psi-\frac 14 \omega_{\mu ab}\Sigma^{ab}\psi)\\
        &=-\frac 14 (\dd_\mu \omega_{\nu ab}-\dd_\nu \omega_{\mu ab})\Sigma^{ab}\psi+\frac 1{16}(\omega_{\mu ab}\omega_{\nu cd}\Sigma^{ab}\Sigma^{cd}\psi-\omega_{\mu ab}\omega_{\nu cd}\Sigma^{cd}\Sigma^{ab}\psi)\\
        &=-\frac 14 g(e_a,\dd_\mu\dd_\nu e_b -\dd_\nu\dd_\mu e_b) \Sigma^{ab}\psi-\frac 14 (g(\dd_\mu e_a,\dd_\nu e_b)-g(\dd_\nu e_a,\dd_\mu e_b))\Sigma^{ab}\psi\\
        &\ \ \ \ \ \ \ \ \ \ \ \ \ +\frac 1{16} \omega_{\mu ab}\omega_{\nu cd}(\Sigma^{ab}\Sigma^{cd}-\Sigma^{cd}\Sigma^{ab})\psi.\\
    \end{split}
\end{equation}
Now $g(\dd_\mu e_a,\dd_\nu e_b)\Sigma^{ab}=\omega_{\mu ca}(e^c)^\lambda\, \omega_{\nu db}(e^d)_\lambda \Sigma^{ab}=m^{cd}\omega_{\mu ac}\, \omega_{\nu bd}\Sigma^{ab}$ (and similar for $\mu$ and $\nu$ interchanged). Also, using the commutator identity of $\Sigma^{ab}$ and $\Sigma^{cd}$ \eqref{commutatorsigmaabsigmacdeq}, along with the symmetry of $m^{ab}$ and antisymmetry of $\Sigma^{ab}$, we get
\begin{equation*}
\begin{split}
    \omega_{\mu ab}\omega_{\nu cd}(\Sigma^{ab}\Sigma^{cd}-\Sigma^{cd}\Sigma^{ab})&=2\, \omega_{\mu ab}\, \omega_{\nu cd}(m^{ac}\Sigma^{bd}-m^{ad}\Sigma^{bc}+m^{bd}\Sigma^{ac}-m^{bc}\Sigma^{ad})\\
    &=8m^{ac}\Sigma^{bd}\omega_{\mu ab}\, \omega_{\nu cd}.
\end{split}
\end{equation*}
Then it is not hard to see that the last two parts in (\ref{commutatorofDfirststep}) cancel out, and hence we have
\begin{equation}
    \begin{split}
        [D_\mu,D_\nu]\psi
        &=-\frac 14 g(e_a,\dd_\mu\dd_\nu e_b -\dd_\nu\dd_\mu e_b) \Sigma^{ab}\psi\\
        &=-\frac 14 (e_a)^\rho (e_b)^\delta R_{\mu\nu\rho\delta}\Sigma^{ab}\psi,
    \end{split}
\end{equation}
and this is the same as $-\frac 14 R_{\mu\nu\rho\delta}\gamma^\rho\gamma^\delta \psi$ by the definition of curved Gamma matrices and the anticommutation property (\ref{anticommutation}).
\end{proof}

\begin{lem}\label{curvaturecontractedwithgammematrices}
The identity $R_{\mu\nu\sigma\delta} \gamma^\nu\gamma^\sigma\gamma^\delta = 2 \gamma^\nu R_{\mu\nu}$ holds. As a result, we have the relation $$R_{\mu\nu\sigma\delta} \gamma^\mu\gamma^\nu\gamma^\sigma\gamma^\delta \psi =-2R\psi.$$
\end{lem}

\begin{proof}
We denote $A=R_{\mu\nu\sigma\delta} \gamma^\nu\gamma^\sigma\gamma^\delta$ and $B= \gamma^\nu R_{\mu\nu}$ respectively. Using the identity $g^{\mu\nu}I=-\frac 12 (\gamma^\mu\gamma^\nu+\gamma^\nu\gamma^\mu)$, the resulting identity $\gamma^\nu\gamma^\delta\gamma^\sigma=\gamma^\delta\gamma^\sigma\gamma^\nu-2g^{\nu\delta}\gamma^\sigma+2 g^{\sigma\nu}\gamma^\delta$, and the first Bianchi identity we have 
\begin{equation*}
    \begin{split}
        B&=-\frac 12\gamma^\mu R_{\mu\sigma\delta\nu}g^{\sigma\delta}=\frac 12\gamma^\nu R_{\mu\sigma\delta\nu}(\gamma^\sigma\gamma^\delta+\gamma^\delta\gamma^\sigma)=\frac 12 R_{\mu\sigma\delta\nu} \gamma^\nu\gamma^\delta\gamma^\sigma+\frac 12R_{\mu\sigma\delta\nu}\gamma^\nu\gamma^\sigma\gamma^\delta  \\
         &=-\frac 12(R_{\mu\nu\sigma\delta}+R_{\mu\delta\nu\sigma})\gamma^\nu\gamma^\delta\gamma^\sigma+\frac 12R_{\mu\sigma\delta\nu}\gamma^\nu\gamma^\sigma\gamma^\delta \\
         &=\frac 12 A+\left(\frac 12A-\frac 12 R_{\mu\delta\nu\sigma}(-2g^{\nu\delta}\gamma^\sigma +2\gamma^\delta g^{\sigma\nu})\right)+\left(\frac 12 A+\frac 12R_{\mu\sigma\delta\nu}(-2g^{\nu\sigma}\gamma^\delta+2\gamma^\delta g^{\delta\nu})\right) \\
         &=\frac 32 A-R_{\mu\sigma}\gamma^\sigma-R_{\mu\delta}\gamma^\delta = \frac 32 A-2B,
    \end{split}
\end{equation*}
so we get $3B=\frac 32 A$, i.e.\ $B=\frac 12 A$ which proves the first identity. Then we get
\begin{equation*}
        R_{\mu\nu\sigma\delta} \gamma^\mu\gamma^\nu\gamma^\sigma\gamma^\delta \psi=2\gamma^\mu \gamma^\nu R_{\mu\nu}\psi=-2R\psi
\end{equation*}
as required.
\end{proof}

We are now ready to state the Sch\"{o}dinger-Lichnerowicz formula, which is first given in \cite{schrodinger1932diracsches}. (We also refer the reader to \cite{kay2020editorial} for related histories)
\begin{Prop}
We have $\slashed{D}^2=-D^\mu D_\mu+\frac 14 R$, where $R$ is the scalar curvature and $\slashed D=\gamma^\mu D_\mu$ is the Dirac operator. Therefore, the solution of Dirac equation $\gamma^\mu D_\mu \psi=0$ satisfies the following equation, by applying the Dirac operator to the equation again:
$$D^\mu D_\mu\psi -\frac R4\psi=0.$$
\end{Prop}

\begin{proof}
By \eqref{onesidedleibnizeq} and Gamma matrices being covariantly constant, along with Lemma \ref{commutatorofcovariantderivative} and Lemma \ref{curvaturecontractedwithgammematrices}, we get 
\begin{equation*}
    \begin{split}
\gamma^\mu D_\mu(\gamma^\nu D_\nu \psi)&=\gamma^\mu\gamma^\nu D_\mu D_\nu \psi=\frac 12\gamma^\mu\gamma^\nu (D_\mu D_\nu\psi +D_\nu D_\mu\psi)+\frac 12 \gamma^\mu\gamma^\nu [D_\mu,D_\nu]\psi \\
&=-g^{\mu\nu}D_\mu D_\nu \psi-\frac 18 R_{\mu\nu\sigma\delta} \gamma^\mu\gamma^\nu\gamma^\sigma\gamma^\delta \psi=-D^\mu D_\mu \psi+\frac 14 R\psi,
    \end{split}
\end{equation*}
and the lemma follows.
\end{proof}

\subsection{Coupled system}The system can be derived by the variation of the action
$$S[\psi,e_{\ a}^\mu]=S_\text{grav}[e_{\ a}^\mu]+S_M[\psi,e_{\ a}^\mu]=\int R\, e\, dx+\int   \frac i2 (\overline{\psi} \gamma^\mu D_\mu \psi-D_\mu \overline{\psi}\gamma^\mu \psi)\, e\, dx,$$
where $e:=|\det(e_{\ \mu}^a)|$. One has $\sqrt{|g|}=e$ using $g_{\mu\nu}=m_{ab} e_{\ \mu}^a e_{\ \nu}^b$.

The equation of motion of the tetrad gives the Einstein equation, see for example \cite{yepez2011einsteinvierbein}. The energy-momentum tensor is determined by 
$$T^a_{\ \mu}=\frac 1e \frac {\delta S_M}{\delta e^\mu_{\ a}}=\frac i2 (\overline{\psi}\gb^a D_\mu\psi-D_\mu\overline{\psi} \gb^a \psi),$$
and consequently, the energy-momentum tensor in curved indices reads
\begin{equation}
T_{\mu\nu}=\frac 12(e^a_{\ \mu} m_{ab}T^b_{\ \nu}+(\mu\leftrightarrow\nu))=\frac i4 (\overline{\psi}\gamma_\mu D_\nu\psi-D_\nu\overline{ \psi} \gamma_\mu \psi)+\frac i4 (\overline{\psi}\gamma_\nu D_\mu\psi-D_\mu\overline{ \psi} \gamma_\nu \psi).
\end{equation}

If $\psi$ satisfies the Dirac equation $\gamma^\mu D_\mu \psi+im\psi=0$, then taking spinor conjugate we also have $D_\mu \overline{\psi} \gamma^\mu-im\psi=0$, so $\mathrm{tr}_g T=m\pb\psi$. As a result, if in addition the Einstein equation also holds, then the scalar curvarture $R=-m\pb\psi$. In particular, if $m=0$, then $T_{\mu\nu}$ is $g$-traceless, and hence $R=0$.

\paragraph{Equivalent system}We see from above that if $\psi$ solves the Dirac equation $\gamma^\mu D_\mu \psi=0$, and regularity permits, it also satisfies the equation $D^\mu D_\mu \psi-m^2 \psi=-\frac 14 m\pb\psi$. We can consider the initial value problem with this equation coupled with the Einstein equation instead, by determing the initial time derivative of $\psi$ using the Dirac equation. On the other hand, one can use a similar argument as in \cite{psarelli2005maxwell} to show that the solution of the new system solves the Einstein-Dirac system: 
if $\psi$ solves the second order equation, and the metric satisfies the Einstein equation, then we know that $(\gamma^\mu D_\mu-im)(\gamma^\nu D_\nu \psi+im\psi)=0$, so $\gamma^\nu D_\nu \psi+im\psi$ is a solution of Dirac equation (with mass $-m$). Recall the way we set time derivative of $\psi$ at $t=0$, and we directly have $(\gamma^\nu D_\nu \psi+im\psi)|_{t=0}=0$. Now since we have got the metric $g$, the solution of the Dirac equation under a fixed metric should be unique (otherwise, we can get contradiction for instance by arguing that they will satisfy the same second order equation), and hence $\gamma^\nu D_\nu \psi+im\psi=0$, i.e.\ $\psi$ solves the Dirac equation. Therefore, to prove the local wellposedness of the Einstein-Dirac system we can instead prove it for the Einstein equation coupled with the second-order equation of $\psi$.

\paragraph{Equations in wave coordinates}
The wave (harmonic) coordinate condition is a gauge condition saying $\Gamma^\gamma:=g^{\a\b}\Gamma_{\a\b}^{\ \ \,\gamma}=0$, or equivalently as in \eqref{wcc}.
From \cite{03} we know that under the wave coordinate condition, the Ricci curvature can be written as 
\begin{equation}
    R_{\mu\nu}=-\frac 12\widetilde{\Box}_g g_{\mu\nu}+\frac 12 F_{\mu\nu}(g)(\d g,\d g)
\end{equation}
Here $F_{\mu\nu}(g)(\d g,\d g)=P(\d_\mu g,\d_\nu g)+Q_{\mu\nu}(\d g,\d g)+G(h)(\d g,\d g)$, where $Q_{\mu\nu}$ is a combination of standard null forms, $G_{\mu\nu}$ is a cubic term (vanishes when $h=g-m=0$), and $$P(\d_\mu g,\d_\nu g)=\frac 14 m^{\a\b}\d_\mu g_{\a\b} \, m^{\rho\sigma} \d_\nu g_{\rho\sigma}-\frac 12 m^{\a\b}m^{\rho\sigma}\d_\mu g_{\a\rho}\, \d_\nu g_{\b\sigma}.$$
We also have $\Box_g \psi=g^{\a\b} \dd_\a \dd_\b \psi=g^{\a\b}\d_\a\d_\b\psi=\widetilde{\Box}_g \psi$ under the coordinate condition. For terms involving the connection term $\Omega_\mu$, we move them to the right hand side. Then the Einstein equation and the second-order equation derived from the Dirac equation in wave coordinates can be written as
\begin{equation*}
    \begin{split}
        \widetilde{\Box}_g h_{\mu\nu} &=F_\mn(g)(\d g,\d g)+ T_{\mu\nu}-\frac 12 g_\mn \mathrm{tr}_g\, T,\\
        \widetilde{\Box}_g \psi &=m^2\psi+\fp(g)(\omega,\psi,\d \psi)+\frac 14 R\psi,
    \end{split}
\end{equation*}
where
\begin{equation}\label{Fpsi}\fp(\omega,\psi,\d \psi)=\frac 12 g^{\mu\nu} \omega_{\mu ab}\Sigma^{ab} \d_\nu \psi + \frac 14 \dd_\mu (\omega_{\nu ab})\Sigma^{ab} \psi +\frac 1{16} \omega_{\mu ab}\, \omega_{\nu cd} \Sigma^{ab}\Sigma^{cd} \psi+\frac 14 R\psi.\end{equation}
Note again that the term the scalar curvature $R$ and $-\mathrm{tr}_g\, T$ can be expressed as $-m\pb\psi$ when the system holds.






\section{Change of coordinates}\label{ChapterSchwarzschildcoordinates}
We will fix a special coordinate in this section, based on the wave coordinate, and use this new coordinate for the remaining parts of this work. After introducing this coordinates, we always raise and lower tensorial indices with respect to the Minkowski metric in the new coordinates.

In order to distinguish components in different coordinate system, we will sometimes use primed index $\a',\b',\cdots$ for the components in \textit{original} coordinates.
\subsection{Asymptotic Schwarzschild coordinates}
We introduce the asymptotic Schwarzschild coordinates:
$$\widetilde{t}=t, \ \ r^*=r+M\chi(\textstyle\frac r{t+1}) \ln r, \ \ \x=r^* \omega,$$
where $\chi$ is defined as in \eqref{h0}.
We denote the transition matrix by $A^\a_{\a'}=\frac{\d\x^\a}{\d x^{\a'}}$, and we can compute the components:
\begin{equation}\label{transitionmatrices}
\begin{split}
    \d_{\a'}\x^0&=\delta^0_{\a'},\\
    \d_0 \x^i&=\delta^i_0-Mr^{-1}\ln r\, \omega^i(\textstyle\frac r{t+1})^2\chi',\\
    \d_j \x^i&=\delta^i_j+Mr^{-1}\ln r((\delta^i_j-\omega_j\omega^i)\chi+\omega_j\omega^i\textstyle\frac r{t+1}\chi')+Mr^{-1}\omega_j\omega^i\chi.
\end{split}
\end{equation}

In the new coordinates the derivatives read
$$\dt_t=\d_t-\d_t (r^*)\d_{r^*},\ \ \dt_r=\d_{r^*}=\frac 1{\d_r (r^*)} \d_r,\ \ \widetilde{\slashed \d}_i=\frac r{r^*}\slashed \d_i,\ \ \dt_i=\omega_i \dt_r+\widetilde{\slashed \d}_i.$$
We also define the corresponding null frame $\lt=\dt_t+\dt_r$, $\ltb=\dt_t-\dt_r$, $\widetilde{S}_i=\frac{r}{r^*}S_i$, and the commuting vector fields: $\z\in \{\dt_\mu,\ \widetilde{\Omega}_{ij}=\x_i\dt_j-\x_j\dt_i,\ \widetilde{\Omega}_{0i}=\x_i\dt_t+t\dt_i,\ \widetilde{S}=t\dt_t+\x^i\dt_i\}$. We will also use $\Tc$, $\Uc$, $\cdots$ instead of $\mathcal T$, $\mathcal U$, $\cdots$ defined before, to denote the counterpart in the new coordinates.

We use $\mt$ to denote the Minkowski metric tensor in the original coordinates, i.e.\ $\mt^{\a\b}=A^\a_{\a'}A^\b_{\b'}m^{\a'\b'}$, and $\gt$ to denote our dynamic Lorentzian metric with components in the new coordinates; similarly we can define $\widetilde H_0$, $\hti^0$, $\Hi$ and $\hti^1$.

Let $\mt_0^\mn=\mt^\mn+\widetilde H_0^\mn$, 
so now $$\gt^\mn=\mt_0^\mn+\Hi^\mn.$$
Similarly we define $\mt^0_\mn=\mt_\mn+\hti^0_\mn$. We also have the Minkowski metric in the new coordinates: $\mh=-d\, \widetilde{t}\, ^2+\sum_i(d\x^i)^2$, and now the perturbation of the metric from the ``Minkowski metric" in the new coordinates is $\gt-\mh$. 

Using the expression of the transition matrices \eqref{transitionmatrices}, one can prove the following estimates (see \cite{KLMKG21} for proof):



\begin{Prop}\label{estimatetransitionmatricesprop}
We have the estimate for transition matrices:
\begin{equation*}
    |\d^I(A^{\a}_{\a'}-\delta_{\a'}^\a)|+|\d^I(A^{\a'}_\a-\delta^{\a'}_\a)|\lesssim_I \frac {M\ln (1+r)}{r^{1+|I|}}.
\end{equation*}
As a consequence we also have 
\begin{equation*}
    |\d^I Z^J(A^{\a}_{\a'}-\delta_{\a'}^\a)|+|\d^I Z^J(A^{\a'}_\a-\delta^{\a'}_\a)|\lesssim_I \frac {M\ln (1+r)}{r^{1+|I|}}.
\end{equation*}
And if $\frac r{1+t}\leq \frac 14$, these terms are zero. Moreover, the same estimates hold if we replace $\d^I$, $Z^J$ by $\dt^I$, $\z^J$.
\end{Prop}

\begin{Prop}
For $\mt_0$ and $\mt^0$ we have the following estimate
\begin{equation}\label{dm0higherorder}
    |\dt^I \z^J \mt_0|+|\dt^I \z^J \mt^0|\lesssim \frac{M\ln(1+r)}{(1+t+r)^{1+|I|}}.
\end{equation}
\end{Prop}

\begin{proof}
By definition $\d^I Z^J (\mt_\mn+\hti^0_\mn)=\d^I Z^J (m_{\a'\b'}A_\mu^{\a'}A_\nu^{\b'}+h^0_{\a'\b'}A_\mu^{\a'}A_\nu^{\b'})$. Now in view of the explicit expression of $h^0$, we get the estimate $|\d^I Z^J \mt^0|\lesssim \frac{M\ln(1+r)}{(1+t+r)^{1+|I|}}$ using Proposition \ref{estimatetransitionmatricesprop}. It is also not difficult here to show that same estimate hold if we replace $\d$ and $Z$ by $\dt$ and $\z$. The estimate for $\mt_0$ follows similarly.
\end{proof}

\begin{Prop}Let $\gh$ be the Christoffel symbol of the metric $\mt$ in the new coordinates, i.e.\ $\gh_{\mu\nu}^{\ \ \,\rho}=\frac 12 \mt^{\lambda\rho}(\dt_\mu \mt_{\nu\lambda}+\dt_\nu \mt_{\lambda \mu}-\dt_\lambda \mt_{\mu\nu})$. Then
\begin{equation}\label{boundforChristoffelhat}
    |\z^I \gh|\lesssim \frac{M(1+\ln (1+r))}{(1+r+t)^2}.
\end{equation}
We also have
\begin{equation}\label{decayofSchwarzschildpart}
    |\z^I(\mt^{\mu\nu}-\mh^{\mu\nu})|\lesssim \frac{M(1+\ln(1+r))}{1+r+t};
\end{equation}
and same estimate follows if we consider $\mt_{\mu\nu}$ and $\mh_{\mu\nu}$.
\end{Prop}


\begin{proof}
This also follows from Proposition \ref{transitionmatrices}.
\end{proof}

The following lemma in \cite{KLMKG21} shows that $\mt_0^{\mu\nu}\sim (1+\frac {M\cht}r)\mh^{\mu\nu}$.

\begin{lem}\label{m0lemma}
We have the decomposition
\begin{equation}\label{m0}
    \mt_0^{\alpha\beta}=\kappa_0\mh^{\alpha\beta}+\kappa_1 \slashed S^{\alpha\beta}+\kappa_2 \omega^\alpha\omega^\beta+\kappa_3 i_+^{\alpha\beta},
\end{equation}
where 
$$\textstyle\kappa_0=1+\frac {M\cht}r,\ \k_1=2\cht \frac{M\ln r}r\chi_1(\frac r{t+1},\frac{M\ln r}r,\frac Mr)-2\cht\frac Mr,\ \k_2=(\frac {M\cht}r)^2(1+\frac {M\cht}r),\ \k_3=\frac{M\ln r}r,$$
and $i_+^{\mu\nu}=\cht'(\frac r{t+1})\chi^{\mu\nu}(\frac r{t+1},\frac{M\ln r}r,\frac Mr,\omega)$ is zero close to the light cone. Moreover, if we let $\k_0^I$ be the function satisfying $\lz^I(\k_0 \mh)=\k_0^I\mh$, then
\begin{equation}
    (\lz^I \mt_0)^{\mu\nu}=\k_0^I\mh^{\mu\nu}+\slashed S^{I\, \mu\nu}+(\lz^I R)^{\mu\nu},
\end{equation}
where \begin{equation}
\begin{split}
    &\textstyle\k_0^I\lesssim 1,\ |\slashed S^I_{\Ut\widetilde V}|\lesssim \frac{M\ln(1+t+r^*)}{1+t+r^*}\chi_1(\frac r{t+1}),\ |\slashed S^I_{\lt\Ut}|\lesssim \frac{(1+|q^*|)\ln(1+t+r^*)}{(1+t+r^*)^2}\chi_1(\frac r{t+1}),\\ &\textstyle|(\lz^I R)_{\Ut\widetilde V}|\lesssim \frac{M^2(\ln(1+r^*))^2}{(1+t+r^*)^2}\chi_1(\frac r{t+1})+\frac{M\ln(1+r)}{1+t+r^*}\chi_2'(\frac r{t+1}),
    \end{split}
\end{equation}
where $\chi_1$ and $\chi_2'$ are functions bounded by 1 with support on $[\frac 14,\infty)$ and $[\frac 14,\frac 34]$ respectively.
\end{lem}

\begin{proof}
Use \eqref{transitionmatrices} we get
\begin{equation*}
    \mt_0^{00}=-1-\textstyle\frac{M\cht}{r},\ \ \ \mt_0^{0i}=\textstyle \frac 1r(\frac r{1+t})^2(-1-\frac{M\cht}{r})\omega_i\cht'M\ln r.
\end{equation*}
The expression of $\mt_0^{00}$ corresponds to $\k_0 \mh$, and $\mt_0^{0i}$ corresponds to $\k_3 i_+$. For $\mt_0^{ij}=m_0^{\a'\b'}A^i_{\a'}A^j_{\b'}$, we have
$$\mt_0^{ij}=m_0^{00} A_0^i A_0^j+(1-\textstyle\frac{M\cht}{r})\delta^{kl}A^i_k A^j_l,$$
and the second term equals by \eqref{transitionmatrices}
\begin{equation*}
    \begin{split}
        (1-\textstyle\frac{M\cht}{r})&(1+\textstyle\frac{M\cht}r)^2\delta^{ij}+(1-\textstyle\frac{M\cht}r)\left(2(1+\textstyle\frac{M\cht}r)(\textstyle\frac{(M\ln r-M)\cht}{r})+(\textstyle\frac{(M\ln r-M)\cht}{r})^2\right)(\delta^{ij}-\omega^i\omega^j)\\
        &=(1-\textstyle\frac{M\cht}{r})(1+\textstyle\frac{M\cht}r)^2\delta^{ij}+(1-\textstyle\frac{M\cht}r)\left((1+\textstyle\frac {M\cht\ln r}{r})^2-(1+\frac{M\cht}r)^2\right)(\delta^{ij}-\omega^i\omega^j)\\
        &=\k_1(\delta^{ij}-\omega^i\omega^j)+(1+\textstyle\frac {M\cht}r)(\delta^{ij}-\omega^i\omega^j)+(1-\textstyle\frac{M\cht}r)(1+\frac{M\cht}r)^2\omega^i\omega^j\\
        &=\k_0 \delta^{ij}+\k_1 \slashed{S}^{ij}+(\textstyle\frac{M\cht}{r})^2(1+\textstyle\frac{M\cht}r)\omega^i\omega^j.
    \end{split}
\end{equation*}
The part $m_0^{00}A^i_{0}A^j_{0}$ is a term with $\cht'$ in view of \eqref{transitionmatrices}, and appears in $\k_3 i_+$.

For higher order, we first notice that $\lzh^I(\k_0 \mh)^{\mu\nu}=(\lzh^I\k_0)\mh^{\mu\nu}$, so $\lz^I(\k_0\mh)^{\mu\nu}=\k_0^I \mh^{\mu\nu}$ for some bounded function $\k_0^I$, using that $\lz^I \mh^{\mu\nu}=c_{\z}\mh^{\mu\nu}$. For spherical tangent part, we use the identity
$$\slashed{S}^{\mu\nu}=r^{-2}\sum_{i<j}\widetilde{\Omega}_{ij}^\mu\widetilde\Omega_{ij}^\nu,$$
so applying $\lz^I$ on $\k_1 \slashed{S}^{\mu\nu}$ we either have the derivative on $\k_1 r^{-2}$, which is straightforward to bound, or on $\widetilde{\Omega}_{ij}^\mu\widetilde\Omega_{ij}^\nu$, which gives factors like $\lz^I \widetilde{\Omega}_{ij}^\mu$. Recall that $\lz X^\mu=[\z,X]$, so using iteratively $[\z_1,\z_2]=\sum_{|I|=1}c_{12I}\z^I$ where $c_{12I}$ are $-1$, $0$ or $1$, we obtain the estimate when $\frac r{t+1}>1/4$ that $|(\lz^J \widetilde{\Omega}_{ij})_\lt|\lesssim 1+|t-r^*|$, as all commuting vector fields satisfiy the bound $|\z_\lt|\lesssim 1+|t-r^*|$ in this region,
and the estimate follows.
\end{proof}

\subsection{Generalized wave coordinate condition}
Since the wave coordinate condition holds for $g$ in the original coordinates, we have for a scalar field $\phi$ that
$$\gt^{\mu\nu}\dt_\mu\dt_\nu\phi-\gt^{\mu\nu}\gat_{\mu\nu}^{\ \ \rho}\dt_\rho\phi=\Box_{\gt}\phi=\Box_g \phi=g^{\a'\b'}\d_{\a'}\d_{\b'}\phi.$$
Here $\gat$ is the Christoffel symbol associated with $\gt$ in the new coordinates. We also have the relation using the transition matrix, noticing that the Levi-Civita connection associated with the metric $m$ (hence also $\mt$) is $\d$:
$$A^{\a'}_\mu A^{\b'}_\nu \d_{\a'}\d_{\b'}\phi=\dt_\mu\dt_\nu \phi-\gh_{\mu\nu}^{\ \ \,\rho}\dt_\rho\phi.$$
Combining the two equations we get
\begin{equation}\label{generalizedwcc}
    \gt^{\mu\nu}\gat_{\mu\nu}^{\ \ \,\rho}=\gt^{\mu\nu}\gh_{\mu\nu}^{\ \ \,\rho}.
\end{equation}
This is the generalized wave coordinate condition which we will use. This is as good as the wave coordinate condition, in the sense that the right hand side is only dependent on the metric itself, not on its derivatives.

\subsection{Equations in generalized wave coordinates}
As in \cite{KLMKG21}, the reduced Einstein equation in the new coordinates becomes
\begin{equation}
    \gt^{\mu\nu}\dh_\mu \dh_\nu \hti_{\rho\sigma}=\widetilde{F}_{\rho\sigma}(\gt)(\dh \hti,\dh\hti)+\widetilde{{T}}_{\rho\sigma}-\frac 12 \mathrm{tr}\, T \gt_{\rho\sigma},
\end{equation}
where $\dh$ is the covariant derivative associated with $\mt$.
We want to establish the equation of $\hi$. Recall that $\frac 1r$ solves the Minkowskian wave equation, so $\mt^\mn \dh_\mu \dh_\nu \hti^0$ is nonzero only when $\chi'$ is nonzero. We have
\begin{equation}\label{equationofhtoh1}\textstyle
    \gt^{\mu\nu}\dh_\mu \dh_\nu \hti_{\rho\sigma}=\gt^{\mu\nu}\dh_\mu \dh_\nu \hti^1_{\rho\sigma}+M\chi'_{\rho\sigma}(\frac{r}{1+t},\frac Mr,\frac{M\ln r}r,\frac xr)r^{-3}+M\chi_{\rho\sigma\eta\xi}(\frac{r}{1+t},\frac Mr,\frac{M\ln r}r,\frac xr)r^{-3}\Hti_1^{\eta\xi},
\end{equation}
\begin{multline}\label{RHSinhtoh1}\textstyle
    \widetilde{F}_{\rho\sigma}(\gt)(\dh \hti,\dh\hti)=\widetilde{F}_{\rho\sigma}(\gt)(\dh \hti^1,\dh\hti^1)+M^2 \chi_{\rho\sigma}(\frac{r}{1+t},\frac Mr,\frac{M\ln r}r,\frac xr,\gt)r^{-4}\\
    \textstyle+M\chi_{\rho\sigma}^{\zeta\eta\xi}(\frac{r}{1+t},\frac Mr,\frac{M\ln r}r,\frac xr,\gt)r^{-2}\dh_\zeta \hti^1_{\eta\xi}.
\end{multline}
Here $\chi_{\a\b}(s,\cdot)$, $\chi_{\a\b\y\rho}(s,\cdot)$ and $\chi_{\a\b}^{\mu\nu\rho}(s,\cdot)$ are functions supported in $\{s\geq 1/4\}$, and $\chi'_{\a\b}$ is supported in $\{1/4\leq s\leq 3/4\}$.

These error terms are clearly in control and we will not go over estimating them. We also have
$$\dt^I \gh \sim M\chi r^{-2-|I|}\ln r,$$
This helps us change from $\gt^{\mu\nu}\dh_\mu \dh_\nu \hti^1_{\rho\sigma}$ to $\gt^{\mu\nu}\dt_\mu\dt_\nu \hti_{\rho\sigma}$, and $\widetilde{F}_{\rho\sigma}(\gt)(\dh \hti^1,\dh\hti^1)$ to $\widetilde{F}_{\rho\sigma}(\gt)(\dt \hti^1,\dt \hti^1)$, with the difference error terms which are in control. Therefore using the reduced wave operator $\widetilde\Box_{\gt}=\gt^{\mn} \dt_\mu\dt_\nu$, the equation of $\hi$ can be written as
\begin{equation}\label{zeroordereqforh1}
    \widetilde{\Box}_{\gt}\, \hi_{\mu\nu}=\widetilde{F}_{\mu\nu}(\gt)(\dt \hi,\dt\hi)+\widetilde{T}_{\mu\nu}+{R}^{mass}_{\mu\nu}+{R}^{cov}_{\mu\nu},
\end{equation}
where the last two terms are in good control, and we will explain the notation in next subsection.


For the spinor field $\psi$ we have $\Box_{\gt}\psi=\gt^{\mu\nu}\dt_{\mu}\dt_\nu \psi-\gt^{\mu\nu}\gh_{\mu\nu}^{\ \ \rho} \dt_\rho \psi$ by the generalized wave coordinate condition. The second here on the right hand side is a covariant error term. Then the equation of $\psi$ can be written as
\begin{equation}\label{zeroorderwaveeqforpsi}
    \widetilde{\Box}_{\gt}\psi=\widetilde F^\psi(\gt)(\omegat,\dt\psi,\psi)+{R}^{cov}.
\end{equation}
For terms involving spin connection terms, we will deal with them in next section.

\subsection{Error terms}
There are three types of error terms that behave much better and can be ignored in the analysis.

\paragraph{Mass Error}Associated with the mass, produced when substracting the Schwarzschild part. The terms with $M$ in (\ref{equationofhtoh1}) and (\ref{RHSinhtoh1}) are mass error terms. We denote these terms by $R^{mass}$.
\paragraph{Covariant Error}The error terms when we change from $\dh$ to $\dt$. This is because we have good estimates of the related Christoffel symbol $\gh$. We denote these terms by $R^{cov}$.
\paragraph{Cubic Error}Terms with an extra factor of decay. The extra factor makes the term much better. This may include an extra factor of $\psi$, $\gt-\mh$, derivatives of them, and terms comparable with them (e.g. derivatives of tetrad, as we will see). We denote these terms by $R^{cube}$.

All these terms remain good after we apply vector fields, and we will simply add the index of vector fields on it, e.g.\ $R^{mass\ I}$.

For quadratic terms that contain at least one tangential derivative, such as the standard null forms, we have much better control and we sometimes denote it by $R^{tan}$.

\section{Choice of tetrad}\label{constructionoftetradsection}

The Gamma matrices adapted to the metric are determined by the tetrad by $\gamma^\mu=\bar{\gamma}^a (e_a)^\mu$. The choice of a tetrad, along with our choice of coordinates, forms a choice of gauge for our system. We present one choice of the tetrad here. We remark that there are other choices that may relate some components of the tetrad to some components of the metric, and we leave the details to Appendix \ref{appendixtetrad}.

For the sake of generality, we prefer to state this in the original notation in $(\mathbb{R}^4,g)$.

\begin{lem}[Gram-Schmidt orthogonalization]\label{GS}
We can do the standard Gram-Schmidt orthogonalization to the set $\{\d_1,\d_2,\d_3,\d_0=\d_t\}$ with respect to the metric $g$ to get a tetrad $\{e_a\}$, and the following estimates hold, where $h_{\mu\nu}=g_{\mu\nu}-m_{\mu\nu}$ is the difference between $g$ and the Minkowski metric:
\begin{enumerate}
    \item $|(e_a)^\mu-(\d_a)^\mu|\lesssim |h|+O(h^2)$;
    \item $|\d_\a (e_a)^\mu|\lesssim |\d_\a h|+O(h\cdot \d h)$;
    \item The above two estimates also hold with vector fields: $|Z^I (e_a)^\mu|\lesssim |Z^I h|+\sum_{J+K=I}O(Z^J h\cdot Z^K h)$, and $|\d_\a Z^I (e_a)^\mu|\lesssim |\d_\a Z^I h|+\sum_{J+K=I} O(\d Z^J h\cdot Z^K h)$.
    \item If $g$ satisfies the reduced Einstein's equation, then the tetrad satisfies the equation $\widetilde\Box_{g}(e_a)^\mu=F_a^\mu(g,\d g,\psi,\d \psi)$ for functions $F_a^\mu$ at least quadratic;
    \item Same estimates hold if one replaces $e_a^\mu$ by $\gamma^\mu$.
\end{enumerate}
\end{lem}

\begin{proof}
We perform a standard Schmidt orthogonalization process to the Minkowski orthonormal frame $\{\d_1,\d_2,\d_3,\d_0=\d_t\}$. We get 
\begin{equation*}
    \begin{split}
        \textstyle e_1&\textstyle=\frac{\d_1}{\sqrt{g_{11}}}= \frac{\d_1}{\sqrt{1+h_{11}}},\\
        \textstyle v_2&\textstyle=\d_2-\mathrm{proj}_{\d_1}(\d_2)=\d_2-\frac{g_{12}}{g_{11}}\d_1,\  e_2=\frac{v_2}{\sqrt{g(v_2,v_2)}}=\frac 1{\sqrt{1+(\frac{h_{12}}{1+h_{11}})^2}} \d_2 -\frac{h_{12}}{1+h_{11}} \frac 1{\sqrt{1+(\frac{h_{12}}{1+h_{11}})^2}}\d_1,\\
        \textstyle v_3&=\textstyle\d_3-\mathrm{proj}_{\d_1}(\d_3)-\mathrm{proj}_{e_2}(\d_3)=\d_3-\frac{h_{13}}{1+h_{11}} \d_1 - \frac{h_{23}-\frac{h_{12}}{1+h_{11}} h_{13}}{1+(\frac{h_{12}}{1+h_{11}})^2} (\d_2-\frac{h_{12}}{1+h_{11}} \d_1)\\
        \textstyle e_3&\textstyle=\frac{v_3}{\sqrt{g(v_3,v_3)}}=\cdots,\\
        \textstyle v_0&\textstyle=\d_t-\mathrm{proj}_{\d_1}(\d_0)-\mathrm{proj}_{e_2}(\d_0)-\mathrm{proj}_{e_3}(\d_0)=\cdots,\ e_0=\frac{v_0}{\sqrt{g(v_0,v_0)}}=\cdots
    \end{split}
\end{equation*}
We omit the precise computation of the last two since it is rather tedious but essentially the same. It is not hard to see using Taylor expansion that the difference between new frame and the original one is comparable to $|h|$. Also, if we take coordinate derivative of frame components, we get the corresponding derivative on the metric, i.e.\ $|\d_\mu (e_a)^\nu| \lesssim |\d_\mu h|+|h||\d h|$. Other estimates also follow from the expression.
\end{proof}

\paragraph{Tetrad in the new coordinates}
Now consider our new coordinates, 
where the Minkowski metric here is $\mh$. 
Let $\{\e_a\}$ this tetrad associated with $\gt$, from the way of construction above. Then the properties in Lemma \ref{GS} holds in our coordinates (i.e.\ add tilde to those vectors), with $h$ replaced by $\gt-\mh$. 

We have the following estimates for the spin connection coefficients. We will discuss higher order cases in Section \ref{higherorderomega}.

\begin{lem}\label{omeganull}
The spin connection coefficients satisfy the following estimate: $$|\omegat_{\mu ab}|\lesssim |\dt \gt|+|\dt (\e_a)|,$$
Moreover, we have for the $\lt$ components
$$|\lt^\mu \omegat_{\mu ab}|\lesssim |\dtb \gt|+|\dt \gt|_{\lc\Tc}+|\dtb (\e_a)|.$$
\end{lem}

\begin{proof}
Recall that $\omegat_{\mu ab}=\gt(\e_a,\dd_\mu \e_b)=\gt_{\alpha\beta} (\e_a)^\alpha \dd_\mu (\e_b)^\beta=\gt_{\alpha\beta} (\e_a)^\alpha \dt_\mu (\e_b)^\beta+\gt_{\alpha\beta}(\e_a)^\alpha \gat_{\mu\nu}^{\ \ \,\beta}(\e_b)^\nu$. Then the first estimate follows in view of (\ref{Christoffelsymbol}), and $|\gt|,|\e|\lesssim 1$.

Now contracting the identity with $\lt^\mu$ we get $\lt^\mu\omegat_{\mu ab}=\gt_{\alpha\beta} (\e_a)^\alpha \lt (\e_b)^\beta+\gt_{\alpha\beta}(\e_a)^\alpha \lt^\mu \gat_{\mu\nu}^{\ \ \,\beta}(\e_b)^\nu$. The first term is comparable with $\lt (\e_b)$. For the second term we recall \eqref{Christoffelsymbol} to get $$\gt_{\alpha\beta}(\e_a)^\alpha \lt^\mu \gat_{\mu\nu}^{\ \ \,\beta}=\lt^\mu (\e_a)^\a  \frac 12 (\dt_{\mu} \gt_{\nu\a}+\dt_{\nu} \gt_{\mu\a}-\dt_{\a} \gt_{\mu\nu})=(\e_a)^\a  \frac 12 ((\lt \gt)_{\nu\a}+\dt_{\nu} \gt_{\lt\a}-\dt_{\a} \gt_{\lt\nu}).$$
We expand this in the null frame. Since $\lt^\mu (\dt_{\mu} \gt_{\nu\a}+\dt_{\nu} \gt_{\mu\a}-\dt_{\a} \gt_{\mu\nu}){\ltb}^\nu\ltb^\a=(\lt \gt)_{\ltb\ltb}$, and other components can clearly be controlled by $|\dtb \gt|+|\dt \gt|_{\lc\Tc}$, we get the desired estimate.
\end{proof}



The second order equation of $\psi$ reads in the massless case
\begin{equation}\label{secondordereqpsi}
    \widetilde{\Box}_{\gt} \psi=\frac 12 \gt^\mn \omegat_{\mu ab}\Sigma^{ab}\dt_\nu \psi+\frac 14 \gt^\mn \dd_\mu (\omegat_{\nu ab})\Sigma^{ab}\psi+\frac 1{16}\gt^\mn\omegat_{\mu ab}\,\omegat_{\nu ab}\Sigma^{ab}\Sigma^{cd}\psi.
\end{equation}
There is a potential loss of derivative from the term $\dd_\mu (\omegat_{\nu ab})\Sigma^{ab}\psi$, as $\omegat_{\mu ab}$ behaves like the first derivative of the metric. We deal with this in the following lemma.

\begin{lem}[Avoiding loss of derivative] \label{avoidloss}
Suppose the generalized wave coordinate condition \eqref{generalizedwcc} holds. Then
\begin{multline}
    \gt^{\mu\nu} \dd_\mu (\omegat_{\nu ab})=\gt^{\mu\nu} (\dd_\mu (\e_a))_\rho (\dd_\nu (\e_b))^\rho+\gt^{\mu\nu}(\e_a)_\rho \gat_{\mu\beta}^{\ \ \rho} \dd_\nu (\e_b)^\beta-\gt^{\mu\nu} (\e_a)_\rho \gat_{\mu\nu}^{\ \ \alpha} \dd_\alpha (\e_b)^\rho\\ -(\gt^{\delta\rho}R_{\sigma\delta} + \gt^{\mu\nu}\gat_{\mu\lambda}^{\ \ \rho}\, \gat_{\sigma\nu}^{\ \ \lambda}-\gt^{\mu\nu}\gat_{\sigma\lambda}^{\ \ \rho}\, \gat_{\mu\nu}^{\ \ \lambda}+(\dt_\sigma \gt^{\mu\nu})\gat_{\mu\nu}^{\ \ \,\rho}+\dt_\sigma (\gt^{\mu\nu}\gh_{\mu\nu}^{\ \ \rho}))(\e_a)_\rho(\e_b)^\sigma \\
    +\gt^{\mu\nu} (\e_a)_\rho \gat_{\nu\sigma}^{\ \ \,\rho}\dt_\mu (\e_b)^\sigma+(\e_a)_\rho \widetilde{\Box}_{\gt} (\e_b)^\rho. 
\end{multline}
\end{lem}

\begin{proof}
We have $$\gt^{\mu\nu} \dd_\mu (\omegat_{\nu ab})=\gt^{\mu\nu} \dd_\mu((\e_a)_\rho (\dd_\nu (\e_b))^\rho)=\gt^{\mu\nu} (\dd_\mu (\e_a))_\rho (\dd_\nu (\e_b))^\rho +\gt^{\mu\nu}(\e_a)_\rho \dd_\mu \dd_\nu (\e_b)^\rho.$$
Using (\ref{covariant derivative}), the second term here equals
$$\gt^{\mu\nu} (\e_a)_\rho \dt_\mu (\dd_\nu (\e_b)^\rho)-\gt^{\mu\nu} (\e_a)_\rho \gat_{\mu\nu}^{\ \ \,\alpha} \dd_\alpha (\e_b)^\rho+\gt^{\mu\nu}(\e_a)_\rho \gat_{\mu\beta}^{\ \ \,\rho} \dd_\nu (\e_b)^\beta.$$
Here the last two terms again behaves well and do not require further calculation. The first term here equals $$\gt^{\mu\nu} (\e_a)_\rho\dt_\mu(\dt_\nu (\e_b)^\rho+\gat_{\nu\sigma}^{\ \ \rho}(\e_b)^{\sigma})=(\e_a)_\rho \widetilde{\Box}_{\gt}(\e_b)^\rho+\gt^{\mu\nu}(\e_a)_\rho \gat_{\nu\sigma}^{\ \ \rho} \dt_\mu (\e_b)^\sigma+\gt^{\mu\nu}(\dt_\mu \gat_{\nu\sigma}^{\ \ \rho})(\e_a)_\rho(\e_b)^\sigma \Sigma^{ab}.$$ 

We need to deal with the last term above. Recall the curvature formula (\ref{curvatureinChristoffelsymbol})
$$R_{\mu\sigma\nu}{}^{\rho}=-\dt_\mu \gat_{\sigma\nu}^{\ \ \,\rho} +\dt_\sigma \gat_{\mu\nu}^{\ \ \,\rho} - \gat_{\mu\lambda}^{\ \ \,\rho}\,\gat_{\sigma\nu}^{\ \ \,\lambda}+\gat_{\sigma\lambda}^{\ \ \,\rho}\,\gat_{\mu\nu}^{\ \ \,\lambda},$$
so we have
$$\gt^{\mu\nu}\dt_\mu \gat_{\nu\sigma}^{\ \ \,\rho}=-\gt^{\mu\nu}R_{\mu\sigma\nu}{}^{\rho}+\gt^{\mu\nu}\dt_\sigma \gat_{\mu\nu}^{\ \ \,\rho} - \gt^{\mu\nu}\gat_{\mu\lambda}^{\ \ \,\rho}\,\gat_{\sigma\nu}^{\ \ \,\lambda}+\gt^{\mu\nu}\gat_{\sigma\lambda}^{\ \ \,\rho}\,\gat_{\mu\nu}^{\ \ \,\lambda}.$$
The first term on the right is $-\gt^{\delta\rho} R_{\sigma\delta}$, where $R_{\mu\nu}$ is the Ricci tensor; for the second term we have $\gt^{\mu\nu}\dt_\sigma \gat_{\mu\nu}^{\ \ \,\rho}=\dt_\sigma (\gt^{\mu\nu}\gat_{\mu\nu}^{\ \ \,\rho})-(\dt_\sigma \gt^{\mu\nu})\gat_{\mu\nu}^{\ \ \,\rho}=\dt_\sigma (\gt^{\mu\nu}\gh_{\mu\nu}^{\ \ \,\rho})-(\dt_\sigma \gt^{\mu\nu})\gat_{\mu\nu}^{\ \ \,\rho}$, where for the last equality we used the generalized wave coordinate condition (\ref{generalizedwcc}). 
\end{proof}

Using the reduced Einstein equation and Proposition \ref{GS}, we know that the tetrad we choose satisfies the equation for its components $\widetilde{\Box}_{\gt} (\e_a)^\mu=\widetilde F_a^\mu(\gt,\dt\gt,\psi,\dt\psi)$. Also, the Ricci curvature term can be similarly written as a term only involving at most first order derivatives of the metric and the field, using the Einstein equation. Then we see that the right hand side of the equation \eqref{secondordereqpsi} only contains zero and first order derivatives of $\gt$ and $\psi$. This, together with the reduced Einstein equation, forms a system of quasilinear wave equations, and we have the existence and uniqueness of a local-in-time solution.\footnote{There is another way to fix the gauge of tetrad, by giving the evolution equation of the tetrad. Once we know the equation for $\widetilde\Box_{\gt} (e_a)^\mu$, the equation of $\psi$ is well-posed, in view of Lemma \ref{avoidloss}. Here our choice just corresponds to the equation $\widetilde{\Box}_{\gt} (\e_a)^\mu=\widetilde F_a^\mu(\gt,\dt\gt,\psi,\dt\psi)$ which is actually derived from our explicit choice. See \cite{Bao:1984bp} for a use of this idea in a more complicated model.}\footnote{It is also clear without any additional work that this local wellposedness also holds for massive case $m\neq 0$.}
We will also use this second order equation for the spinor field to study the global problem in the remaining part of this work.

\section{Energy estimates}

\subsection{\texorpdfstring{$L^2$}{L2} estimate for Dirac equation}

Consider the Dirac equation with an inhomogeneous term in a curved background:
$$\gamma^\mu D_\mu \psi+im\psi=F.$$
From Section \ref{Chnotation} we also have the adjoint Dirac equation $D_\mu \pb \gamma^\nu-im\pb=\overline{F}$.
Then by Corollary \ref{spinorleibnizcor}, we have
$\dd_\mu (\pb \gamma^\mu \psi)=\pb (F-im\psi)+(\overline{F}+im\pb)\psi=\pb F+\overline F\psi$. Writing this in our coordinates gives $$\dt_\mu(\pb \gamma^\mu \psi)=\pb F+\overline{F}\psi-\gat_{\mu\rho}^{\ \ \,\mu} \pb\gamma^\rho \psi,$$
and we recall here $\gamma^\rho=\gb^a (\e_a)^\rho$. We then integrate on the spacetime region between two time slices to get
\begin{equation*}
       \int_{\Sigma_{t}} \pb \gamma^0 \psi\, d\x =\int_{\Sigma_0} \pb \gamma^0 \psi \, d\x  + \int_0^t \int_{\Sigma_\tau} (\pb F+\overline{F}\psi-\gat_{\mu\rho}^{\ \ \,\mu} \pb\gamma^\rho \psi) \,d\x d\tau.
\end{equation*}

We will use the following estimate for the massless case $m=0$:
\begin{Prop}\label{energyestimateDirac} Suppose $\psi\in \mathbb{C}^4$ is a solution of the equation $\gamma^\mu D_\mu\psi+im\psi=F$ which vanishes at spacelike infinity. Assume that the background metric $\gt=\hti^0+\hi$ satisfies the bounds 
\begin{equation*}
    \begin{split}
        &|\dt \hi|\leq C\varepsilon(1+t)^{-\frac 34}(1+|q^*|)^{-\frac 12},\ \ \ |\hi|\leq C\varepsilon (1+t)^{-\frac 34}(1+|q^*|)^{\frac 12},\\
        &|\dt\, \slashed{\mathrm{tr}}\hi|\leq C\varepsilon (1+t)^{-\frac 54},\ \ \ |\dt\hi|_{\lc \Uc}\leq C\varepsilon(1+t)^{-1},\ \ \ M\leq \varepsilon,
    \end{split}
\end{equation*}
where $\Tc=\{\lt,\widetilde S_1,\widetilde S_2\}$, and $\Uc= \{\lt,\ltb,\widetilde S_1,\widetilde S_2\}$. Then 
\begin{equation*}
\int_{\Sigma_{t}} |\psi|^2 \, d\x \leq 2\int_{\Sigma_0} |\psi|^2 \, d\x
+4\int_{0}^{t} \int_{\Sigma_\tau} \left(\frac{C\varepsilon}{1+\tau} |\psi|^2  +|\psi||F|\right) \, d\x d\tau
\end{equation*}
\end{Prop}

\begin{proof}
Notice that $$\gat_{\mu\rho}^{\ \ \,\mu}=\frac 12 \gt^{\mu\nu} (\dt_\mu \gt_{\rho\nu}+\dt_{\rho}\gt_{\mu\nu}-\dt_\nu \gt_{\mu\rho})=\frac 12 \gt^{\mu\nu}\dt_\rho \gt_{\mu\nu}=\frac 12\mh^{\mu\nu} \dt_\rho \gt_{\mu\nu}+\frac 12 (\gt-\mh)^{\mu\nu} \dt_\rho \gt_{\mu\nu},$$
so expanding in null frame we have the estimate 
$$|\gat_{\mu\rho}^{\ \ \,\mu}|\lesssim |\dtb \gt|+|\dt\, \slashed{\mathrm{tr}}\gt|+|\dt \gt|_{\lc \Uc}+|\gt-\mh||\dt\gt|,$$
and hence $|\gat_{\mu\rho}^{\ \ \,\mu} \pb\gamma^\rho \psi|\leq C\varepsilon(1+t)^{-1}|\psi|^2$ using the assumption and (\ref{dm0higherorder}). Now using the simple bound $|\gamma^0-\gb^0|\leq \frac 14$ we have the relation
\begin{equation*}
    \frac 34\int_{\Sigma_t} |\psi|^2 \, d\x\leq \int_{\Sigma_t}\pb\gamma^0 \psi\, d\x\leq \frac 54\int_{\Sigma_t}|\psi|^2\, d\x,
\end{equation*}
so the estimate follows.
\end{proof}

\subsection{Energy estimate for wave equation}

We now give the energy estimate for the quasilinear wave equation, which was used in \cite{04,LT,KLMKG21}.
We consider two weight functions of $q^*=r^*-t$
\begin{equation}\label{eqforw}
w(q^*)=\begin{cases}
1+(1+|q^*|)^{-2\mu},\, q^*<0\\
1+(1+|q^*|)^{1+2\gamma},\, q^*>0
\end{cases}\end{equation}
and \begin{equation}\label{eqforw1}
    w_1(q^*)=\begin{cases}
1+(1+|q^*|)^{-2\mu},\, q^*<0\\
1+(1+|q^*|)^{2+2s},\, q^*>0
\end{cases}
\end{equation}
where $0<\gamma<1$, $\mu>0$, and $0<s<1$. We have $$(1+|q^*|)^{-1}(1+q^*_-)^{-2\mu}w\lesssim_\mu w'\lesssim (1+|q^*|)^{-1}w,\ \ (1+|q^*|)^{-1}(1+q^*_-)^{-2\mu}w_1\lesssim_\mu w_1'\lesssim (1+|q^*|)^{-1}w_1.$$
\begin{Prop}\label{energyestimatewave}
If $g$ with $\gt^\mn=\mt_0^\mn+\Hi^\mn$ satisfies the following assumption:
\begin{equation}
    \begin{split}
        M&\leq \varepsilon,\\
        |\Hi|+(1+|q^*|)|\dt\Hi|+(1+t+r^*)|\dtb \Hi|&\leq C\varepsilon (1+|q^*|)^{\frac 12 -\mu}(1+t)^{-\frac 12-\mu},\\
        |\Hi|_{\lc\lc}+(1+|q^*|)|\dt\Hi|_{\lc\lc}&\leq C\varepsilon (1+|q^*|)(1+t+r^*)^{-1-2\mu},
    \end{split}
\end{equation} then
$$\int_{\Sigma_t} |\dt\phi|^2 w\, d\x+\int_0^t \int_{\Sigma_\tau}|\dtb\phi|^2 w' \, d\x d\tau\leq 8\int_{\Sigma_0} |\dt\phi|^2w\, d\x+12\int_0^t \int_{\Sigma_\tau} |\widetilde{\Box}_{\gt}\phi||\dt\phi|w\, d\x d\tau,$$
and the same estimate holds if we replace $w$ by $w_1$, and $\phi$ by a 4-component spinor field $\psi$.
\end{Prop}

\begin{proof}
We follow the idea of the proof from \cite{04}. One can prove the following identity
\begin{multline*}
    \frac d{dt} \int (-\gt^{00}\dt_t \phi \dt_t\phi+\gt^{ij} \dt_i\phi\dt_j\phi)w(q^*)\, d\x=-\int w'(q^*)(\gt^{\a\b}\dt_\a\phi\dt_\b\phi+2(\omega_i \gt^{i\a}-\gt^{0\a})\dt_t\phi \dt_\a \phi)\, d\x\\
    -\int w(q^*)(2\widetilde\Box_{\gt} \phi \dt_t\phi-(\dt_t\gt^{\a\b})\dt_\a\phi\dt_\b\phi+2(\dt_\a g^{\a\b})\dt_\b\phi\dt_t\phi)\, d\x
\end{multline*}
We also have $\mh^{\a\b}\dt_\a\phi\dt_\b\phi+2(\omega_i \mh^{i\a}-\mh^{0\a})\dt_t\phi \dt_\a \phi=(\dt_t \phi+\d_{r^*}\phi)^2+\delta^{ij}\dtb_i\phi\dtb_j\phi=|\dtb\phi|^2$. Therefore, using the equivalence
$$\frac 12|\dt\phi|^2\leq -\gt^{00}\dt_t \phi \dt_t\phi+\gt^{ij} \dt_i\phi\dt_j\phi\leq 2|\dt\phi|^2, $$and integrating over $[0,t]$, we have for $\Hh^{\a\b}=\gt^{\a\b}-\mh^{\a\b}$ that
\begin{multline*}
    \int_{\Sigma_0} |\dt\phi|^2 w\, d\x+2\int_0^t\int_{\Sigma_\tau} |\dtb\phi|^2w'\, d\x\leq 4\int_{\Sigma_t} |\dt\phi|^2w\, d\x\\
    +2\int_0^t\int_{\Sigma_\tau} |\Hh^{\a\b}\dt_\a\phi \dt_\b\phi+2(\omega_i \Hh^{i\b}-\Hh^{0\b})\dt_\b\phi\dt_t\phi|w'\, d\x dt\\
    +2\int_0^t\int_{\Sigma_\tau}|(2\dt_\a\Hh^{\a\b})\dt_\b\phi\dt_t\phi-(\dt_t\Hh^{\a\b})\dt_\a\phi\dt_\b\phi+2\widetilde\Box_{\gt} \phi \dt_t\phi|w\, d\x,
\end{multline*}
Now by the null condition and the assumption we have 
$$|(\dt_t\Hi^{\a\b})\dt_\a\phi\dt_\b\phi|\leq C\varepsilon \frac 1{(1+t+r^*)^{1+2\mu}}|\dt\phi|^2w+C\varepsilon |\dtb\phi|^2w',$$
$$|(\dt_\a\Hi^{\a\b})\dt_\b\phi\dt_t\phi|\leq C\varepsilon\frac 1{(1+t+r^*)^{1+2\mu}}|\dt\phi|^2w+C\varepsilon |\dtb\phi|^2w',$$
and if we replace $\Hi$ by $\mt_0$, this two terms are also good in view of \eqref{dm0higherorder}.
Similarly, using $(1+|q^*|)w'\lesssim w$
\begin{multline*}
    |\Hi^{\a\b}\dt_\a\phi \dt_\b\phi|+|(\omega_i \Hi^{i\b}-\Hi^{0\b})\dt_\b\phi\dt_t\phi|w'\leq C\varepsilon \frac 1{(1+t+r^*)^{1+2\mu}}|\dt\phi|^2w\\
    +C\varepsilon \frac{(1+|q^*|)^{\frac 12-\mu}}{(1+t+r^*)^{\frac 12+\mu}}|\dt\phi||\dtb\phi|)w'\\
    \leq C\varepsilon \frac 1{(1+t+r^*)^{1+2\mu}}|\dt\phi|^2w+C\varepsilon  |\dt\phi|^2w'.
\end{multline*}
Also, in view of the decomposition \eqref{m0}, we have
$|(\mt_0-\mh)|_{\lc\lc}\lesssim \frac{M\ln(1+r)}{(1+t+r^*)^2}+\frac{M\chi'(\frac r{1+t})\ln(1+r)}{1+t+r^*}\lesssim C\varepsilon (1+|q^*|)(1+t+r^*)^{-1-2\mu}$, so the same estimate holds. Now absorb terms with $|\dt\phi|^2$ using the spacetime integral on the left hand side, and use standard Gronwall's inequality, we get the desired estimate.
\end{proof}

\subsection{Hardy's inequality}
We state the following lemma from \cite{04}, which will be used to control the $L^2$ norms of Lie derivatives of the metric itself.

\begin{lem}\label{Hardy}
Let $w(q^*)$ be the weight function defined in \eqref{eqforw}. Then for any $-1\leq a\leq 1$ and any $\phi\in C_0^\infty (\mathbb{R}^3)$,
\begin{equation}
    \int \frac{|\phi|^2}{(1+|q^*|)^2}\frac{w\, d\x}{(1+t+|q^*|)^{1-a}}\lesssim \int |\d \phi|^2 \frac{w\, d\x}{(1+t+|q^*|)^{1-a}}.
\end{equation}
If in addition $a<2\min(\y,\mu)$, then
\begin{equation}
    \int \frac{|\phi|^2}{(1+|q^*|)^2}\frac{(1+|q^*|)^{-a}}{(1+t+|q^*|)^{1-a}}\frac{w\, d\x}{(1+q^*_-)^{2\mu}}\lesssim \int |\d \phi|^2 \min(w',\frac{w}{(1+t+|q^*|)^{1-a}})\, d\x.
\end{equation}
\end{lem}

\section{Decay estimates}
\subsection{Weighted Klainerman-Sobolev inequality}\label{KS}
We will use the following version of Klainerman-Sobolev inequality, which slightly generalizes the one in \cite{04} simply in view of the proof there:
\begin{Prop}
For any smooth function $\phi(t,\x)$ which is spatially compactly supported, we have $$|\phi(t,\x)|(1+t+|q^*|)(1+|q^*|)w_i^{1/2}(q^*)\leq C\sum_{|I|\leq 3} ||w_i^{1/2}(q^*)\z^I\phi(t,\cdot)||_{L^2(\mathbb{R}^3)},\ \, i=0,1,2$$
where $w_0=w$ and $w_1$ are defined in \eqref{eqforw}, \eqref{eqforw1}, and $w_2=1$.
\end{Prop}

The proof are the same since $w_i$ satisfies $|w_i'(q^*)|(1+|q^*|)\leq Cw_i(q^*)$ and $w_i(q^*)\approx w_i(t)$ in the far interior $r^*\leq t/2$.

This also holds if we replace $\z$ by $\lz$, or the version adapted to the spinor field (\ref{vectorfieldDirac}) which we will define, by the equivalence (\ref{equivalenceLiederivative}) and (\ref{equivalencespinor}).

\subsection{H\"{o}rmander \texorpdfstring{$L^1$}{L1} - \texorpdfstring{$L^\infty$}{Li} estimate}\label{HormanderSection}
We can improve the interior decay of the metric using H\"{o}rmander's $L^1$-$L^\infty$ estimates, see \cite{04}:
\begin{Prop}\label{Hormanderprop}
Suppose $u$ is a solution to the linear inhomogeneous wave equation $\Box u=F$ with vanishing initial data $u|_{t=0}=\d_{t} u|_{t=0}=0$. Then
$$|u(t,x)|(1+t+|x|)\leq C \sum_{|I|\leq 2} \int_0^t \int_{\mathbb{R}^3} \frac{|Z^I F(s,y)|}{1+s+|y|}\, dyds.$$
\end{Prop}

We also need to estimate the linear homogeneous solution (also see \cite{04}):
\begin{lem}\label{homogeneouslemma}
If $v$ is the solution to $\Box v=0$, with initial data $v|_{t=0}=v_0$ and $\d_t v|_{t=0}=v_1$, then for any $\y>0$
$$(1+t+r)|v(t,x)|\leq \sup_x \left((1+|x|)^{2+\y} (|v_1(x)|+|\d v_0(x)|)+(1+|x|)^{1+\y}|v_0(x)|\right).$$
\end{lem}

\subsection{Weighted \texorpdfstring{$L^\infty$}{Li} - \texorpdfstring{$L^\infty$}{Li} estimate}\label{improveddecay}


We need to integrate on the direction of the outgoing light cone in order to get the estimates that are predicted in the asymptotic system. Following \cite{lindblad1990lifespan} and \cite{KLMKG21} we have 
\begin{lem}\label{improveddecaylemma}
Let $D_t=\{(t,\x)\colon |t-|\x||\leq c_0 t\}$ for some constant $0<c_0<1$. Let $\wb(q^*)$ be any continuous function, where $q^*=r^*-t$, $r^*=r+M\ln r$. Suppose that $\Box^*\phi_{\mu\nu}=F_{\mu\nu}$. Then for $\Ut,\widetilde V\in \{\lt,\ltb,\widetilde S_1,\widetilde S_2\}$ and $\phi_{\Ut\widetilde V}=\phi_{\mu\nu}\Ut^\mu \widetilde{V}^\nu$ we have
\begin{multline*}
    (1+t+r^*)|\dt \phi_{\Ut\widetilde V}(t,x) \wb(q^*)|\lesssim \sup_{|q^*|/4\leq \tau\leq t}\sum_{|I|\leq 1} ||\z^I\phi(\tau,\cdot) \wb||_{L^\infty}  \\
    +\int_{|q^*|/4}^t (1+\tau)||F_{\Ut\widetilde V}(\tau,\cdot) \wb||_{L^\infty(D_\tau)}+\sum_{|I|\leq 2} (1+\tau)^{-1} ||\z^I \phi(\tau,\cdot)\wb||_{L^\infty(D_\tau)} \, d\tau.
\end{multline*}
\end{lem}

\begin{proof}
For tangential derivatives, and for points $(t,\x)$ outside the region $D_t$, we can directly use \eqref{dtoZ}, so we focus on $\d_{q^*}\phi_{\Ut\widetilde V}$ at points inside $D=\cup_{\tau\geq 0} D_\tau$. Recall the decomposition of the Minkowski wave operator $\Box^*\phi=4(r^*)^{-1}\d_{s^*}\d_{q^*}(r^*\phi)+(r^*)^{-2}\triangle_\omega\phi$, where $\d_{q^*}=\frac 12(\d_{r^*}-\dt_t)$, $\d_{s^*}=\frac 12(\d_{r^*}+\dt_t)$, and $\bigtriangleup_\omega=\sum_{i,j}\Omega_{ij}^2$. Notice that the null frame vectors commute with $\d_{s^*}$ and $\d_{q^*}$, so this gives the estimate
\begin{equation*}
    |\d_{s^*}\d_{q^*}(r^*\phi_{\Ut\widetilde V})|\lesssim r^*|F_{\Ut\widetilde V}|+(r^*)^{-1}\sum_{|I|\leq 2}|\Ut^\mu \widetilde V^\nu \z^I \phi_{\mu\nu}|.
\end{equation*}
Now integrating along the flow lines of $\d_{s^*}$ from the boundary of $D=\cup_{\tau\geq 0} D_\tau$ to $(t,\x)$, and using that $\wb(q^*)$ is constant along these flow lines, we get the result.
\end{proof}

Clearly, same estimate works if we consider the scalar case $\Box^*\phi=F$, and we shall apply this to components of the spinor field.


\section{Bootstrap assumptions}\label{BA}
We define the energy for the metric: $\displaystyle E_N(t)=\sup_{0\leq\tau\leq t}\int_{\Sigma_\tau} |\dt \z^I \hi|^2w(q^*)\, d\x,$ where $$w(q^*)=\begin{cases}
1+(1+|q^*|)^{-2\mu},\, q^*<0,\\
1+(1+|q^*|)^{1+2\gamma},\, q^*\geq 0.
\end{cases}$$
Similarly, we define $\displaystyle E^1_N(t)=\sup_{0\leq\tau\leq t}\int_{\Sigma_\tau}|\dt \z^I \psi|^2 w_1(q^*)\, d\x$,
where $$w_1(q^*)=\begin{cases}
1+(1+|q^*|)^{-2\mu},\, q^*<0,\\
1+(1+|q^*|)^{2+2s},\, q^*\geq 0.
\end{cases}$$
These are the same weight functions we defined in \eqref{eqforw}, \eqref{eqforw1}.
We also consider the $L^2$ norm for the spinor field itself: $\displaystyle C_N(t)=\sup_{0\leq\tau\leq t}\int_{\Sigma_\tau} |\psi|^2 \,d\x$.

Let $T$ be the maximal time such that the following inequalities hold:
\begin{equation}\label{BAeq}
    E_N(t)+C_N(t)\leq C_b\varepsilon^2(1+t)^{2\delta},\  E_N^1(t)\leq C_b\varepsilon^2,
\end{equation}
where $C_b$ is a constant which we will determine later (and all constant $C$ below might be dependent on $C_b$). We take $\delta>0$ small so that $\delta<s/10$. By local existence result, we know that $T>0$, and we want to show that these estimates can be improved, so by continuity we must have $T=\infty$.

\begin{Prop}[Weak decay]\label{weakdecay}
Suppose the bootstrap assumption \eqref{BAeq} holds. Then
\begin{equation}
\stepcounter{equation}
\tag{\theequation a}\label{weakdecaydh1}
    |\dt \z^I \hi|\leq
\begin{cases}
C\varepsilon(1+t)^{\delta}(1+t+r^*)^{-1} (1+|q^*|)^{-\frac12},\ q^*<0\\
C\varepsilon(1+t)^{\delta}(1+t+r^*)^{-1} (1+|q^*|)^{-1-\gamma},\ q^*\geq 0
\end{cases}
,\ |I|\leq N-3
\end{equation}

\begin{equation}\label{weakdecayh1}
\tag{\theequation b}|\z^I \hi|\leq
    \begin{cases}
    C\varepsilon(1+t+r^*)^{-1+\delta} (1+|q^*|)^{\frac 12},\ q^*<0\\
    C\varepsilon(1+t+r^*)^{-1+\delta} (1+|q^*|)^{-\gamma},\ q^*\geq 0
    \end{cases}
,\ |I|\leq N-3
\end{equation}

\begin{equation}
\tag{\theequation d}
|\dt \z^I \psi|\leq 
\begin{cases}
C\varepsilon(1+t)^{\delta}(1+t+r^*)^{-1} (1+|q^*|)^{-\frac 32},\ q^*<0\\
C\varepsilon(1+t+r^*)^{-1} (1+|q^*|)^{-\frac 32-s},\ q^*\geq 0
\end{cases}
 ,\ |I|\leq N-4\end{equation}

\begin{equation}
\tag{\theequation c}
|\z^I \psi|\leq\begin{cases}
C\varepsilon (1+t+r^*)^{-1+\delta} (1+|q^*|)^{-\frac 12},\ q^*<0\\
C\varepsilon(1+t+r^*)^{-1} (1+|q^*|)^{-\frac 12-s},\ q^*\geq 0
\end{cases}
,\ |I|\leq N-3
\end{equation}
\end{Prop}

\begin{proof}
The decay (\ref{weakdecaydh1}) follows from the bootstrap assumption and the Klainerman-Sobolev inequality (Proposition \ref{KS}). 

From \eqref{weakdecaydh1} at $t=0$ we see that 
\begin{equation}\label{dhinr}
    |\d\z^I\hi(0,\x)|\leq C\varepsilon(1+r^*)^{-2-\y},\ \ |I|\leq N-3,
\end{equation}
we can integrate along $\d_{r^*}$ direction to spacelike infinity to get for the inital data that 
\begin{equation}\label{hinr}
    |\z^I h^1(0,\x)|\leq C\varepsilon (1+r^*)^{-1-\gamma},\ \ |I|\leq N-3.
\end{equation} Then we integrate along the integral curve of $\d_{q^*}$ to the initial hypersurface
and use the decay property of initial data to get (\ref{weakdecayh1}) for $q^*\leq 0$, and the decay on $q^*=0$ to get for $q^*<0$. 

For the spinor field $\psi$, when $q^*\geq 0$ we first use the bound of $E^1_N(t)$, and integrate along $\d_{q^*}$ as the metric case; for $q^*<0$, we use the bound of $C_N(t)$ first to get estimate for $|\z^I \psi|$ and then use (\ref{dtoZ}).
\end{proof}

\begin{Cor}
For tangential derivatives we have 
\begin{equation}
    |\dtb\z^I\hi|\leq C\varepsilon (1+t+r^*)^{-2+\delta}(1+|q^*|)^{\frac 12}(1+q_+^*)^{-\frac 12-\gamma},\ \ |I|\leq N-4,
\end{equation}
\begin{equation}
    |\dtb\z^I\psi|\leq C\varepsilon (1+t+r^*)^{-2+\delta}(1+|q^*|)^{-\frac 12}(1+q_+^*)^{-s},\ \ |I|\leq N-4.
\end{equation}
\end{Cor}
\begin{proof}
These estimate immediately follow from the decay above and $\eqref{dtoZ}$.
\end{proof}




We will often need to estimate terms involing the inverse metric, where we need the following lemmas in \cite{KLMKG21} to get same estimates for $\Hi$ as $\hi$:
\begin{lem}
Let $g_{\a\b}=m_{\a\b}+h^0_{\a\b}+h^1_{\a\b}$, and let $g^{\a\b}=m^{\a\b}+H_0^{\a\b}+H_1^{\a\b}$, where $h_0$ and $H^0$ are defined as in \eqref{h0} and \eqref{H0}. Suppose the estimate $|Z^I(h^0_{\a\b})|+|Z^I(h^1_{\a\b})|\lesssim \varepsilon (1+t+r)^{-1+\delta}$ holds for $|I|\leq N$, and a constant $\delta\in (0,1/2)$. Then for vector fields $X,Y$ with bounded components and sufficently small $\varepsilon>0$, we have
\begin{equation*}
    |X^\a Y^\b Z^I (H_{1 \a\b})|\leq |X^\a Y^\b Z^I (h^1_{\a\b})|+C\varepsilon^2 (1+t+r)^{-2+2\delta},
\end{equation*}
where $H_{1\a\b}$ is from lowering indices of $H_1^{\a\b}$ using Minkowski metric.
\end{lem}

\begin{lem}\label{lemmaH1equivh1integral}
Let $\wb$ be a positive function, and take $N$ such that $|Z^I(h^0_{\a\b})|+|Z^I(h^1_{\a\b})|\lesssim \varepsilon (1+t+r)^{-1+\delta}$ holds for $|I|\leq N/2$. Then, for a given region $\Omega$ with volume element $dV$
\begin{equation*}
    \int_\Omega |Z^I(H_{1\a\b})|^2\wb dV\leq \int_\Omega |Z^I(h^1_{\a\b})|^2 \wb dV+C\varepsilon \sum_{|J|\leq |I|-1}\int_\Omega (1+t+r)^{-2+2\delta}(|Z^J h|^2+|Z^J H|^2)\wb dV.
\end{equation*}
\end{lem}

\section{Commutators}
We estimate the commutator in this section. To compute the commutator involving the spinor field explicitly, we need to fix a choice of tetrad, which we did in Section \ref{constructionoftetradsection}. Therefore, when commuting we do not take local Lorentz covariance into consideration anymore. For instance, for Lie derivative of $\psi$ we view $\psi$ as a four component complex scalar field.




\subsection{Reduced wave operator}

We define the modified Lie derivative in the new coordinates by 
$$\lzh \widetilde{K}^{\a_1\cdots\alpha_r}_{\b_1\cdots\b_s}+\textstyle\frac{r-s}4 (\dt_\lambda \z^\lambda) \widetilde{K}^{\a_1\cdots\alpha_r}_{\b_1\cdots\b_s}.$$
Then one can verify that $\lzh \mh^{\a\b}=0$ for $\z\in\{\dt_\mu,\widetilde \Omega_{ij},\widetilde\Omega_{0i},\widetilde S\}$.

It is straightforward to verify the following identity, where $\k=\k_0^{-1}$ with $\k_0$ defined in Lemma \ref{m0lemma}:
$$\lz (\k \gt^{\a\b}\dt_\a \dt_\b \hi_{\mu\nu})=(\lzh (\k \gt^{\a\b})) \dt_\a \dt_\b \hi_{\mu\nu}+\k \gt^{\a\b}\dt_\a\dt_\b \lzh \hi_{\mu\nu}.$$
For the spinor field, we formally view $\psi$ as a 4-component scalar in computation. Define $\widehat{\widehat \l}_{\z} \psi=\lz \psi-c_{\z} \psi$, where $c_{\widetilde S}=2$ and $c_{\z}=0$ for other vector fields. Then similarly
$$\lz (\k \gt^{\a\b}\dt_\a \dt_\b \psi)=(\lzh (\k \gt^{\a\b})) \dt_\a \dt_\b \psi+\k \gt^{\a\b}\dt_\a\dt_\b \lhhz \psi.$$

Applying more $\lz$ on left hand side, and iterating the process, we have
\begin{equation}\label{RhcomI}
    R^{com\ I}_{\hi}=\k\gt^{\a\b} \dt_\a\dt_\b \lzh^I \hi_{\mu\nu}-\lz^I(\k\gt^{\a\b}\dt_\a\dt_\b\hi_{\mu\nu})=\sum_{J+K=I,\, |K|<|I|}\lzh^J (\kappa \gt-\mh)^{\alpha\beta} \dt_\a\dt_\b \lzh^K \hi_{\mu\nu},
\end{equation}
and
\begin{equation}\label{RpsicomI}
    R^{com\ I}_\psi=\k\gt^{\a\b} \dt_\a\dt_\b (\lhhz)^I \psi-\lz^I(\k\gt^{\a\b}\dt_\a\dt_\b\psi)=\sum_{J+K=I,\, |K|<|I|}\lzh^J (\kappa \gt-\mh)^{\alpha\beta} \dt_\a\dt_\b (\lhhz)^K \psi,
\end{equation}

Note that $\k\gt^{\a\b}-\mh^{\a\b}=\k\Hi^{\a\b}+(\k\mt_0^{\a\b}-\mh^{\a\b})$. Using \eqref{m0} and (\ref{dtoZ}) we have 
\begin{equation}
    |(\lz^I(\k\mt_0^{\mu\nu}-\mh^{\mu\nu}))\dt_\mu\dt_\nu \phi|\lesssim \frac {M\ln(1+t+r)}{(1+t+r)^2}\sum_{|J|\leq 1}|\dt \z^J \phi|,
\end{equation} where $\phi$ can be $\lzh^I\hi_{\mu\nu}$ or $(\lhhz)^I\psi$. For the $\Hi$ part, we expand in the null frame, using \eqref{dtoZ}, to get
\begin{equation}\label{commutatorestimate}
    |R^{com\ I}_{\hi}|\lesssim \sum_{\substack{|J|+|K|\leq |I|+1\\ 1\leq |J|\leq |I|}} \left( \frac{M\ln (1+t+r)}{(1+t+r)^2}+\frac{|\lz^J \Hi|_{\lc\lc}}{1+|q^*|} +\frac{|\lz^J \Hi|}{1+t+r}\right)|\dt \lz^K \hi|,
\end{equation}
and same estimate holds if we replace $\hi$ by $\psi$.

We will also commute to get the constant wave operator $\Box^*=\mh^{\mu\nu}\dt_\mu\dt_\nu$ for decay estimates:
\begin{equation*}
    R^{com*I}_{\hi}=\Box^*\lzh^I \hi_{\mu\nu}-\lz^I(\k\gt^{\a\b}\dt_\a\dt_\b \hi_{\mu\nu})=\sum_{J+K=I}\lzh^J (\kappa \gt-\mh)^{\alpha\beta} \dt_\a\dt_\b \lzh^K \hi_{\mu\nu},
\end{equation*}
\begin{equation*}
    R^{com*I}_\psi=\Box^*(\lhhz)^I \psi-\lz^I(\k\gt^{\a\b}\dt_\a\dt_\b \psi)=\sum_{J+K=I}\lzh^J (\kappa \gt-\mh)^{\alpha\beta} \dt_\a\dt_\b (\lhhz)^K \psi,
\end{equation*}
and similarly,
\begin{equation}\label{commutatorestimatefordecay}
    |R^{com*I}_{\hi}|\lesssim \sum_{\substack{|J|+|K|\leq |I|+1\\ |J|\leq |I|}} \left( \frac{M\ln (1+t+r)}{(1+t+r)^2}+\frac{|\lz^J \Hi|_{\lc\lc}}{1+|q^*|} +\frac{|\lz^J \Hi|}{1+t+r}\right)|\dt \lz^K \hi|.
\end{equation}
Again, same estimate holds if we replace $\hi$ by $\psi$.

\subsection{Dirac operator}
We also want to commute the vector fields with the Dirac equation $\gamma^\mu D_\mu \psi=0$. Like \cite{bachelot1988DiracVectorField} and \cite{1902DiracHyperboloid}, we modify the boost and rotation vector fields so that they commute with the flat Dirac operator. For the scaling vector field, the modification is similar with the wave equation case.

We define in our coordinates $\gb^\mu=\gb^a (\dt_a)^\mu$. In this way, we formally assign a tensorial index to constant Gamma matrices, and the Lie derivative reads $\lz \gb^\mu =\z(\gb^\mu)-\gb^\nu \dt_\nu \z^\mu=-\gb^\nu \dt_\nu \z^\mu$. 

For the spinor field $\psi$, we define the modified derivative with respect to vector fields $\widetilde\Omega_{\mu\nu}=\x_\mu\dt_\nu-\x_\nu \dt_\mu$ \footnote{We here lower the indices using $\mh$, so $\x_0=-t$. In this convention we get the negative Lorentz boost $\Omega_{0i}=-t\dt_i-\x_i \dt_t$, but clearly there is no essential difference.} by \begin{equation}\label{vectorfieldDirac}
    \l_{\widetilde\Omega_{\mu\nu}}^\psi\psi=\l_{\widetilde\Omega_{\mu\nu}}\psi-\frac 12 \gb_\mu\gb_\nu \psi=\widetilde{\Omega}_\mn\psi-\frac 12 \gb_\mu\gb_\nu \psi,
\end{equation}
where $\gb_\mu=\mh_\mn \gb^\nu$. The modified Lie derivative with respect to the scaling vector field $\widetilde S=\x^\mu\dt_\mu$ by $$\lh_{\widetilde S}^\psi \psi=\l_{\widetilde S}^\psi \psi-\psi=\widetilde S\psi-\psi.$$

We also denote $\widehat\l_{\Omega_{\mu\nu}}\psi=\l_{\Omega_{\mu\nu}}\psi$. Note that since $\gb_\mu$ are constant invertible matrices, we have the equivalence of the norm:
\begin{equation}\label{equivalencespinor}\sum_{|J|\leq |I|}|\lz^J \psi|\lesssim |\lhzp^I \psi|\lesssim \sum_{|J|\leq |I|}|\lz^J \psi|. \end{equation}

\begin{Prop}
The following commuting properties hold:
$$\l^\psi_{\widetilde\Omega_{\mu\nu}}(\gb^\rho \dt_\rho \psi)=\gb^\rho \dt_\rho \l^\psi_{\widetilde\Omega_{\mu\nu}} \psi;$$
$$\l^\psi_{\widetilde S} (\gb^\rho\dt_\rho \psi)=\gb^\rho \dt_\rho\lh^\psi_{\widetilde S}\psi.$$
\end{Prop}

\begin{proof}
By definition
\begin{equation*}
    \begin{split}
        \l^\psi_{\widetilde\Omega_{\mu\nu}}(\gb^\rho \dt_\rho \psi)&= (\l_{\widetilde\Omega_{\mu\nu}} \gb^\rho)\dt_\rho \psi+\gb^\rho \dt_\rho \l_{\widetilde\Omega_{\mu\nu}} \psi-\frac 12 \gb_\mu\gb_\nu\gb^\rho \dt_\rho \psi \\
        &=-(\gb^\sigma \dt_\sigma (\widetilde\Omega_{\mu\nu})^\rho) \dt_\rho \psi+\gb^\rho \dt_\rho \l_{\widetilde\Omega_{\mu\nu}} \psi-\frac 12 \gb_\mu\gb_\nu\gb^\rho \dt_\rho \psi\\
        &=-(\gb^\sigma \dt_\sigma (\widetilde\Omega_{\mu\nu})^\rho) \dt_\rho \psi+\gb^\rho \dt_\rho \l^\psi_{\widetilde\Omega_{\mu\nu}} \psi+\frac 12 \gb^\rho \gb_\mu\gb_\nu\dt_\rho \psi-\frac 12 \gb_\mu\gb_\nu\gb^\rho \dt_\rho \psi.
    \end{split}
\end{equation*}
For the first term we calculate $$(\gb^\rho \dt_\rho (\widetilde\Omega_{\mu\nu})^\sigma)\dt_\sigma \psi=\gb^\rho(\dt_\rho(\x_\mu\delta_\nu^\sigma-\x_\nu \delta_\mu^\sigma))\dt_\sigma \psi=\gb^\rho(\mh_{\rho\mu}\delta_\nu^\sigma-\mh_{\rho\nu}\delta_\mu^\sigma)\dt_\sigma \psi=\gb_\mu\dt_\nu\psi-\gb_\nu\dt_\mu \psi,$$
and this equals $-\frac 12(\gb_\mu\gb_\nu\gb^\rho \dt_\rho \psi-\gb^\rho\gb_\mu\gb_\nu\dt_\rho\psi)$ by noticing $\gb_\mu\gb^\mu=-I$ and $\gb^\sigma \gb^\rho=-\gb^\rho\gb^\sigma$. Therefore $$\l^\psi_{\widetilde\Omega_{\mu\nu}}(\gb^\rho \dt_\rho \psi)
=\gb^\rho \dt_\rho\l^\psi_{\widetilde\Omega_{\mu\nu}} \psi. $$

For the second identity, notice that $\l_{\widetilde S} \gb^\mu=-\gb^\rho\dt_\rho \widetilde S^\mu=-\gb^\mu$, so 
$$\l_{\widetilde S} (\gb^\rho\dt_\rho \psi)=-\gb^\rho\dt_\rho \psi+\gb^\rho\l_{\widetilde S} \dt_\rho \psi=\gb^\rho \dt_\rho(\l_{\widetilde S}\psi-\psi)=\gb^\rho\dt_\rho\lh^\psi_{\widetilde S}\psi.\eqno\qed$$
\renewcommand{\qedsymbol}{}
\end{proof}








We are now ready to commute with curved Dirac operator. Recall that $\gamma^\mu=\gb^a (\e_a)^\mu$, where $\e_a$'s are constructed in Section \ref{constructionoftetradsection}. Define $\go^\mu:=\gamma^\mu-\gb^\mu=\gb^a((\e_a)^\mu-(\dt_a)^\mu)$, and
we calculate
\begin{multline}\label{Diraccommutatoridentity}
        \lzp^I (\gamma^\mu D_\mu \psi)=\lzp^I (\mathring{\gamma}^\mu \dt_\mu \psi)+\lzp^I (\gb^\mu \dt_\mu \psi)- \frac 14 \lzp^I (\gamma^\mu \omegat_{\mu ab} \Sigma^{ab} \psi) \\
        =\lzp^I (\mathring{\gamma}^\mu \dt_\mu \psi)+\gb^\mu\dt_\mu \lhzp^I\psi- \frac 14 \lzp^I (\gamma^\mu \omegat_{\mu ab} \Sigma^{ab} \psi)\\
        =\left(\lzp^I (\go^\mu \dt_\mu \psi)-\go^\mu\dt_\mu \lhzp^I \psi\right) +\gamma^\mu D_\mu \lhzp^I \psi\\
        +\frac 14\left(\gamma^\mu \omegat_{\mu ab}\Sigma^{ab}\lhzp^I \psi-  \lzp^I (\gamma^\mu \omegat_{\mu ab} \Sigma^{ab} \psi)\right).
\end{multline}
Note that $\lzp^I (\go^\mu \d_\mu \psi)$ equals $\lz^I(\go^\mu\dt_\mu \psi)$ plus terms involing at most $|I|-1$ vector fields, and $\go^\mu\d_\mu \lhzp^I \psi$ equals $\go^\mu\dt_\mu \lz^I \psi$ plus terms with at most $|I|-1$ vector fields, we see that there are cancellations for terms with all $|I|$ vector fields falling on $\psi$. Similar cancellation occurs in the last parenthesis. Therefore, we have the following estimate for the commutator, using that $\Sigma^{ab}$ and $\gb^a$ are constant matrices, $\gamma^\mu=\gb^a(\e_a)^\mu$, and the equivalence \eqref{equivalencespinor}:
\begin{multline}\label{Diraccommutatorinequality}
|\gamma^\mu D_\mu \lhzp^I \psi-\lzp^I(\gamma^\mu D_\mu \psi)|\lesssim \sum_{|J|+|K|\leq |I|,\, |K|<|I|}\left(|\lz^J \go^\mu \dt_\mu\lz^K \psi| +|\lz^J (\gamma^\mu \omegat_{\mu ab})  \lz^K\psi| \right)\\
\lesssim \sum_{|J|+|K| \leq |I|,\, |K|<|I|}\left(|\lz^J ((\e_a)^\mu-(\dt_a)^\mu) \dt_\mu\lz^K \psi| +|\lz^J ((\e_c)^\mu \omegat_{\mu ab})  \lz^K\psi| \right)
\end{multline}



\section{Improved estimate from generalized wave coodinate condition}
The wave coordinate condition can be written as $\Box_g x^\nu=0$, i.e.\ $\d_\mu(g^\mn \sqrt{|\det g|})=0$ in the original coordinates, which implies for the tensor $H^{\mu\nu}=g^{\mu\nu}-m^{\mu\nu}$ that $\d_\mu(H^{\mu\nu}-\frac 12m^{\mu\nu}m_{\a\b}H^{\a\b})=W^\nu(g)(H,\d H)$. Since $H^0$ is explicit, we can compute to see that $|\d_\mu(H_0^{\mu\nu}-\frac 12m^{\mu\nu}m_{\a\b}H_0^{\a\b})|\lesssim M(1+t+r)^{-2}$ decays well, so we have a similar estimate as $H$ for $H_1=H-H_0$, and in the new coordinates it reads
$$\dt_\mu(\Hi^{\mu\nu}-\frac 12\mt^{\mu\nu}\mt_{\a\b}\Hi^{\a\b})=\widetilde{W}^\nu(\gt)(\Hi,\dt \Hi)+\widetilde{W}^\nu_{mass}+\widetilde{W}^\nu_{cov},$$
where $\widetilde{W}^\nu_{mass}$ and $\widetilde{W}^\nu_{cov}$ are the mass and covariant error terms, which can be ignored just as in Section \ref{ChapterSchwarzschildcoordinates}. This also commutes well with modified Lie derivatives, since $\lzh$ falling on $\mt$ provides extra decay.
Then as in \cite{KLMKG21} the following estimate holds:
\begin{lem}Suppose the generalized wave coordinate condition \eqref{generalizedwcc} holds. Then
\begin{equation}\label{wccestimate}
    |\d_{q^*} \lz^I \Hi|_{\lc\Tc} +|\d_{q^*} \slashed{\mathrm{tr}}\lz^I \Hi|\lesssim |\dtb \lz^I \Hi|+\sum_{\substack{|J|+|K|\leq |I|\\i,j\in \{0,1\}}}|\lz^J \Hti_i||\dt \lz^K \Hti_j|+\frac{M|\chi'(\frac{r}{t+1})|}{(1+t+r)^2}
\end{equation}
where $\chi'(s)$ is nonzero when $1/4\leq s\leq 1/2$.
\end{lem}

\begin{proof}
Let $\Hh_1^{\mu\nu}=\Hi^{\mu\nu}-\frac 12 \mt^{\mu\nu}\mt_{\a\b}\Hi^{\a\b}$. By the property of Lie derivative, we have for the commuting vector fields $\z$ that
\begin{equation}\label{wccLiederivative}
    \dt_\mu \lzh^I\Hh_1^{\mu\nu}=\sum_{|J|\leq |I|}c^I_J\lzh^J \dt_\mu \Hh_1^{\mu\nu}.
\end{equation}
Now we expand the divergence in the null frame
$$\dt_\mu \lzh^I\Hh_1^{\mu\nu}=\lt_\mu \d_{q^*}\lzh^I\Hh_1^{\mu\nu}-\ltb_\mu \d_{s^*} \lzh^I\Hh_1^{\mu\nu}+\sum_{i=1,2}\widetilde S_{i\mu}\d_{\widetilde S_i}\lzh^I\Hh_1^{\mu\nu}.$$
Notice that $(\lzh^I\Hh_1)^{\mn}\approx\lzh^I\Hi^\mn-\frac 12 \mh^\mn\mh_{\a\b}\lzh^I \Hi^{\a\b}$ modulo quadratic terms, so $(\lzh^I\Hh_1)_{\lt\ltb}=\slashed{\mathrm{tr}}\lzh^I\Hi$ plus quadratic terms. Therefore, contracting with $\Tt_\nu$ and $\ltb_\nu$ respectively we get
$$|\d_{q^*} \lz^I \Hi|_{\lc\Tc} +|\d_{q^*} \slashed{\mathrm{tr}}\lz^I \Hi|\lesssim |\dtb \lz^I \Hi|+|\dt_\mu \lzh^I\Hh_1^{\mu\nu}|$$
plus error terms. Then, using \eqref{wccLiederivative} and the decay of error terms, we get the result.
\end{proof}

\begin{Prop}\label{wccL2}
Suppose the weak decay \eqref{weakdecay} holds. Then
\begin{multline}
    \int_0^t\int_{\Sigma_\tau} (|\d_{q^*} \lz^I \Hi|^2_{\lc \Tc}+|\d_{q^*} \slashed{\mathrm{tr}} \lz^I \Hi|^2)w' \, d\x d\tau \lesssim \int_0^t \int_{\Sigma_\tau} |\dtb \lz^I \hi|^2 w'\, d\x d\tau\\
    +(M^2+\varepsilon^2)\left(1+\sum_{|J|\leq|I|} \int_0^t \frac{1}{(1+\tau)^{2-2\delta}}\int_{\Sigma_\tau} |\dt\lz^J \hi|^2 w\, d\x d\tau\right).
\end{multline}
\end{Prop}


\section{Estimate of inhomogeneous terms}

\subsection{Einstein's equation}
Recall that the main part of the inhomogeneous term $\widetilde{F}_{\mu\nu}(\gt)(\dt \hi,\dt\hi)$ is $\widetilde P(\gt)(\dt_\mu \hti^1,\dt_\nu \hti^1)$. Modulo cubic terms we can consider $\widehat{P}=P(\mh)$, and $P(\mh)(\pi,\theta)=\frac 14 \mh^{\a\b} \pi_{\a\b} \, \mh^{\rho\sigma} \theta_{\rho\sigma}-\frac 12 \mh^{\a\b}\mh^{\rho\sigma}\pi_{\a\rho}\, \theta_{\b\sigma}$. First we have \begin{equation}
\label{LmuLnu}|\widehat P(\dt_\mu \hi,\dt_\nu \hi)-\lt_\mu\lt_\nu \widehat P(\d_{q^*}\hi,\d_{q^*} \hi)|\lesssim |\dtb \hi||\dt \hi|+|\dt\hi||\dtb \hi|.\end{equation}
Now expanding $\widehat P$ in the null frame we obtain $\widehat P(\pi,\theta)=-(\pi_{\lt\lt}\theta_{\ltb\ltb}+\pi_{\ltb\ltb}\theta_{\lt\lt})/8-(2\pi_{AB}\theta^{AB}-\slashed{\mathrm{tr}}\pi \, \slashed{\mathrm{tr}}\theta)/4+\delta^{AB}(2\pi_{A\lt}\theta_{B\ltb}+2\pi_{A\ltb}\theta_{B\lt}-\pi_{AB}\theta_{L\ltb}-\pi_{\lt\ltb}\theta_{AB})/4$.
Therefore $$|\widehat P(\pi,\theta)-\widehat P_{\mathcal{S}}(\pi,\theta)|\lesssim (|\pi|_{\lc\Tc}+|\slashed{\mathrm{tr}} \pi|)|\theta|+|\pi|(|\theta|_{\lc\Tc}+|\slashed{\mathrm{tr}} \theta|),$$ 
with
$$P_\mathcal{S}(\pi,\theta)=-\widehat{\pi}_{AB}\widehat{\theta}^{AB}/2,\ A,B\in \widetilde{\mathcal{S}}=\{\widetilde S_1,\widetilde S_2\},\ \text{where}\ \widehat \pi_{AB}=\pi_{AB}-\delta_{AB}\slashed{\mathrm{tr}}\pi/2,$$
so clearly $|P_\mathcal{S}(\pi,\theta)|\lesssim |\pi|_{\Tc\Tc}|\theta|_{\Tc\Tc}$.

For higher order, the structure is well preserved thanks to the properties of Lie derivatives:
\begin{equation}
    \lz S_{\mn}(\dt h,\dt k)=S_\mn (\dt\lzh h,\dt k)+S_\mn (\dt h,\dt\lzh k),
\end{equation}
where $S$ can be $\widehat P$ or $Q$.
Therefore
$$\lz^I(\widetilde F_{\mu\nu}(\gt)(\dt\hti^1,\dt \hti^1))=\sum_{J+K=I}\widetilde F_{\mu\nu}(\gt)(\dt\lzh^J \hi,\dt \lzh^K \hi)+R^{cube\ I},$$
and by (\ref{LmuLnu}) we have 
\begin{equation}\label{ZLmuLnu}
    \k^2 \widetilde F_{\mu\nu}(\gt)(\dt\lzh^J\hti^1,\dt \lzh^K \hti^1)=\lt_\mu \lt_\nu \widehat P(\d_{q^*} \lzh^J \hti^1,\d_{q^*}\lzh^K \hti^1)+\k^2 R^{tan}+\k^2 R^{cube},
\end{equation} with the same $\widehat P$ as above.

We also need to estimate the term from the spinor field.
Recall from \eqref{EMT} that
\begin{equation}
T_{\mu\nu}=\frac i4 (\overline{\psi}\gamma_\mu D_\nu\psi-D_\nu\overline{ \psi} \gamma_\mu \psi)+\frac i4 (\overline{\psi}\gamma_\nu D_\mu\psi-D_\mu\overline{ \psi} \gamma_\nu \psi).
\end{equation}
Note that in this expression, we are lowering the index of $\gamma$ using the metric $g$, which is different from $\gb_\mu=\mh_{\mu\nu}\gb^\nu$. Then $\gamma_\mu-\gb_\mu=\gt_{\mu\nu}\gamma^\nu-\mh_{\mu\nu}\gb^\nu=\gt_{\mu\nu}(\gamma^\nu-\gb^\nu)+(\gt_{\mu\nu}-\mh_{\mu\nu})\gb^\nu$, so modulo error terms we consider instead $\gamma_\mu$ to $\gb_\mu$. Also $D_\mu \psi=\dt_\mu\psi+\frac 14\omegat_{\mu ab}\Sigma^{ab}\psi$, and the terms involing the connection $\omegat$ are also with an extra factor decay. The main term then reads
\begin{equation}
    \frac i4(\pb \gb_\mu \dt_\nu\psi+\dt_\nu \pb\gb_\mu\psi)+(\mu\leftrightarrow\nu)
\end{equation}
Noticing that $\lz{\gb_\mu}=Z(\gb_\mu)+(\dt_\mu Z^\nu)\gb_\nu=(\dt_\mu Z^\nu)\gb_\nu$ is a constant multiple of $\gb_\mu$, we derive the following estimate:
\begin{Prop}
Assume the weak decay in Proposition \ref{weakdecay} holds. Then
\begin{equation}\label{eqEMThigherorderTT}
    |\lz^I T|_{\Tc\Tc}\lesssim \sum_{|J|+|K|\leq |I|}|\lz^J \psi||\dtb\lz^K \psi|+\sum_{|J_1|+|J_2|+|J_3|\leq |I|}|\lz^{J_1} \psi||\dt\lz^{J_2} \psi||\lz^{J_3}(\gt-\mh)|.
\end{equation}
\end{Prop}
We now look at other components. The Dirac equation reads $\gamma^\mu D_\mu \psi=0$. If we decompose this in the Minkowski null frame, we get $$\frac 12\gamma_{\lt}D_{\ltb}\psi=-\frac 12\gamma_{\ltb}D_{\lt}\psi+\sum_{A\in\widetilde{\mathcal{S}}}\gamma_{A}D_{A}\psi,$$ where $\gamma_{\widetilde U}=\mh_{\mu\nu}\gamma^{\mu}\widetilde U^\nu$, and similarly we can do for the conjugate equation. This means that for the solution of Dirac equation, the expression $\gamma_{\lt} D_{\ltb} \psi$ actually behaves like a good (tangential) derivative. Then in view of \eqref{EMT}, if we contract $T_{\mu\nu}$ with $\lt^\mu$, then modulo error terms (e.g.\ from $\gamma_\mu\lt^\mu=\gamma_{\lt}+(\gt-\mh)_{\mu\nu}\gamma^\mu\lt^{\nu}$ because of the notation), which can be ignored, we can always control it by $\psi\cdot\dtb\psi$.

For higher order, it is similar as we have $\gamma^\mu D_\mu \lhzp^I \psi$ equal to the commutator $F^{com\ I}$, which is quadratic in view of \eqref{Diraccommutatorinequality}. When contracting $(\lz^I \Tt)_{\mu\nu}$ with $\lt^\mu \Ut^\nu$, we get tangential derivatives with one exception, which is $\gamma_{\lt}D_{\ltb}\lhzp^I\psi$. Then we can do the same decomposition for the equation $\gamma^\mu D_\mu \lhzp^I \psi=F^{com\ I}$ in null frame the bad derivatives can again be controlled by good derivatives. Therefore, we also have
\begin{Prop}
Suppose that $\psi$ satisfies the Dirac equation $\gamma^\mu D_\mu\psi=0$, and the weak decay in Proposition \ref{weakdecay} holds. Then for $\lc=\{\lt\},\ \Uc=\{\lt,\ltb,\widetilde S_1,\widetilde S_2\}$
\begin{equation}\label{eqTLU}
    |\lz^I T|_{\lc\Uc}\lesssim \sum_{|J|+|K|\leq |I|}|\lz^J \psi||\dtb\lz^K \psi|+\sum_{|J_1|+|J_2|+|J_3|\leq |I|}|\lz^{J_1} \psi||\dt\lz^{J_2} \psi||\lz^{J_3}(\gt-\mh)|.
\end{equation}
\end{Prop}

\subsection{Spin connection coefficients}\label{higherorderomega}

We now consider when $\lz$ applied on the spin connection coefficients. Recall that 
$$\omegat_{\mu ab}=\gt(\e_a,\dd_\mu \e_b)=\gt_{\alpha\beta} (\e_a)^\alpha \dd_\mu (\e_b)^\beta=\gt_{\alpha\beta} (\e_a)^\alpha \dt_\mu (\e_b)^\beta+\gt_{\alpha\beta}(\e_a)^\alpha \gat_{\mu\nu}^{\ \ \beta}(\e_b)^\nu.$$
Note that $\lz$ does not interact with tetrad indices by our definition. Applying $\lz^I$, one can verify that
\begin{equation*}
    \lz^I(\gt_{\alpha\beta} (\e_a)^\alpha \dt_\mu (\e_b)^\beta)=\sum_{J+K=I}\lz^{J}(\gt_{\a\b}(\e_a)^\a)  \dt_\mu\lzh^{K} (\e_b)^\beta.
\end{equation*}
Also, recall that $\gat_{\mu\nu}^{\ \ \b}=\frac 12\gt^{\b\rho}(\dt_\nu \gt_{\mu\rho}+\dt_\mu \gt_{\nu\rho}-\dt_\rho \gt_{\mu\nu})$. Then similarly
\begin{equation*}
    \lz^I(\gt_{\alpha\beta}(\e_a)^\alpha \gat_{\mu\nu}^{\ \ \beta}(\e_b)^\nu)=\frac 12 \sum_{J+K=I}\lz^J(\gt^{\b\rho}\gt_{\alpha\beta}(\e_a)^\alpha (\e_b)^\nu)(\dt_\nu \lz^K\gt_{\mu\rho}+\dt_\mu\lz^K \gt_{\nu\rho}-\dt_\rho \lz^K\gt_{\mu\nu})
\end{equation*}

For both terms, the factors with $\lz^J$ are clearly bounded, and the factors with $\lz^K$ will determine the decay properties.


This gives the estimate for general components of the spin connection.
Also, if we contract with $\lt^\mu$, the first term becomes $\lt$ derivatives; for the second term, we expand in the null frame just like the proof of Lemma \ref{omeganull}, also using the estimate in Proposition \ref{GS}, to get the following estimate:


\begin{lem}\label{omeganullwithZ}
Suppose the weak decay \eqref{weakdecay} holds. Then
$|\lz^I \omegat_{\mu ab}|\lesssim \sum_{|J|\leq |I|}|\dt \lz^J \gt|$. Moreover, for components in $\lt$ direction we have
$$|\lt^\mu\lz^I (\omegat_{\mu ab})|\lesssim \sum_{|J|\leq |I|}|\dtb \lz^J\gt|+|\dt \lz^J\gt|_{\lc\Tc}.$$
\end{lem}


\subsection{Second-order equation for spinor field}
The main term on the right hand side of the second order equation of $\psi$ is
\begin{equation}\label{eqfptilde}
    \widetilde F^\psi (\omegat,\dt\psi,\psi)= \frac 12\gt^{\mu\nu} (\omegat_{\mu ab}\Sigma^{ab} \dt_\nu \psi + \frac 12 \dd_\mu (\omegat_{\nu ab})\Sigma^{ab} \psi +\frac 1{8} \omegat_{\mu ab}\omegat_{\nu cd} \Sigma^{ab}\Sigma^{cd} \psi).
\end{equation}

We use the calculation in Section \ref{constructionoftetradsection} to estimate the second term here. From Lemma \ref{avoidloss} we know it equals $-\gt^{\delta\rho}R_{\sigma\delta}-\dt_\sigma (\gt^{\mu\nu}\gh_{\mu\nu}^{\ \ \rho})(\e_a)_\rho(\e_b)^\sigma \Sigma^{ab}\psi+(\e_a)_\rho \widetilde{\Box}_{\gt} (\e_b)^\rho \Sigma^{ab} \psi$ plus cubic terms. For $R_{\sigma\delta}$ and $\widetilde{\Box}_{\gt} (\e_b)^\rho$, we use the Einstein equation to replace them, and then we again have cubic terms. The term with $\gh$ also decays well, in view of \eqref{boundforChristoffelhat}.

\begin{lem}\label{secondorderDiracRHS}
We have
\begin{equation}
|\lz^I \widetilde F^\psi|\lesssim \sum_{|J|+|K|\leq |I|} |\dtb \lz^J \psi||\dt \lz^K \gt|+\sum_{|J|+|K|\leq |I|} |\dt \lz^J \psi|(|\dtb \lz^K \gt|+|\dt \lz^K \gt|_{\lc\Tc})
+|R^{cube\ I}|.
\end{equation}

\end{lem}

\begin{proof}
We have already discussed the second term in \eqref{eqfptilde}, and the third term is clearly a cubic term. 
For the first term, modulo a cubic error term, we can consider instead
$$\lz^I (\mh^{\mu\nu}\omegat_{\mu ab}\Sigma^{ab}\dt_\nu \psi)=\sum_{I_1+I_2+I_3=I} (\lz^{I_1} \mh^{\mu\nu})\lz^{I_2} \omegat_{\mu ab} \Sigma^{ab} \dt_\nu\lz^{I_3}\psi.$$
Since $\lz^{I_1}\mh$ is a multiple of $\mh$, we can expand in the null frame, and use Lemma \ref{omeganullwithZ} to get the estimate.
\end{proof}

Using the weak decay \eqref{weakdecay} for the factor with less vector fields, we obtain
\begin{Prop}
We have
\begin{multline}\label{eqsecondorderDiracRHSwithdecaypluggedin}
    |\lz^I \fp|\lesssim \sum_{|J|\leq |I|}\left(\frac \varepsilon{(1+t)^{2-2\delta}}+\frac {M\ln (1+t+r^*)}{(1+t+r^*)^2}\right)|\dt \lz^J \psi|+\varepsilon(1+t)^{-1+\delta} (1+|q^*|)^{-\frac 12}|\dtb \lz^J \psi|\\
    +\varepsilon(1+t)^{-1+\delta}(1+|q^*|)^{-\frac 32}(1+q_+^*)^{-s}(|\dtb \lz^J \hi|+|\dt \lz^J \hi|_{\lc\Tc})\\
    +\varepsilon(1+t)^{-2+\delta}(1+|q^*|)^{-\frac 12}(1+q_+^*)^{-s}|\dt\lz^J \hi|.
\end{multline}
\end{Prop}

\section{Improved decay estimates of the metric and field}

\subsection{Decay estimates from H\"{o}rmander estimate}
\begin{Prop}\label{improvementfromHormanderProp}We have
\begin{equation}\label{eqHormanderforh}
    |\z^I h^1|\leq C\varepsilon (1+t)^{2\delta} (1+t+r^*)^{-1}(1+q_+^*)^{-\gamma},\ \ \ |I|\leq N-3
\end{equation}
and consequently
\begin{equation}\label{eqHormanderfordh}
    (1+t+r^*)|\dtb\z^I h^1|+(1+|q^*|)|\dt\z^I h^1|\leq C\varepsilon (1+t)^{2\delta} (1+t+r^*)^{-1}(1+q_+^*)^{-\gamma},\ |I|\leq N-4
\end{equation}
\end{Prop}

\begin{proof}
The estimates for $q^*\geq 0$ follow directly from the weak decay \eqref{weakdecay}. We now prove the case for $q^*<0$. Consider $\lzh^I\hi_{\mu\nu}=v^I_{\mu\nu}+u^I_{\mu\nu}$, where
\begin{equation*}
    \Box^* v^I_{\mu\nu}=0,\ \ \ v^I_{\mu\nu}|_{t=0}=\lzh^I\hi_{\mu\nu}|_{t=0},\ \ \ \dt_t v^I_{\mu\nu}|_{t=0}=\dt_t\lzh^I \hi_{\mu\nu}|_{t=0},
\end{equation*}
and
\begin{equation*}
    \Box^* u^I_{\mu\nu}=\Box^* \lzh^I \hi_{\mu\nu},\ \ \ \ \ u^I_{\mu\nu}|_{t=0}=\dt_t\lzh^I u_{\mu\nu}|_{t=0}=0.
\end{equation*}
By (\ref{commutatorestimatefordecay}) we have 
\begin{multline*}
    |\Box^*\lzh^I \hti^1|\lesssim |\Box^*\lzh^I \hti^1-\lz^I(\k\widetilde{\Box}_{\gt} \hti^1)|+|\lz^I \widetilde{F}(\gt)(\dt\hi,\dt\hi)|+|\lz^I \widetilde T|+|\widetilde R^{mass\ I}|+|\widetilde{R}^{cov\ I}|\\
    \lesssim \sum_{\substack{|J|+|K|\leq |I|+1 \\|J|\leq |I|}}\left(\frac{M\ln (1+t+r^*)}{(1+t+r^*)^2}+\frac{|\lzh^J\Hi|}{1+|q^*|}\right)|\dt\lzh^K \hti^1|+\sum_{|J|+|K|\leq |I|}|\dt\lzh^J \gt||\dt\lzh^K \gt|\\
    +\sum_{|J|+|K|\leq |I|}|\lz^J \psi||\dt\lz^K \psi|+|\widetilde R^{mass\ I}|+|\widetilde{R}^{cov\ I}|,
\end{multline*}
The error terms are easier to control and we omit the estimate. We have for $|J|+|K|\leq |I|\leq N-3$ that
\begin{multline*}
    \int_{\Sigma_t}|\dt\lzh^J \gt||\dt\lzh^K \gt|\, d\x\lesssim \int_{\Sigma_t} |\dt\lzh^K \hi|^2+|\dt\lzh^K \mt_0|^2\, d\x\\
    \lesssim \int_{\Sigma_t} |\dt\z^K \hi|^2\, d\x+\int_{\Sigma_t} \frac{\varepsilon^2 \ln(1+t+r^*)^2}{(1+t+r^*)^4}\, d\x\lesssim \varepsilon^2 (1+t)^{2\delta},
\end{multline*}
and
\begin{equation*}
    \int_{\Sigma_t}|\lz^J \psi||\dt\lz^K \psi|\, d\x\lesssim \left(\int_{\Sigma_t} |\z^{J_1} \psi|^2 \, dx\right)^{\frac 12} \left(\int_{\Sigma_t} |\dt\z^{J_2} \psi|^2 \, dx\right)^{\frac 12}
    \lesssim \varepsilon^2 (1+t)^{\delta}.
\end{equation*}
The term from the commutator can also be controlled similarly, using Lemma \ref{Hardy} and Lemma \ref{lemmaH1equivh1integral}.
Now applying Proposition \ref{Hormanderprop} to $u^I_{\mu\nu}$ we have 
\begin{equation*}
    |u^I_{\mu\nu}(t,\x)|(1+t+r^*)\lesssim \int_0^t \frac{\varepsilon(1+\tau)^{2\delta}}{1+\tau}\, d\tau\lesssim \varepsilon (1+t)^{2\delta}.
\end{equation*}
Then also applying Lemma \ref{homogeneouslemma} to $v^I_{\mu\nu}$, using \eqref{dhinr} and \eqref{hinr}, we get the result.
\end{proof}

\subsection{Improved decay estimates from generalized wave coordinate condition}
Using the decay we get from Proposition \ref{improvementfromHormanderProp}, we can get the improved decay estimate from our (generalized) wave coordinate condition. 
\begin{Prop}\label{wccdecayprop}
For $|I|\leq N-4$,
\begin{equation}\label{dHLL}
    |\d_{q^*} \lz^I \Hi|_{\lc\Tc}+|\d_{q^*} \slashed{\mathrm{tr}} \lzh^I \Hi| \lesssim \varepsilon (1+t+r^*)^{-2+4\delta} (1+q^*_+)^{-\y} (1+|q^*|)^{-2\delta},
\end{equation}
\begin{equation}\label{HLL}
    |\lz^I \Hi|_{\lc\Tc}+|\slashed{\mathrm{tr}} \lzh^I \Hi|\lesssim \varepsilon (1+t+r^*)^{-1-\y+2\delta} (1+q^*_-)^{\y}.
\end{equation}
\end{Prop}

\begin{proof}
The first estimate follows from (\ref{wccestimate}). Then we integrate along $\d_{q^*}$ to get the second estimate.
\end{proof}

\begin{remark}
Similarly, one can show that the same estimates hold with $\Hi$ replaced by $\hi$.
\end{remark}

\subsection{Improved decay estimates from the \texorpdfstring{$L^\infty$}{Li}-\texorpdfstring{$L^\infty$}{Li} estimate}

\begin{Prop}
For $|I|\leq N-5$, 
\begin{equation}\label{dpsiimproved}
    |\dt\lz^I \psi| \lesssim \varepsilon (1+t+r^*)^{-1}(1+|q^*|)^{-\frac 32+\delta}(1+q^*_+)^{-s},
\end{equation}
\begin{equation}\label{psiimproved}
|\lz^I \psi|\lesssim \varepsilon(1+t+r^*)^{-1}(1+|q^*|)^{-\frac 12+\delta}(1+q^*_+)^{-s}.
\end{equation}
\end{Prop}


\begin{proof}
Let $\wb(q^*)=(1+|q^*|)^{\frac 32-2\delta}(1+q_+^*)^s$. Then by the weak decay $|\lz^I \psi|\lesssim \varepsilon(1+t+r)^{-1+\delta}(1+|q^*|)^{-\frac 12}(1+q_+^*)^{-s}$, we have $||\z^I \psi(\tau,\cdot)\wb||_{L^\infty(D_\tau)}\lesssim\varepsilon (1+\tau)^{-\delta}$, so using Lemma \ref{improveddecaylemma} applied to $(\lhhz)^I\psi$, we get
\begin{multline*}
    (1+t+r)|\dt \lz^I\psi(t,x) \wb(q^*)|\lesssim \varepsilon \sup_{|q^*|/4\leq \tau\leq t} ||(1+\tau)^{-\delta}||_{L^\infty}+\varepsilon\int_{|q^*|/4}^t (1+\tau)^{-1-\delta} \, d\tau \\
    +\sum_{|J|\leq |I|}\int_{|q^*|/4}^t (1+\tau)||\Box^*(\lhhz)^I\psi(\tau,\cdot) \wb||_{L^\infty(D_\tau)} \, d\tau.
\end{multline*}
Using the decay we now have and (\ref{HLL}) for the commutator estimate (\ref{commutatorestimatefordecay}), along with the estimate in Lemma \ref{secondorderDiracRHS}, we obtain
\begin{multline*}
|\Box^*(\lhhz)^I \psi|\lesssim |\Box^*(\lhhz)^I \psi-\lz^I(\widetilde{\Box}_{\k\gt} \psi)|+|\lz^I \fp|\\
\lesssim\sum_{|J|+|K|\leq |I|+1, |J|\leq |I|} \left( \frac{M\ln(1+t+r^*)}{(1+t+r^*)^2}+\frac{|\lz^J \Hi|_{\lc\lc}}{1+|q^*|} +\frac{|\lz^J \Hi|}{1+t+r^*}\right)|\dt \lz^K \psi|\\
+\sum_{|J|+|K|\leq |I|} \left(|\dtb \lz^J \psi||\dt \lz^K \gt|+ |\dt \lz^J \psi|(|\dtb \lz^K \gt|+|\dt \lz^K \gt|_{\lc\Tc})\right)
+|R^{cube\ I}|\\
\lesssim \varepsilon^2 (1+t+r^*)^{-2-\y+3\delta}(1+|q^*|)^{-\frac 52+\y}(1+q_+^*)^{-s}.    
\end{multline*}
Therefore we obtain
$$\int_{|q^*|/4}^t (1+\tau)||\Box^*\lz^I\psi(\tau,\cdot) \wb||_{L^\infty(D_\tau)} \, d\tau\lesssim \varepsilon^2 \int_{|q^*|/4}^t (1+\tau)^{-1-\y+3\delta}\, d\tau\lesssim \varepsilon (1+|q^*|)^{-\y+3\delta}.$$
Now for every terms on the right we get at least an additional decay factor $(1+|q^*|)^{-\delta}$, and the Proposition follows in view of the weight we use. For the second estimate, we again integrate along $\d_{q^*}$ to get it\footnote{For the case $q^*<0$, we need to integrate to $r=0$, along the direction of $\d_{q^*}$, which is different from the case when we integrate from the light cone, since we need to get decay in $q^*$ here. One can also integrate from the light cone, which then gives no decay in $q^*$, but it is still enough for our problem.}.
\end{proof}

\begin{Prop}
For $|I|\leq N-5$, we have

\begin{equation}\label{dh1TTeq}
|\dt \lz^I h^1|_{\Tc\Tc}\lesssim \varepsilon(1+t+r^*)^{-1}(1+|q^*|)^{-1+2\delta}(1+q_+^*)^{-\min(\y,2s)},
\end{equation}
\begin{equation}\label{dh1LUeq}
|\dt \lz^I h^1|_{\lc\Uc}\lesssim \varepsilon(1+t+r^*)^{-1}(1+|q^*|)^{-1+2\delta}(1+q_+^*)^{-\min(\y,2s)},
\end{equation}
\begin{equation}
|\dt \lz^I h^1|\lesssim \varepsilon (1+t+r^*)\ln \left(1+\frac{1+t+r^*}{1+|q^*|}\right)(1+|q^*|)^{-1+2\delta}(1+q_+^*)^{-\min(\y,2s)}.
\end{equation}

\end{Prop}

\begin{proof}
Since $|(\Box^*\lzh^I \hti^1)_{\Ut\widetilde V}|\lesssim |R^{com*I}_{\Ut\widetilde{V}}|+|(\lz^I \widetilde{F})_{\Ut\widetilde V}|+|(\lz^I\widetilde T)_{\Ut\widetilde V}|$, contracting (\ref{ZLmuLnu}) with $\Tt^\mu \Tt^\nu$ and $\lt^\mu \Ut^\nu$ we get 
$$|\lz^I \widetilde F|_{\Tc\Tc}+|\lz^I \widetilde F|_{\lc\Uc}\lesssim R^{tan\ I}+R^{cube\ I}.$$ 
By \eqref{eqEMThigherorderTT} we have $$|\lz^I\widetilde{T}|_{\Tc\Tc}\lesssim |\dtb \lz^J \psi||\lz^K \psi|+|R^{cube\ I}|\lesssim \varepsilon^2(1+t+r^*)^{-3+\delta}(1+|q^*|)^{-1}(1+q_+^*)^{-2s}.$$
The estimate of the commutator is very similar with the proposition above. Hence, using Lemma \ref{improveddecaylemma} with $\wb(q^*)=(1+|q^*|)^{1-3\delta}(1+q_+^*)^{\min(\y,2s)}$ we get the first estimate. The second estimate follows similarly using \eqref{eqTLU} instead.


Now for general components, in view of \eqref{ZLmuLnu}, we estimate the main term
\begin{equation}
\begin{split}
    |\widehat P(\d_{q^*} \lzh^J \hti^1,\d_{q^*}\lzh^K \hti^1)|&\lesssim |\dt \lz^J \hti^1|_{\Tc\Tc}|\dt \lz^K \hti^1|_{\Tc\Tc}+(|\d_{q^*}\lz^J \hti^1|_{\lc\Tc}+|\d_{q^*}\slashed{\mathrm{tr}} \lz^J \hti^1|)|\dt \lz^K \hti^1|\\
    &\ \ \ \ \ \ \ \ \ +|\dt \lz^J \hti^1|(|\d_{q^*}\lz^K \hti^1|_{\lc\Tc}+|\d_{q^*}\slashed{\mathrm{tr}} \lz^K \hti^1|)\\
    &\lesssim \varepsilon^2 (1+t+r)^{-3+3\delta}(1+|q^*|)^{-1-\delta}(1+q_+^*)^{-2\y}+|\dt \lz^J \hti^1|_{\Tc\Tc}|\dt \lz^K \hti^1|_{\Tc\Tc},
\end{split}
\end{equation}
and for the last term we use the estimate of $|\dt \lz^I \hi|_{\Tc\Tc}$, which we have just established, to get
\begin{equation*}
    |\widehat P(\d_{q^*} \lzh^J \hti^1,\d_{q^*}\lzh^K \hti^1)|\lesssim \varepsilon^2 (1+t+r^*)^{-2}(1+|q^*|)^{-2+4\delta}(1+q_+^*)^{-2\min(\gamma,2s)}
\end{equation*}
Also, we have for the field that
\begin{equation}
    |(\lz^I\widetilde T)_{\Ut\widetilde V}|\lesssim \sum_{|J|+|K|\leq |I|} |\lz^J \psi||\dt\lz^K \psi|+|R^{cube\ I}|\lesssim \varepsilon^2(1+t+r^*)^{-2}(1+|q^*|)^{-2+2\delta}(1+q_+^*)^{-2s}.
\end{equation}
The $t^{-2}$ decay rate in light cone direction will give the logarithmic growth when we use Lemma \ref{improveddecaylemma}, and the third estimate follows.
\end{proof}

Integrating along the integral curve of $\d_{q^*}$ we obtain
\begin{Prop}For $|I|\leq N-5$
\begin{equation}\label{h1general}
|\lz^I \hi|\lesssim \varepsilon (1+t+r^*)^{-1}\ln (1+t+r^*)\, (1+|q^*|)^\delta (1+q^*_+)^{-\min(\y,2s)}.
\end{equation}
\end{Prop}

\section{Improved Energy estimate for the system}

\subsection{Einstein's equation}
Commuting the equation multiplied by $\k$ with $\lz^I$, using the definition of the commutator \eqref{RhcomI}, we get \begin{multline*}
    \widetilde{\Box}_{\gt}\lzh^I \hi_{\mu\nu}+\!\k^{-1}\, R_{\hi}^{com\ I}=\k^{-1}\!\sum_{I'+I''=I}\!\lz^{I''}\k\Big(\sum_{J+K=I'}\! \widetilde F_{\mu\nu}(\gt)(\dt \lz^J \hi,\dt \lz^K \hi)+\lz^I\widetilde T_{\mu\nu}\\
    +R_{\mu\nu}^{cube\ I}+R_\mn^{mass \ I}+R_\mn^{cov\ I}\Big).
\end{multline*}
The factor with $\kappa$ clearly behaves good and we can focus on the other factor. 
By Proposition \ref{energyestimatewave} we now have
\begin{multline}\label{metricenergystep1}
    \int_{\Sigma_t}|\dt\lzh^I \hi|^2w\, d\x+\int_0^t \int_{\Sigma_\tau} |\dtb\lzh^I \hi|^2 w'\, d\x d\tau\\
    \leq 8\int_{\Sigma_0}|\dt\lzh^I \hi|^2 w\, d\x+12\int_0^t\int_{\Sigma_\tau}|\widetilde{\Box}_{\gt}\lzh^I \hi||\dt\lzh^I \hi|w\, d\x d\tau\\
    \leq 8\int_{\Sigma_0}|\dt\lzh^I \hi|^2 w\, d\x+12\int_0^t\int_{\Sigma_\tau}\frac{\varepsilon}{1+\tau}|\dt\lzh^I \hi|^2+\varepsilon^{-1}(1+\tau)|\widetilde{\Box}_{\gt}\lzh^I \hi|^2w\, d\x d\tau.
\end{multline}

We need to estimate the last term. From \eqref{ZLmuLnu} we know that the main contributions from the last term are from $\widehat P(\d_{q^*} \lzh^J \hti^1,\d_{q^*}\lzh^K \hti^1)$, the commutator $R_{\hi}^{com\ I}$, and the matter $\lz^I \Tt_\mn$. We have
\begin{multline}
    |\lz^I \widetilde F|\lesssim \sum_{|J|+|K|\leq |I|}|\widehat P(\d_{q^*} \lzh^J \hti^1,\d_{q^*}\lzh^K \hti^1)|\lesssim \sum_{|J|+|K|\leq |I|}|\dt \lz^J \hti^1|(|\d_{q^*}\lz^K \hti^1|_{\lc\Tc}+|\d_{q^*}\slashed{\mathrm{tr}} \lz^K \hti^1|)\\
    +(|\d_{q^*}\lz^J \hti^1|_{\lc\Tc}+|\d_{q^*}\slashed{\mathrm{tr}} \lz^J \hti^1|)|\dt \lz^K \hti^1|+|\dt \lz^J \hti^1|_{\Tc\Tc}|\dt \lz^K \hti^1|_{\Tc\Tc}\\
    \lesssim \sum_{|J|+|K|\leq |I|}\frac{\varepsilon(1+q_+^*)^{-\y}}{(1+t+r^*)^{1-2\delta}(1+|q^*|)}(|\d_{q^*}\lz^J \hti^1|_{\lc\Tc}+|\d_{q^*}\slashed{\mathrm{tr}} \lz^J \hti^1|)\\
    +\frac{\varepsilon(1+q_+^*)^{-\min(\y,2s)}}{(1+t+r^*)(1+q^*)^{1-2\delta}}|\dt\lz^J \hti^1|
\end{multline}
using \eqref{eqHormanderfordh}, \eqref{dHLL} and \eqref{dh1TTeq}. Then
\begin{multline}
    \int_0^t \int_{\Sigma_\tau} \varepsilon^{-1}(1+\tau) |\lz^I \widetilde F|^2 w\, d\x d\tau\lesssim \sum_{|J|\leq |I|}\int_0^t \frac{\varepsilon}{1+\tau}|\dt\lz^J \hi|^2w\, d\x d\tau\\
    +\int_0^t\int_{\Sigma_\tau}\varepsilon (1+\tau)^{-1+4\delta}(1+|q^*|)^{-2}(|\dtb\lz^J\hi|^2+|\d_{q^*}\lz^J \hti^1|^2_{\lc\Tc}+|\d_{q^*}\slashed{\mathrm{tr}} \lz^J \hti^1|^2)w\, d\x d\tau.
\end{multline}

For the commutator, by (\ref{commutatorestimate}), (\ref{HLL}) and the weak decay,
\begin{multline*}
    |\k^{-1}R_{\hi}^{com\ I}|\lesssim \frac{\varepsilon (1+q_+^*)^{-\y}}{(1+|q^*|)^{\frac 12}(1+t+r^*)^{1-\delta}}\sum_{|J|\leq |I|}\left(\frac{|\lz^J\Hi|_{\lc\lc}}{1+|q^*|}+\frac {|\lz^J \Hi|}{1+t+r^*}\right)\\
    +\frac{\varepsilon}{(1+t+r^*)^{1+\y-\delta}}\sum_{|J|\leq |I|}|\dt\lz^J \hi|.
\end{multline*}

We have \begin{multline*}
    \int_0^t \int_{\Sigma_\tau} \varepsilon^{-1}(1+\tau) |\k^{-1} R_{\hi}^{com\ I}|^2 w\, d\x d\tau \lesssim \int_0^t \int_{\Sigma_\tau} \frac\varepsilon{(1+\tau+|q^*|)^{1-2\delta}} (1+|q^*|)^{-3} |\lz^J \Hi|^2_{\lc\lc} w\, d\x d\tau\\
    +\int_0^t \int_{\Sigma_\tau}\frac\varepsilon{(1+\tau+|q^*|)^{3-2\delta}}(1+|q^*|)^{-1}|\lz^J \Hi|^2 w\, d\x d\tau\\
    +\int_0^t \frac\varepsilon{(1+\tau)^{1+2\y-2\delta}}\int_{\Sigma_\tau}|\dt\lz^J\hi|^2 w\, d\x d\tau.
\end{multline*}
For this we use Lemma \ref{Hardy} to see that
\begin{multline}
    \int_0^t \int_{\Sigma_\tau}\frac\varepsilon{(1+\tau+|q^*|)^{3-2\delta}}\frac{|\lz^J \Hi|^2}{1+|q^*|} w\, d\x d\tau\lesssim \int_0^t \int_{\Sigma_\tau}\frac\varepsilon{(1+\tau+|q^*|)^{2-2\delta}}\frac{|\lz^J \Hi|^2}{(1+|q^*|)^2}w\, d\x d\tau\\
    \lesssim \int_0^t \int_{\Sigma_\tau}\frac\varepsilon{(1+\tau+|q^*|)^{2-2\delta}}|\dt\lz^J \Hi|^2 w\, d\x d\tau,
\end{multline}
and
\begin{equation}
    \int_0^t \int_{\Sigma_\tau} \varepsilon\frac{|\lz^J\Hi|^2_{\lc\Tc}}{(1+|q^*|)^2}\frac{(1+|q^*|)^{-2\delta}}{(1+\tau+|q^*|)^{1-2\delta}}\frac{w\, d\x d\tau}{(1+q_-^*)^{2\mu}}\lesssim \int_0^t \int_{\Sigma_\tau} (|\dt_{q^*} \lz^J \hi|_{\lc\Tc}^2+|\dtb \lz^J\hi|^2) w'\, d\x d\tau.
\end{equation}

For the term involing the spinor field, we have
\begin{equation*}
    |\lz^I \widetilde T_{\mu\nu}|\lesssim \varepsilon \frac{(1+|q^*|)^{-\frac 12+\delta}(1+q_+^*)^{-s}}{1+t+r^*}\sum_{|J|\leq |I|}|\dt\lz^J \psi|+\frac{(1+|q^*|)^{-\frac 32+\delta}(1+q_+^*)^{-s}}{1+t+r^*}\sum_{|J|\leq |I|}|\lz^J \psi|.
\end{equation*}
Recall that we require $s>\delta$ and $\gamma<1$. Under this assumption
\begin{multline}
    \int_0^t\int_{\Sigma_t} \varepsilon^{-1}(1+\tau)|\lz^I \Tt_{\mu\nu}|^2 w\, d\x d\tau \lesssim\sum_{|J|\leq |I|}\int_0^t \frac{\varepsilon}{1+\tau}(|\dt\lz^J\psi|^2 w_1+|\lz^J \psi|^2\, d\x d\tau\\
    \lesssim \int_0^t \frac\varepsilon{1+\tau}E_N^1(\tau)\, d\tau+\int_0^t \frac\varepsilon{1+\tau}C_N(\tau)\, d\tau.
\end{multline}




Define $$S_N(t)=\sum_{|I|\leq N}\int_0^t \int_{\Sigma_\tau} |\dtb\lz^I \hi|^2w'\, d\x d\tau.$$Then term involing tangential derivatives can be absorbed by $S_N(t)$.
Combining these estimate together, and summing \eqref{metricenergystep1} up all $|I|\leq N$, we get
\begin{equation}\label{estimateforEN}
    E_N(t)+S_N(t)\leq 8E_N(0)+\int_0^t \frac{C\varepsilon}{1+\tau}(E_N(\tau)+E_N^1(\tau)+C_N(t))\, d\tau.
\end{equation}

\subsection{Second order equation of the field}
Commuting the equation \eqref{zeroorderwaveeqforpsi} multiplied by $\k$ with $\lz^I$, and using the definition \eqref{RpsicomI} we get
\begin{equation*}
    \widetilde{\Box}_{\gt} (\lhhz)^I \psi+\k^{-1}R_\psi^{com\ I}=\k^{-1}\sum_{I'+I''=I}\lz^{I'} \k\, \lz^{I''}\fp+R^{cov\ I}.
\end{equation*}

By Proposition \ref{energyestimatewave} we have
\begin{multline}\label{wave_energyspinorstep1}
    \int_{\Sigma_t} |\dt(\lhhz)^I \psi|^2 w_1\, d\x+\int_0^t \int_{\Sigma_\tau}|\dtb(\lhhz)^I \psi|^2 w_1' \, d\x d\tau\leq 8\int_{\Sigma_0} |\dt(\lhhz)^I \psi|^2w_1\, d\x\\
    +12\int_0^t \int_{\Sigma_\tau} |\widetilde{\Box}_{\gt}(\lhhz)^I \psi||\dt(\lhhz)^I \psi|w_1\, d\x d\tau.
\end{multline}

The main contribution from the last term can be estimated by
\begin{multline*}
    \int_0^t\int_{\Sigma_\tau} (|R_\psi^{com\ I}|+\sum_{|I''|\leq |I|}|(\lhhz)^{I''} \fp|)|\dt(\lhhz)^I \psi|w_1\, d\x dt\\
    \lesssim \sum_{|J|\leq |I|}\int_0^t\int_{\Sigma_\tau} \frac{\varepsilon}{(1+\tau)^{1+2\delta}}|\dt\lz^J\psi|^2 w_1\, d\x d\tau\\
    +\int_0^t\int_{\Sigma_\tau}\varepsilon^{-1}(1+\tau)^{1+2\delta} (|R_\psi^{com\ I}|^2+\sum_{|I''|\leq |I|}|\lz^{I''} \fp|^2)w_1\, d\x d\tau.
\end{multline*}

For the commutator term, using \eqref{commutatorestimate} with $\hi$ replaced by $\psi$, \eqref{HLL} and the weak decay we have
\begin{multline*}
    |R_\psi^{com\ I}|\lesssim \frac{\varepsilon (1+q_+^*)^{-s}}{(1+t+r^*)^{-1+\delta}(1+|q^*|)^{\frac 32}}\sum_{|J|\leq |I|}\left(\frac{|\lz^J\Hi|_{\lc\lc}}{1+|q^*|}+\frac {|\lz^J \Hi|}{1+t+r^*}\right)\\
    +\frac{\varepsilon}{(1+t+r^*)^{1+\y-\delta}}\sum_{|J|\leq |I|}|\dt\lz^J \psi|.
\end{multline*}
Then
\begin{multline*}
    \int_0^t\int_{\Sigma_\tau}\varepsilon^{-1}(1+\tau)^{1+2\delta} |R_\psi^{com\ I}|^2w_1\, d\x d\tau\\
    \lesssim \sum_{|J|\leq |I|}\int_0^t \int_{\Sigma_\tau} \varepsilon\frac 1{(1+\tau+r^*)^{3-4\delta}}\frac {(1+q_+^*)^{-2s}}{(1+|q^*|)^3}|\lz^J\Hi|^2w_1\, d\x d\tau\\
    +\int_0^t \int_{\Sigma_\tau} \varepsilon \frac{(1+|q^*|)^{-3}(1+q_+^*)^{-2s}}{(1+\tau+|q^*|)^{1-4\delta}}\frac{|\lz^J \Hi|_{\lc\lc}^2}{(1+|q^*|)^2} w_1\, d\x d\tau+\frac\varepsilon{(1+\tau+r^*)^{1+2\y-4\delta}}|\dt\lz^J\psi|^2 w_1\, d\x d\tau
\end{multline*}
Recall that $w_1(q^*)=w(q^*)(1+q_+^*)^{1+2s-2\y}$, so by Lemma \ref{Hardy} we have
\begin{multline*}
        \int_0^t\int_{\Sigma_\tau} \varepsilon\frac{(1+|q^*|)^{-3}(1+q_+^*)^{-2s}}{(1+\tau+|q^*|)^{1-4\delta}}\frac{|\lz^J \Hi|_{\lc\lc}^2}{(1+|q^*|)^2}w_1\, d\x d\tau\\
        \lesssim \int_0^t\int_{\Sigma_\tau} \varepsilon \frac{(1+|q^*|)^{-2-2\gamma}}{(1+\tau+|q^*|)^{1-4\delta}}\frac{|\lz^J \Hi|_{\lc\lc}^2}{(1+|q^*|)^2}w\, d\x d\tau \\
        \lesssim \int_0^t \varepsilon(1+\tau)^{-8\delta}\int_{\Sigma_\tau}\frac{|\lz^J \Hi|_{\lc\lc}^2}{(1+|q^*|)^2}\frac{(1+|q^*|)^{-12\delta}}{(1+\tau+|q^*|)^{1-12\delta}}\frac{1}{(1+|q^*|)^{2\mu}}w\, d\x d\tau\\
        \lesssim \int_0^t \varepsilon (1+\tau)^{-8\delta} \int_{\Sigma_\tau} |\d_{q^*} \lz^J \Hi|^2_{\lc\lc}+|\dtb \lz^J \Hi|^2 w'\, d\x d\tau,
\end{multline*}
and
\begin{multline*}
    \int_0^t \int_{\Sigma_\tau} \varepsilon\frac 1{(1+t+r^*)^{3-4\delta}}\frac {(1+q_+^*)^{-2s}}{(1+|q^*|)^3}|\lz^J\Hi|^2w_1\, d\x d\tau\\
    \lesssim \int_0^t\int_{\Sigma_\tau}\varepsilon\frac 1{(1+t+r^*)^{3-4\delta}}\frac {|\lz^J\Hi|^2}{(1+|q^*|)^2}w\, d\x d\tau
    \lesssim\int_0^t\int_{\Sigma_\tau}\varepsilon\frac 1{(1+t+r^*)^{3-4\delta}}|\dt\lz^J\Hi|^2w\, d\x d\tau.
\end{multline*}
For the term involing $\fp$, we can use \eqref{eqsecondorderDiracRHSwithdecaypluggedin} to estimate it, and it is not hard to see terms from there is no worse than what we get above, so we omit estimating them.

\vspace{1ex}
Define $$S_N^1(t)=\sum_{|I|\leq N}\int_0^t \int_{\Sigma_\tau} |\dtb\lz^I \psi|^2w_1'\, d\x d\tau.$$ Using Lemma \ref{wccL2} and summing \eqref{wave_energyspinorstep1} up all $|I|\leq N$, we have
\begin{multline}\label{estimateforEp}
    E^1_N(t)+S^1_N(t)\leq 8E^1_N(0)+\int_0^t \frac{C\varepsilon}{(1+\tau)^{1+2\delta}}(E^1_N(\tau)+E_N(\tau))\, d\tau\\
    +\sum_{|J|\leq N}\int_0^t\int_{\Sigma_\tau} \frac{C\varepsilon}{(1+\tau)^{8\delta}} |\dtb \lz^J \hi|^2 w' \, d\x d\tau
\end{multline}

The last term here can be bounded by $C\varepsilon S_N(t)$. However, we can use the dyadic decomposition in time to get better estimate. For fixed $t\in [0,T]$, pick $K=K(t)\in \mathbb{N}$ such that $t\in [2^{K-1},2^K]$, and define $I_k=[2^{k-1},2^k]$ for $k=1,\cdots,K-1$, $I_0=[0,1]$, $I_K=[2^{K-1},T]$. Then
\begin{multline}\label{dyadicestimate}
    \int_0^t\int_{\Sigma_\tau} C\varepsilon(1+\tau)^{-8\delta} |\dtb \lz^J \hi|^2 w' \, d\x d\tau=\sum_{k=0}^K \int_{I_k} \int_{\Sigma_\tau} C\varepsilon 2^{-8\delta k} |\dtb \lz^J \hi|^2 w' \, d\x d\tau\\
    \leq C\varepsilon\sum_{k=0}^{K-1} 2^{-8\delta k} S_N(2^k)+C\varepsilon 2^{-8\delta K} S_N(t).
\end{multline}
We will use this later in this section.

\subsection{Dirac equation}
We commute the Dirac equation with modified vector fields. By (\ref{Diraccommutatoridentity}) the higher order equation reads
$$\gamma^\mu D_\mu \lhzp^I \psi=F^{com\ I}.$$

By \eqref{dh1LUeq} and \eqref{dm0higherorder} we have the decay $|\dt\gt|_{\lc\Uc}\leq C\varepsilon (1+t)^{-1}$, so using Proposition \ref{energyestimateDirac} we get
\begin{equation}\label{Diracestimatewithvectorfield}
\int_{\Sigma_t} |\lhzp^I \psi|^2 \, d\x
\leq 2\int_{\Sigma_0} |\lhzp^I \psi|^2 \, d\x+4\int_0^t \int_{\Sigma_\tau} \frac {C\varepsilon}{1+\tau}|\lhzp^I\psi|^2+|F^{com\ I}||\lhzp^I \psi|
\, d\x d\tau.
\end{equation}
In view of \eqref{Diraccommutatoridentity}, the last term can be estimated by
\begin{multline*}
    \int_0^t\int_{\Sigma_\tau} \frac{\varepsilon}{1+\tau}|\lhzp^I \psi|^2+\varepsilon^{-1}(1+\tau)|F^{com\ I}|^2\, d\x d\tau \lesssim \int_0^t\int_{\Sigma_\tau} \frac{\varepsilon}{1+\tau}|\lhzp^I \psi|^2\\
    +\sum_{\substack{|J|+|K|\leq |I|\\|K|<|I|}}\varepsilon^{-1} (1+\tau) \left(|\lz^J ((\e_a)^\mu-(\dt_a)^\mu) \d_\mu\lz^K \psi|^2 +|\lz^J ((\e_a)^\mu \omega_{\mu ab})  \lz^K\psi|^2 \right)\, d\x d\tau.
\end{multline*}
We need to estimate the last two integrands. Recall from Proposition \ref{omeganullwithZ} that $|\lz^{J}\omega_{\mu ab}|\lesssim \sum_{|J'|\leq |J|} |\dt\lz^{J'} \gt|$.
When $|K|\leq |J|$, it is equivalent to estimate
$$\sum_{\substack{|J|+|K|\leq |I|\\|K|\leq |I|/2}}\int_0^t\int_{\Sigma_\tau} \varepsilon^{-1} (1+\tau) \left(|\dt\lz^K \psi|^2|\lz^J (\gt-\mh)|^2 +    |\lz^K\psi|^2 |\dt \lz^J \gt|^2\right) \, d\x d\tau.$$

Using \eqref{dpsiimproved}, \eqref{psiimproved}, and $|\lz^J (\mt^0-\mh)|\lesssim \frac{M\chi \ln r}{1+t+r}$ (which implies $|\lz^J (\mt^0-\mh)|\lesssim \frac{M\chi \ln (1+t)}{1+t}$ in view of the support of given by $\chi$ and the fact that $f(x)=\frac {\ln x}x$ is decreasing when $x>3$) by \eqref{dm0higherorder},  
the integral can be controlled by
\begin{multline}
    \int_0^t\int_{\Sigma_\tau} \frac{\varepsilon}{1+\tau}(1+|q^*|)^{-1+2\delta}(1+q_+^*)^{-2s} \left(\frac{|\lz^J \hi|^2}{(1+|q^*|)^{2}}+{|\dt\lz^J \hi|^2}\right)\\ +\frac{M(\ln (1+\tau))^2}{1+\tau}|\dt \lz^K\psi|^2+\frac{M(\ln (1+\tau))^2}{(1+\tau)^3}|\lz^K \psi|^2 \, d\x d\tau\\
    \leq\int_0^t\frac{\varepsilon}{1+\tau}\int_{\Sigma_\tau}  \left(\frac{|\lz^J \hi|^2}{(1+|q^*|)^{3-2\delta}}+|\dt\lz^J \hi|^2\right) w\, d\x d\tau+\int_0^t\int_{\Sigma_\tau} \frac{\varepsilon (\ln (1+\tau))^2}{1+\tau} \, |\dt \lz^K \psi|^2 w_1\, d\x d\tau \\
    \lesssim \int_0^t \frac{\varepsilon}{1+\tau}E_N(\tau)+ \frac{\varepsilon (\ln (1+\tau))^2}{1+\tau} E^1_{[N/2]}(\tau)+\frac{\varepsilon (\ln (1+\tau))^2}{(1+\tau)^3} C_{[N/2]}(\tau)\, d\tau,
\end{multline}
where $[\cdot]$ is the greatest integer function.

When $|K|$ is bigger we need more delicate estimate, since we do not have sharp decay of general components of the metric. In view of the equivalence (\ref{equivalencespinor}), it suffices to estimate 
\begin{equation*}
    \sum_{\substack{|J|+|K|\leq |I|\\|I|/2\leq |K|<|I|}}\int_0^t\int_{\Sigma_\tau} \varepsilon^{-1} (1+\tau) \left(|\lz^J (\gt-\mh)|^2|\dt\lz^K \psi|^2+|\dt \lz^J \gt|^2|\lz^K\psi|^2\right) \, d\x dt
\end{equation*}

For the first term here, we interpolate two estimates. Recall from \eqref{dtoZ} that $$|\dt \lz^K \psi|\leq \min \left(|\dt \lz^K \psi|,\sum_{|K'|\leq |K|+1}\frac{|\lz^{K'} \psi|}{1+|q^*|}\right),$$
so $|\dt \lz^K \psi|^2\leq \sum_{|K'|\leq |K|+1}(1+|q^*|)^{-1}|\dt \lz^K \psi||\lz^{K'} \psi|$. Now since $|\lz^{J}(\gt-\mh)|\lesssim |\lz^{J}\hi|+|\lz^{J}(\mt^0-\mh)|$, using $\eqref{h1general}$ and \eqref{dm0higherorder} we derive
\begin{multline}\label{energyestimatelzgo}
    \int_0^t\int_{\Sigma_\tau} \varepsilon^{-1}(1+\tau)|\lz^J (\gt-\mh)|^2|\dt\lz^K \psi|^2 \, d\x d\tau\\
    \lesssim \sum_{|K|\leq |I|-1}\int_0^t\int_{\Sigma_\tau} \left(\frac {\varepsilon(1+q^*_-)^{2\delta}(\ln (1+\tau))^2}{1+\tau}+\frac{\varepsilon^{-1} M^2(\ln (1+\tau))^2}{1+\tau}\right)|\dt\lz^K \psi|^2 \, d\x d\tau\\
    \lesssim \sum_{\substack{|K'|\leq |I|\\|K|\leq |I|-1}}\int_0^t\int_{\Sigma_\tau} \frac \varepsilon{1+\tau}(\ln (1+\tau))^2 (1+|q^*|)^{-1+2\delta} |\dt\lz^K \psi||\lz^{K'}\psi|\, d\x d\tau\\
    \lesssim \sum_{\substack{|K'|\leq |I|\\|K|\leq |I|-1}}\int_0^t\int_{\Sigma_\tau} \frac \varepsilon{1+\tau}\left((1+|q^*|)^{-2+4\delta} |\lz^{K'}\psi|^2+(\ln (1+\tau))^4|\dt\lz^K \psi|^2 \right)\, d\x d\tau\\
    \lesssim \sum_{|K'|\leq |I|}\int_0^t\int_{\Sigma_\tau} \frac{\varepsilon}{1+\tau}|\lz^{K'}\psi|^2\, d\x d\tau+C\varepsilon^3 (\ln (1+t))^5,
\end{multline}
where we have used the bootstrap assumption on $E_N^1(t)$ for the last inequality.



For the other term, using the decay $|\dt \lz^J h^1|\lesssim \varepsilon (1+t)^{-1}\ln t\, (1+|q^*|)^{-1+\delta}(1+q^*_+)^{-\y}$ and $\eqref{dm0higherorder}$ we have
\begin{equation}\label{energyestimatebulkterm}
        \sum_{\substack{|J|+|K|\leq |I|\\|I|/2\leq |K|<|I|}}\int_0^t\int_{\Sigma_\tau} \varepsilon^{-1}(1+\tau)|\dt \lz^J \gt|^2 |\lz^K\psi|^2\, d\x d\tau\lesssim\sum_{|K|<|I|} \int_0^t\int_{\Sigma_\tau}\frac{\varepsilon(\ln (1+\tau))^2}{1+\tau}|\lz^K \psi|^2\, d\x d\tau.
\end{equation}

Therefore, summing up $|I|\leq k$, and using (\ref{energyestimatelzgo}) and (\ref{energyestimatebulkterm}) we get
\begin{equation}\label{induction}
    C_k(t)\leq 8C_k(0)+\int_0^t \frac {C\varepsilon}{1+\tau}\left(C_k(\tau)+E_k(\tau)\right)+\frac{C\varepsilon(\ln (1+\tau))^2}{1+\tau}C_{k-1}(\tau)\, d\tau+C\varepsilon^3 (\ln (1+t))^5.
\end{equation}

\subsection{Improvement of bootstrap assumptions}
Define $G_k(t)=E_k(t)+E^1_k(t)+C_k(t)$. Then by the assumption on initial data, we have $G_N(0)\leq 3\varepsilon^2$. Adding (\ref{estimateforEN}), (\ref{estimateforEp}) with last term bounded by $C\varepsilon S_k(t)$, and (\ref{estimateofGk}) we get
\begin{equation}\label{estimateofGk}
    G_k(t)+S_k(t)+S^1_k(t)\leq 24\varepsilon^2+\int_0^t \frac{C\varepsilon}{1+\tau}G_k(\tau)+\frac{C\varepsilon(\ln \tau)^2}{1+\tau}G_{k-1}(\tau)\, d\tau+C\varepsilon^3 (\ln (1+t))^5.
\end{equation}

We will use the following form of Gronwall's inequality:
\begin{lem}
For $T>0$ and continuous functions $G,f,g\colon [0,T]\ra \mathbb{R}$ such that $f\geq 0$ and $g$ is non-decreasing, if 
$$G(t)\leq \int_0^t f(\tau)G(\tau)\, d\tau+g(t),$$
for all $t\in [0,T]$, then
$$G(t)\leq g(t)\exp(\int_0^t f(\tau)\, d\tau).$$
\end{lem}
Now we do the induction. We want to prove that $G_k(t)\leq 25\varepsilon^2 (1+t)^{c_k\varepsilon}$, where $c_k$ are constants, satisfying $2c_{i-1}\leq c_i$. For $k=0$ by (\ref{induction}) we have 
$$G_0(t)\leq 24\varepsilon^2+\int \frac{C\varepsilon}{1+t} G_0(\tau)\, d\tau+C\varepsilon^3 (\ln(1+t))^5,$$
so using Gronwall's inequality we get $G_0(t)\leq (24\varepsilon^2+C\varepsilon^3(\ln (1+t))^5) (1+t)^{C\varepsilon}\leq 25\varepsilon^2 (1+t)^{c_0\varepsilon}$ for $\varepsilon$ small enough and $c_0=2C$. Now suppose the bound holds for $k-1$. Then 
$$G_k(t)\leq 24\varepsilon^2+\int_0^t\frac{C\varepsilon}{1+\tau}G_k(\tau)\, d\tau+C\varepsilon^3 (1+t)^{c_{k-1}\varepsilon}\ln (1+t)+C\varepsilon^3 (\ln (1+t))^3.$$
Using again Gronwall's inequality we get
$$G_k(t)\leq (24\varepsilon^2+C\varepsilon^3) (1+t)^{2c_{k-1}\varepsilon}\leq 25\varepsilon^2 (1+t)^{c_k\varepsilon}$$ as desired.

Now we have proved that $G_k(t)\leq 25\varepsilon^2(1+t)^{c_k\varepsilon}$. This improves the boostrap assumptions on $E_N(t)$ and $C_N(t)$, once we let $C_b>25$ in the beginning and $\varepsilon>0$ small so that $c_N\varepsilon<\delta$. Substituting the bound to the right hand side of (\ref{estimateofGk}) we get $S_k(t)\leq C\varepsilon^2 (1+t)^{c_k \varepsilon}$. Using this bound, the improved bound for $E_N(t)$, and (\ref{dyadicestimate}) for \eqref{estimateforEp} we have
\begin{equation*}
    \begin{split}
        E^1_N(t)+S^1_N(t)&\leq 8E^1_N(0)+\int_0^t \frac{C\varepsilon}{(1+\tau)^{1+2\delta}}(E^1_N(\tau)+E_N(\tau))\, d\tau+C\varepsilon^3\sum_{k=0}^{\infty} 2^{-8\delta k}2^{\delta k}\\
        &\leq 8\varepsilon^2+\int_0^t \frac{C\varepsilon}{(1+\tau)^{1+2\delta}}E^1_N(\tau)\, d\tau+C\varepsilon^3,
    \end{split}
\end{equation*}
so using Gronwall's inequality we get $E^1_N(t)\leq 8\varepsilon^2+C\varepsilon^3< C_b\varepsilon^2$ when $\varepsilon$ is small. Therefore we have improved all bootstrap assumptions, and it follows that $T=\infty$.

\paragraph{Conflict of Interest} The author declares that he has no conflict of interest.

\appendix 
\section{Appendix: A choice of the tetrad}\label{appendixtetrad}
We have seen that when dealing with the Dirac equation, especially in the massless case, one often needs to estimate the contraction of the form $\gamma^\mu \d_\mu \psi$ (and higher order version with vector fields), where ${\lb}\psi$ (or $\d_q \psi$) has worse behavior than other derivatives of $\psi$. Therefore, by expanding the contraction in the null frame, we see that the behavior of the term $\gamma_{L} \d_q\psi$ is important. We are interested in two questions: 

(1) We know that $\gamma^\mu-\gb^\mu$ behaves like the perturbation of metric $h=g-m$, which is extensively used in this work, but does the component $(\gamma-\gb)_L$ has better decay? (One can think about the contraction $H^{\a\b}\d_\a\d_\b\phi$, the worst term is like $H_{LL}\d_q\d_q\phi$, and in our context $H_{LL}$ indeed has better behavior, using wave coordinate condition)

(2) Is there any special structure of the matrix $\gb_L$, so that terms involving this matrix have better estimate? In fact, we have already used a version of this statement, when we use the Dirac equation itself to show that the term $\gb_L\d_q\psi$ can be controlled by tangential derivatives.

We will discuss the second question in the next appendix. For the first question, in the presence of $\gamma^\mu$ we know that it depends on the tetrad we choose. We present a different choice of tetrad in this appendix, which might improve some estimates in other situations. We will show that we are able to control the $L$ component of the perturbation of Gamma matrices $(\gamma-\gb)_{L}$ without using the $\lb\lb$ component of the metric, which is expected to be the worst component. From here we go back to the notation in the original coordinates $(\mathbb R^4,g)$.

\begin{lem}\label{Zonhomogeneouscoefficients}
On the region where $r>t/2$ and $t>1$, We have $|\d^I Z^J U^\mu|\lesssim (1+t+r)^{-|I|}$, where $U\in \{L,\lb,S_1,S_2\}$.
\end{lem}

\begin{proof}
Consider a derivative $\d_\a$ operating on $U^\mu$. Notice that $U^\mu$ does not change along the radial direction $\d_r$ and time direction $\d_t$, so the relavant part is the angular derivative. Therefore we have $|x|\d_i U^\mu(t,x)=\d_i U^\mu(t,x/|x|)$, and writing all vector fields explicitly we get the estimate.
\end{proof}

\begin{Prop}\label{existenceofsuchtetrad}
Suppose that $g_\mn=m_\mn+h_\mn$ satisfies $|Z^I h|\lesssim (1+t+r)^{-\frac 12}$ for $|I|\leq N/2$. Then there exists a choice of tetrad $\{e_a\}$ satisfying the following properties, with $\chi=\chi(\frac r{1+t})$ defined in \eqref{h0}:
\begin{enumerate}
    \item $|(e_a)^\mu-(\d_a)^\mu|\lesssim |h|$;
    \item $|\d_\mu (e_a)^\nu|\lesssim |\d_\mu h|+\chi\frac 1r|h|$;
    \item $|Z^I ((e_a)^\mu-(\d_a)^\mu)|\lesssim \sum_{|J| \leq |I|} |Z^J h|$, and $|\d_\mu Z^I (e_a)^\nu|\lesssim \sum_{|J|\leq|I|}|\d_\mu Z^J h|+\chi\frac 1r |Z^J h|$ plus quadratic terms;
    \item On the region where $\chi(\frac r{1+t})=1$, i.e.\ $r\geq\frac 34 (t+1)$, we have $|m(e_a-\d_a,L)|+|m(e_a-\d_a,S_i)|\lesssim |h|_{\mathcal{T}\mathcal{U}}+O(h^2)$. 
\end{enumerate}
\end{Prop}

\begin{proof}
Let $a$ be a real number such that $g(L+a\lb,L+a\lb)=0$. This is equivalent to the equation
$$h_{\lb\lb}a^2+2g_{L\lb}a+h_{LL}=0.$$
Note that $g_{L\lb}=-2+h_{L\lb}$ so its absolute value is significantly bigger than other coefficients. If $h_{\lb\lb}=0$, then we have the unique solution $a=-\frac{h_{LL}}{2g_{L\lb}}$; if $h_{\lb\lb}\neq 0$, then the discriminant is greater than zero, and $a=\frac{-g_{L\lb}\pm \sqrt{g^2_{L\lb}-h_{LL}h_{\lb\lb}}}{h_{\lb\lb}}$. We take the root whose absolute value is smaller so that $a$ really means a perturbation, i.e.
$$a=\frac{g_{L\lb}}{h_{\lb\lb}}(-1+\sqrt{1-\frac{h_{LL}h_{\lb\lb}}{g^2_{L\lb}}})\approx \frac{g_{L\lb}}{h_{\lb\lb}}(-\frac 12 \frac{h_{LL}h_{\lb\lb}}{g^2_{L\lb}})+O(h^3)=-\frac{h_{LL}}{2g_{L\lb}}+O(h^2)\approx \frac 14 h_{LL}+O(h^2).$$
The limit here as $h_{\lb\lb}\ra 0$ coincides with the value in the first case.

Similarly we can find $b$ such that $g(\lb+bL,\lb+bL)=0$ with $|b|<<1$. We also have $g(L+a\lb,\lb+bL)=(1+ab)g_{L\lb}+ah_{\lb\lb}+bh_{LL}$, which differs from $-2$ by a quadratic term. Let 

$$\textstyle\lt=\sqrt{-\frac 2{(1+ab)g_{L\lb}+ah_{\lb\lb}+bh_{LL}}}(L+a\lb),\ \ \ \ltb=\sqrt{-\frac 2{(1+ab)g_{L\lb}+ah_{\lb\lb}+bh_{LL}}}(\lb+bL).$$
Then $\lt,\ltb$ are $g$-null vectors and $g(\lt,\ltb)=-2$. 

Let $$\textstyle \Tt=\frac 12(\lt+\ltb)=\sqrt{-\frac 2{(1+ab)g_{L\lb}+ah_{\lb\lb}+bh_{LL}}}(\d_t+\frac a2\lb+\frac b2 L),$$ 
$$\textstyle \Rt=\frac 12(\lt-\ltb)=\sqrt{-\frac 2{(1+ab)g_{L\lb}+ah_{\lb\lb}+bh_{LL}}}(\d_r+\frac a2 \lb-\frac b2 L).$$ Then we have 
\begin{equation}\label{innerproductofTandRwithL}
    m(\Tt,L)=-1-a+O(h^2),\ \ \ m(\Rt,L)=1-a+O(h^2),
\end{equation}
and we note that $a=\frac 14 h_{LL}+O(h^2)$.

We now define $\et_0=\d_t+\chi(\frac{r}{t+1})(\Tt-\d_t)$, $\et_i=\d_i+\omega_i\chi(\frac{r}{t+1})(\Rt-\d_r)$, where $\chi(s)$ is a smooth cutoff function with $\chi(s)=1$ when $s\leq 1/2$, and $\chi(s)=0$ when $s\geq 3/4$. Then their inner products in $g$ read
\begin{equation}\label{innerproductintermediateframe}
\begin{split}
g(\et_0,\et_0)&=1+h_{00}+2g(\d_t,\chi(\Tt-\d_t))+O(h^2),\\ g(\et_0,\et_i)&=h_{0i}+g(\chi(\Tt-\d_t),\d_i)+\omega_ig(\d_t,\chi(\Rt-\d_r))+O(h^2),\\ g(\et_i,\et_j)&=\delta_{ij}+h_{ij}+g(\d_i,\omega_j \chi(R-\d_r))+g(\omega_i(R-\d_r),\d_j)+O(h^2).
\end{split}
\end{equation}

The set $\{\et_a\}$ is not an orthonormal frame yet, and we now do a Gram-Schmidt orthogonaliztion to this set, just like the process in Lemma \ref{GS}. This will result in a modification comparable with $h_{UV}$, possibly with coefficients $\chi(\frac r{t+1})$, $\omega_i$ and components of $S_1,S_2$ (we call these coeffecients of power 0), plus quadratic terms. We denote the orthonormal frame we get by $\{e_a\}$. Then clearly we have $|e_a^\mu-(\d_a)^\mu|\lesssim |e_a^\mu-\et_a^\mu|+|(\et_a)^\mu-(\d_a)^\mu|\lesssim |h|+O(h^2)\lesssim |h|$, where we use the decay assumption for the last inequality.

For derivatives we have $$\d_\mu (e_a)^\nu=\d_\mu ((e_a)-(\et_a))^\nu+\d_\mu (\et_a)^\nu.$$ For the second term, in view of the definition of $\et_a$, the derivative may fall on $\chi(\frac r{t+1})$ or $\omega_i$, which leads to $\frac 1{1+t}\chi'$ and $\frac 1r \chi$ respectively, or the vector field $\Tt$ and $\Rt$, i.e.\ fall on terms like $h_{UV}$, modulo quadratic terms. So also using $\d_\mu U^\nu\approx \frac 1r$ when $r>t/10$ from Lemma \ref{Zonhomogeneouscoefficients}, we have the second term controlled by $|\d_\mu h|+|\frac \chi r h|+O(h\cdot \d h)$. 
For the first term, the analysis are almost same using those expressions of inner products (\ref{innerproductintermediateframe}). 

The case for applying vector fields also follows, once we recall that $|\d^I Z^J U^\nu|\lesssim r^{-|I|}$ on the support of $\chi$, by Lemma \ref{Zonhomogeneouscoefficients}. Note that when the partial derivative $\d_\mu$ falls on $h$, the direction of the derivative is preserved. Therefore the first three statement holds.


\textbf{Moreover}, in the region where $\chi(\frac r{t+1})=1$, we have $\et_0=\Tt$, $\et_i=\db_i+\omega_i \Rt$, so
\begin{equation}
    \begin{split}
        g(\et_0,\et_0)&=1, \\
        g(\et_0,\et_i)&=h(\Tt,\db_i)\approx h(\d_t,\db_i)+O(h^2),\\
        g(\et_i,\et_j)&=\delta_{ij}+h(\db_i,\db_j)+\omega_jh(\db_i,\Rt)+\omega_i h(\db_j,\Rt)\\
        &\approx h(\db_i,\db_j)+\omega_jh(\db_i,\d_r)+\omega_i h(\db_j,\d_r)+O(h^2).
    \end{split}
\end{equation}
In view of the expression after the Gram-Schmidt process (Lemma \ref{GS}), we know that $e_a-\et_a$, hence $m(e_a-\et_a,L)$, equals a combination of $SU$ components of $h$ ($S\in \mathcal S$, $U\in \mathcal U$), with product of coefficients of order 0, plus terms of type $O(h^2)$.

Also we have $m(\et_0-\d_0,L)=-a+O(h^2)\approx -\frac 14 h_{LL}+O(h^2)$, and $m(\et_i-\d_i,L)=-\omega_i a+O(h^2)\approx -\frac 14\omega_i h_{LL}+O(h^2)$. Therefore $$|m(e_a-\d_a,L)|\leq |m(e_a-\et_a,L)|+|m(\et_a-\d_a,L)|\leq \sum_{T\in \mathcal T,\, U\in \mathcal U} |h|_{\mathcal T \mathcal U}+|h|^2,$$ 
and similar estimates holds if we change $L$ to $S_1,S_2$. Hence the last estimate follows.
\end{proof}

\begin{remark}
From this we can prove the control of $\gamma_L-\gb_L$ 
when $\frac r{t+1}>\frac 34$:
\begin{equation}
    |\gamma_L-\gb_L|\lesssim |h|_{\mathcal T \mathcal U}+|h|^2.
\end{equation}
\end{remark}





\section{Appendix: A spacetime integral estimate from Dirac equation}\label{bulktermappendix}
We discuss the second question, and we will see that for Dirac equation (flat or curved one with the metric satisfying the assumptions below), one can derive the estimate of a spacetime integral term involing $\gb_L$. This can be viewed as an analogue of the spacetime integral we get from the weighted energy estimate of the wave equation in Proposition \ref{energyestimatewave}.

Consider the Dirac equation in a curved background spacetime $(\mathbb{R}^4,g)$:
$$\gamma^\mu D_\mu \psi+im\psi=F,$$
where the metric $g$ satisfies the decay
\begin{equation}\label{eqmetricassumptionappendix}
    (1+|q|)^{-1}|h^1|+|\d h^1|\leq C\varepsilon (1+t)^{-\frac 34}(1+|q|)^{-\frac 12},\ \ \ |\db h^1|\leq C\varepsilon (1+t)^{-\frac 54}
\end{equation}
\begin{equation}\label{eqmetricassumptionTUappendix}
    |h^1|_{\mathcal T \mathcal U}\leq C\varepsilon (1+t)^{-1}(1+|q|)^{\frac 12}
\end{equation}
with $g_\mn=m_\mn+h^0_\mn+h^1_\mn$, and $h_0=\chi(\frac r{1+t})\frac Mr\delta_\mn$. We note that these bounds hold for the solution of Einstein vacuum equation in wave coordinates (and for some coupled system, e.g.\ massless Einstein-Maxwell-Klein-Gordon \cite{KLMKG21}). This gives $|\gamma^\mu-\gb^\mu|\lesssim C\varepsilon (1+t)^{-\frac 34}(1+|q|)^{\frac 12}$ if we choose the tetrad as in Proposition \ref{GS}, or as in Appendix \ref{appendixtetrad}.

We consider a weight function $w=w(q)$ with $w'(q)\leq 0$, and $w'(q)\lesssim (1+|q|)^{-1}w(q)$. Again by Corollary \ref{spinorleibnizcor} we have
$\dd_\mu (\pb \gamma^\mu \psi)=\pb F+\overline{F}\psi$, which implies in our coordinates $\d_\mu(\pb \gamma^\mu \psi)=\pb F+\overline{F}\psi-\Gamma_{\mu\rho}^{\ \ \,\mu} \pb\gamma^\rho \psi$. Then adding the weight to the estimate we have $$\d_\mu( w(q)\pb \gamma^\mu \psi)=w(q) (\pb F+\overline{F}\psi)+w'(q)\pb\gamma_L \psi-w(q) \Gamma_{\mu\rho}^{\ \ \,\mu} \pb\gamma^\rho \psi,$$
where $\gamma_L:=m_{\mu\nu}L^\mu\gamma^\nu$.

We put the center value (which corresponds to $\gb_L$) of the second term on the right hand side to the left, and integrate on the spacetime region between two time slices to get
\begin{equation}\label{identitywithbulkterm}
   \begin{split}
       \int_{\Sigma_{t}} w&(q)\pb \gamma^0 \psi\, dx +\int_0^{t} \int_{\Sigma_\tau}  (-w'(q)) \pb \bar{\gamma}_L \psi\, dx d\tau=\int_{\Sigma_0} w(q)\pb \gamma^0 \psi \, dx \\
    & + \int_0^t \int_{\Sigma_\tau} w(q)(\pb F+\overline{F}\psi-\Gamma_{\mu\rho}^{\ \ \,\mu} \pb\gamma^\rho \psi)+w'(q) \pb (\gamma^\mu-\bar{\gamma}^\mu)L_\mu \psi \, dx d\tau.
    \end{split} 
\end{equation}

We now show that the spacetime term on the left hand side is nonnegative. Recall that we have $w'(q)\geq 0$.

\begin{lem}
The matrix $\bar{\gamma}^0 \bar{\gamma}_L=(\gb_L)^\dagger \gb^0$ is semi-negative definite at every point where $L$ is defined.
\end{lem}

\begin{proof}
Recall that $L=\d_t+\d_r$, so $\gb_L=\gb^\mu L_\mu=-\gb^0+\omega_i \gb^i$ we have
\begin{equation}
\begin{split}    
\gb^{0}\gb_L&=\begin{pmatrix}
1 & 0 & 0 & 0 \\
0 & 1 & 0 & 0 \\
0 & 0 & -1 & 0 \\
0 & 0 & 0 & -1
\end{pmatrix}
\begin{pmatrix}
-1 & 0 & \omega_3 & \omega_1-i\omega_2 \\
0 & -1 & \omega_1+i\omega_2 & -\omega_3 \\
-\omega_3 & -\omega_1+i\omega_2 & 1 & 0 \\
-\omega_1-i\omega_2 & \omega_3 & 0 & 1
\end{pmatrix},\\
&=\begin{pmatrix}
-1 & 0 & \omega_3 & \omega_1-i\omega_2 \\
0 & -1 & \omega_1+i\omega_2 & -\omega_3 \\
\omega_3 & \omega_1-i\omega_2 & -1 & 0 \\
\omega_1+i\omega_2 & -\omega_3 & 0 & -1
\end{pmatrix},
\end{split}
\end{equation}
and it is straightforward to verify that $\bar{\gamma}^0 \bar{\gamma}_L=(\gb_L)^\dagger \gb^0$. It remains to compute all possible principal minors by linear algebra (\cite{prussing1986principal}, \cite[Page 566]{semidefinitebook}). We do not need to compute the $4\times 4$ determinant since we know it is zero from the following observation:
$$\gb_L^2=(\gb^\mu L_\mu) (\gb^\nu L_\nu)=\gb^\mu \gb^\nu L_\mu L_\nu=-m^{\mu\nu}L_\mu L_\nu=0.$$

For remaining principal minors, we just compute one of them (Row 1,3,4) since others are similar:
$$\begin{vmatrix}
        -1  & \omega_3 & \omega_1-i\omega_2 \\
        \omega_3  & -1 & 0 \\
        \omega_1+i\omega_2  & 0 & -1
\end{vmatrix}=-1(1-\omega_3^2)+(\omega_1+i\omega_2)(\omega_1-i\omega_2)=-1+\omega_1^2+\omega_2^2+\omega_3^2=0.$$
\end{proof}

From this we know that $-\pb \gb_L \psi=\pb(-\gb^0\gb_L)\psi\geq 0$ for every spinor field $\psi$ which takes value in $\mathbb{C}^4$.

\begin{thm}
Suppose $\psi$ is a solution of the equation $\gamma^\mu D_\mu\psi+im\psi=F$, and vanishes at infinity. Assume the background metric $g$ satisfies the assumption \eqref{eqmetricassumptionappendix}, \eqref{eqmetricassumptionTUappendix}, then 
\begin{multline}
\int_{\Sigma_{t}} |\psi|^2 w\, d\x+\int_{0}^{t} \int_{\Sigma_\tau} (-\pb\bar{\gamma}_L \psi) w' \, d\x d\tau \leq 2\int_{\Sigma_0} |\psi|^2 w\, d\x\\
+2\int_{0}^{t} \int_{\Sigma_\tau} \left(\frac{C\varepsilon}{1+\tau} |\psi|^2  +|\pb F|+|\overline{F}\psi|\right) w \, d\x d\tau
\end{multline}
\end{thm}

\begin{proof}
We use \eqref{identitywithbulkterm}. Notice that $$\Gamma_{\mu\rho}^{\ \ \,\mu}=\frac 12 g^{\mu\nu} (\d_\mu g_{\rho\nu}+\d_{\rho}g_{\mu\nu}-\d_\nu g_{\mu\rho})=\frac 12 g^{\mu\nu}\d_\rho g_{\mu\nu},$$
so we have the estimate $|\Gamma_{\mu\rho}^{\ \ \,\mu} \pb\gamma^\rho \psi|\lesssim |\d g||\pb \bar{\gamma}_L \psi|+ |\db g||\psi|^2 + |\gamma^\mu-\bar{\gamma}^\mu||\d g||\psi|^2$. We can absorb the first term on the right hand side using the spacetime integral we have just got on the left hand side. For the last term here, we pick the tetrad in Appendix \ref{appendixtetrad} and we have one extra factor of decay.

We still need to deal with the last integrand in \eqref{identitywithbulkterm}. In view of the control we derived in Appendix \ref{appendixtetrad}, we have $|(\gamma^\mu-\gb^\mu)L_\mu|\lesssim |h|_{\mathcal T \mathcal U}+O(h^2)\lesssim \varepsilon (1+t)^{-1}(1+|q|)^{\frac 12}$, so using also $w'(q)\lesssim (1+|q|)^{-1}w(q)$ we get the estimate.
\end{proof}

\bibliographystyle{abbrv}
\bibliography{Ref}

\small\textsc{Department of Mathematics, Johns Hopkins University, Baltimore, MD 21218, USA}

\small\textit{Email address:} \texttt{{xchen165@jhu.edu}}

\end{document}